\documentclass[runningheads]{llncs}

\usepackage[T1]{fontenc}
\usepackage{graphicx}
\usepackage{amsmath}
\usepackage{amssymb}
\usepackage{mathtools}
\usepackage{subcaption}
\usepackage[natbib=true, style=authoryear]{biblatex}

\begin{document}
\title{Custom Keep-Alive Cache Policies}
%
%

\author{Sushirdeep Narayana\inst{1,2}\ \and
Ian A. Kash\inst{2}}
\authorrunning{}
%
\institute{Wesleyan University, Middletown, CT, USA \\
\email{snarayana@wesleyan.edu, snaray25@uic.edu}\\
\and
University of Illinois at Chicago, Chicago, IL, USA\\
\email{iankash@uic.edu}}
\maketitle              
\begin{abstract}
 We study the market design of keep-alive caching policies applicable in serverless computing. Prior work has assumed that the cost of a cache miss (cold start) is uniform across all customer applications. However, the cost of a cache miss depends on the customer's application. We investigate the market design where the customers submit a bid for their cost of a cache miss. We design a cache allocation policy based on online learning from a mixture of fixed allocation experts. We show that our custom cache allocation policy is asymptotically efficient and monotonically non-increasing with respect to the submitted bid. We examine two ways of charging customers to achieve good incentives. In the first payment scheme the customers are charged based on Myerson's theory, whereas in the second payment scheme the customers are charged their externality. We show via a mix of simulations and theory that both schemes have desirable revenue and incentive properties.
 
\end{abstract}

\section{Introduction}

In serverless computing (also known as Function-as-a-Service, or FaaS) customers provide code to be invoked in response to a trigger, such as an HTTP request or a timer.  A cloud provider takes responsibility for provisioning the necessary hardware and software resources to run the code, with the customer typically only charged based on the time those resources are being used to run the code. This leaves the cloud provider with a challenging problem about when to keep these resources provisioned.  If resources are already provisioned when a function is invoked, the response has a very low latency because the code can start running immediately (a ``warm start'').  If the resources are not provisioned, there will be a delay while provisioning occurs (a ``cold start''), which may be problematic for some applications.  Warm starts are always preferable, but there is an opportunity cost in that the provisioned resources could have been used for other purposes.

This creates a form of caching problem where the cloud provider needs to decide how long to keep the resources provisioned once the function invocation completes (the ``keep-alive window'').  In contrast to traditional caches which have a fixed size, cloud providers can flexibly adjust how much of their overall pool of resources they are allocating to serverless computing. 
Thus, prior work has modeled this as having a cost-per-unit-time to keep a function cached~\citep{narayana2023keep}.  Prior work has developed caching policies designed to optimize the tradeoff between the cost of cold starts and the cost of caching using historical data~\citep{shahrad2020serverless,narayana2023keep}.  It has also developed techniques to reduce the latency of cold starts, thus reducing (but not eliminating) the need for caching~\citep{oakes2018sock,mohan2019agile,roy2022icebreaker}.

However, the cost of a cold start varies from customer to customer and application to application depending on the need for low latency.  Recognizing that cold starts are extremely costly in some cases, AWS Lambda offers customers the option to have their function permanently cached in exchange for paying for the resources the entire time, not just when the function code is running.  Azure Functions recently launched a similar feature, suggesting customizing cache policies in response to customer costs is an important practical consideration.  However, the current market design provides only a limited opportunity for this.  Such permanent caching may not be optimal for all customers.  Furthermore, whether it is optimal depends not only on the customer's cost for a cold start (which seems reasonable to know) but also on the stochastic process governing function invocations (about which the customer may have much less understanding).

We examine an alternative market design where customers report a single number representing their cost each time one of their function invocations experiences a cold start.  Based on this, an online learning algorithm can automatically customize the length of the keep-alive window, achieving approximately high efficiency as if the optimal length of the keep-alive window (among the options available) had been selected in advance. Prior work has developed sophisticated online learning algorithms for more general settings which can guarantee monotonicity of allocations~\citep{babaioff2015truthful} or approximate VCG allocations and payments~\citep{kandasamy2023vcg}.  In contrast, we show that our setting has nice structure such that the classic exponential weights algorithm is inherently monotone.  The structure we exploit (single parameter utilities combined with an allocation approach also controlled by a single parameter), seems natural in other settings as well, suggesting our approach may be useful beyond the specific market we study.

We study two different ways of charging customers to incentivize truthful reporting of cold-start costs and compensate the cloud provider for the costs of providing the caching.  The first is the standard \citet{myerson1981optimal} approach of charging payments based on the calculation of an integral of the allocation.  This is the integral of a parameterized softmax, and does not appear to have a closed form in general, although it does in special cases and can be evaluated numerically.  This guarantees that the resulting mechanism is incentive compatible, but the resulting payments are somewhat complex, which  may make it less attractive to customers.  In general, it also doesn't provide strong guarantees to the cloud provider about the relationship between the revenue received and the cost the cloud provider's incurs. However, we provide evidence that revenue and costs are reasonably close for two natural stochastic processes and on a trace from Mircosoft Azure.

The second, motivated by the fact that our algorithm is trying to approximate the efficient outcome, is to charge the customer their externality, as in VCG.  Here, the payments correspond to charging customers exactly the caching cost incurred by the provider, so we have a guarantee that the revenue exactly matches the provider's costs. Furthermore, payments are easy to explain to customers and consistent with the current market design of allowing permanent caching in exchange for full payment of the resulting costs.  In general, such approximate VCG payments need not be incentive compatible, but again we provide empirical evidence that they generally have good incentive properties in our setting.

\section{Related Work}
\label{related_work}

 \citet{shahrad2020serverless} describe a system that uses historical data to optimize keep-alive windows and also consider adding items back to the cache at a later time if invocations are predictably periodic.  \citet{narayana2023keep} show that the class of keep-alive policies we study is optimal for a class of arrival processes including Poisson and Hawkes and that Hawkes processes are a relevant model for many existing serverless customers.  However, both rely on historical data and optimize only based on the arrival process, not the customer's personal cold start cost.  \citet{lin2020serverless} discuss the desirability of customized policies, but do not provide a market design to implement them.

\citet{babaioff2015truthful} show how to make bandit algorithms with monotone allocation rules incentive compatible and design one such monotone bandit algorithm for a stochastic ad auction setting.  Motivated by providing service to cloud computing customers and advertising auctions, \citet{kandasamy2023vcg} develop a sophisticated bandit algorithm that learns to approximate VCG allocations and payments.  In contrast, we work in the classic online learning from experts setting rather than the bandit setting as once the invocation occurs we can determine whether every possible caching policy would have a cold start and how long it kept the function in the cache.  Furthermore, we exploit known structure in the cost function while their approaches assume it to be general.  Thus, we are able to use a standard algorithm such as multiplicative weights.  We believe our work explores an interesting new point in the design space combining online learning algorithms with information elicitation.

There is also literature that focuses on the use of bandit algorithms to set posted prices. \citet{babaioff2015dynamic} design revenue-maximizing online posted-price mechanisms where the seller has limited supply based on the UCB1 algorithm where the index of an arm (price) is defined according to the estimated expected total payoff from this arm given the known constraints. \citet{roth2020multidimensional} investigate the problem of dynamic pricing of distinct types of indivisible goods, when buyers having unit-demand buyers drawn independently from an unknown distribution arrive in an online manner.
Our work shares the approach of exploiting known structure of the loss function, but has different goals (efficiency and incentives rather than revenue maximization).

There has also been work on incentivizing exploration in bandits without payments using a Bayesian framework. \citet{mansour2020bayesian} model a social planner's recommendation policy under incentive compatibility constraints induced by the agent's Bayesian priors.
Our work shares the focus on choosing exploration actions on behalf of an agent, but uses payments to achieve good incentives.

\section{Preliminaries}
\label{preliminaries_model}

First, we review some background on serverless computing to better understand our market design challenge. 
In traditional caches, an object is evicted from the cache when space is required, so the key decision is \emph{what} to keep in the cache. In contrast, the decision in keep-alive caches is \emph{when} and for \emph{how long} to keep the object in the cache.  Keep-alive caches essentially assume the potential size of the cache is infinite because in cloud contexts, such as serverless computing, hardware resources can be freely added or removed from the cache.
The reason we don't simply want to cache everything in keep-alive cache systems is that doing so ties up resources that could instead be used for other purposes.  Thus, there is an (opportunity) cost based on the duration of keeping an application image in the cache.  There is also the traditional caching cost of a cache miss, which in the context of serverless computing corresponds to a cold start.  The keep-alive cache policy thus trades off the cost of caching a user's application against the cost of cold starts. 


Let $\mathcal{H}_{i} = \{t_{1}, \; t_{2}, \; \cdots, \; t_{i}\}$ denote the history of $i$ previous {\em arrivals} or {\em requests} from the user's application.  Here, $t_j$ denotes the arrival time of the $j$-th request. Let $x_{i} = t_{i} - t_{i-1}$ denote the \textit{i}-th inter-arrival time. We consider the simplest form of a keep-alive cache policy, known as a \textit{Time-To-Live} (TTL) policy. This has the form of a time window after the most recent arrival during which the customer's application image is kept in the cache (each subsequent arrival resets the TTL). 
We denote the length of such a time window by $\tau$, and when we want to specify parameters on which a particular window length depends, we put them in the subscript (e.g. $\tau_{\mathcal{H}_{i}}$ for dependence on history).

An application encounters a warm start if its arrival (i.e.~the invocation of the function) occurs when the keep-alive window is active, analogous to a cache hit ($x_i \leq \tau$). An application has a cold start when the keep-alive window is not active at the time of arrival, analogous to a cache miss ($x_i > \tau$). 

Prior work has sought to identify an optimal choice of $\tau $ for an application either theoretically based on known parameters of the arrival process or empirically based on historical arrival data~\cite{shahrad2020serverless,narayana2023keep}.
In contrast, we seek to (a) adapt a (randomized) window length automatically using online learning techniques and (b) use pricing to elicit a customer's cost for cold starts and customize the window length accordingly.




\begin{definition}
A {\em custom keep-alive cache policy}  takes the submitted report $\hat{\theta}$ of the cost of a cold start from the user as an input, and after the $i$-th arrival of the user's application outputs a keep-alive window length of $\tau_{i+1} = \tau_{\hat{\theta},\mathcal{H}_{i}}$ for the $(i+1)$-st inter-arrival, while charging a payment $p_i=p_{\hat{\theta}, \mathcal{H}_{i}}$ for the $i$th arrival. In our approach $\tau_{i+1}$ is randomized, but as is standard $p_i$ can be deterministic or randomized, as preferred.
\end{definition}

This formulation assumes the user makes a single report of their cost for a cold start when first deploying their application and that the cost does not vary over time so can be applied to all future arrivals.
Given $\tau_i$ and $x_i = t_i - t_{i-1}$ we can now define the loss experienced by both agents.  The customer experiences a cold start if the window length is shorter than the inter-arrival.  We capture this with the following indicator.
\begin{equation*}
 cs(\tau_i, x_{i}) = \begin{cases}
& 0, \quad \text{if} \quad x_i \leq \tau_i\\
& 1, \quad \text{if} \quad x_i > \tau_i
\end{cases}       
\end{equation*}
This results in a loss to the customer of $\theta \cdot cs(\tau_i,x_i)$, which captures the loss in the quality of service experienced by the customer due to the latency delay during a cold start.  This loss depends on the customer's true type $\theta$, which is not known to the cloud provider.

The cloud provider receives a loss based on the resource (primarily memory) ``wasted'' while caching as they are not in active use but are unavailable for other uses. 
Following~\citet{narayana2023keep}, we assume that the cost of keeping an application in cache is a fixed constant $c_p$ per unit time.  The length of time this memory is held is the entire inter-arrival for a warm start, but limited by the keep-alive window for a cold start.
\begin{equation*}
 wm(\tau_i, x_{i}) = \begin{cases}
& c_{p} \cdot x_i, \quad \text{if} \quad x_i \leq \tau_i\\
& c_{p} \cdot \tau_i, \quad \text{if} \quad x_i > \tau_i
\end{cases}       
\end{equation*}

The two losses, wasted memory loss and cold start loss, combine to give the overall loss.

\begin{definition}
The cumulative loss $L_{i}(\hat{\theta}, \mathcal{H}_{i}, \theta)$, or social cost, over \textit{i} arrivals is  
\begin{equation*}
\begin{split}
 L_{i}(\hat{\theta}, \mathcal{H}_{i}, \theta) & = \sum_{j=1}^{i} wm(\tau_{\hat{\theta}, \mathcal{H}_{j-1}}, x_{j}) + \sum_{j=1}^{i} \theta \cdot cs(\tau_{\hat{\theta}, \mathcal{H}_{j-1}}, x_{j}) .
\end{split}
\end{equation*}

where $\hat{\theta}$ is the submitted report, and $\theta$ is the true cost of a cold start for the customer. 
\end{definition}

In the definition, the cumulative loss depends on the report $\hat{\theta}$ only through the sequence of keep-alive window lengths it induces.  Thus, in a slight abuse of notation we can also apply it to an arbitrary fixed window length $\tau$ as 
\begin{equation*}
\begin{split}
 L_{i}(\tau, \mathcal{H}_{i}, \theta) & = \sum_{j=1}^{i} wm(\tau, x_{j}) + \sum_{j=1}^{i} \theta \cdot cs(\tau, x_{j}) .
\end{split}
\end{equation*}

\section{Learning a Mixture of Fixed Keep-Alive Policies}
\label{learn-experts}

For cloud providers to deliver good results for both themselves and their customers, we take as our primary objective to optimize custom keep-alive policies in terms of the cumulative loss $L_{i}(\hat{\theta}, \mathcal{H}_{i}, \theta)$. In prior work, \citet{narayana2023keep} and \citet{shahrad2020serverless} developed a closed form of the optimal keep-alive policy when the cloud provider has knowledge of the distributions of the arrivals. However, our work focuses on deriving keep-alive policies where the cloud provider may not know this distribution. In this section, we examine the design of custom keep-alive policies that get close to a fixed optimal keep-alive policy via online learning.


Our custom keep-alive policy is starts from a finite set of keep-alive window lengths $\mathcal{T} = \{ \tau_{1}, \tau_{2},   \cdots, \tau_m\}$.  Our custom keep-alive policy is then constructed using the standard exponential weights approach from online learning.  We assign positive weights to every fixed policy $\tau_{j} \in \mathcal{T}$ and select one with probability proportional to the weights. 


\begin{definition} \label{exp_prob_wts}
The exponentially weighted average custom cache policy chooses fixed policy $\tau_{j} \in \mathcal{T}$ after $n$ arrivals with probability
\begin{equation*}
Pr( \tau_{\hat{\theta},\mathcal{H}_{n}} =  \tau_j )  = \frac{e^{-  L_{n} (\tau_j, \mathcal{H}_{n}, \theta) }}{\sum_{\tau_{k} \in \mathcal{T}} e^{ -  L_{n}  (\tau_k, \mathcal{H}_{n}, \theta)}}
\end{equation*}

\end{definition}


Exponential weights are an attractive approach because they have strong guarantees about the expected {\em regret} of the resulting policy.

\begin{definition}
The regret of the custom policy with respect to $\tau_j \in \mathcal{T}$ after $n$ arrivals is denoted
\begin{align*}
\begin{split}
R_{n}(\hat{\theta},\tau_j, \mathcal{H}_{n}, \theta) & = L_{n} (\hat{\theta}, \mathcal{H}_{n},\theta) - L_{n} (\tau_j, \mathcal{H}_{n}, \theta)  .
\end{split}  
\end{align*}
\end{definition}

Exponential weights guarantee that the regrets are sub-linear in $n$ for every $\tau_j \in \mathcal{T}$, meaning our custom policy does almost as well as the best fixed policy in hindsight.  This is attractive in that this is achieved without needing to know the parameters of the arrival process in advance.  Prior work has shown that such fixed policies are a strong baseline because there often exists a fixed policy that is a good approximation of the optimal policy \citep{shahrad2020serverless,narayana2023keep}. 
While we focus on exponential weights as it is so common and simple, our results likely apply to other regret minimization algorithms as the structure of our domain naturally promotes desirable properties, as the results in the next section show.

\section{Properties of Custom Policies Learned from Fixed Policies}
\label{properties-custom}

In this section, we examine the properties of exponentially weighted average custom keep-alive policies.  We focus on those properties that only depend on the keep-alive policy itself, deferring those that also depend on the payments to Section~\ref{sec:payments}.  We begin by defining the two relevant properties.
\begin{enumerate} 
    \item  Efficiency. We call a custom keep-alive policy \textit{asymptotically efficient}\footnote{Since our algorithm only optimizes over $\tau \in \mathcal{T}$, strictly speaking the outcome achieved is asymptotically maximal-in-range~\citep{nisan2007computationally}.} if the custom policy has a vanishing per-round expected regret. That is, 
    \begin{equation*}
     \lim_{n \rightarrow \infty}\max_{\tau \in \mathcal{T}} E\left[\frac{R_{n}(\hat{\theta},\tau, \mathcal{H}_{n}, \theta)}{n}\right] \rightarrow 0  
    \end{equation*}
    \item Monotonicity. A custom keep-alive policy is  monotone if the probability (over the randomness in the policy) of a cold start is monotonically non-increasing in $\hat{\theta}$. Mathematically, 
\begin{equation*}
\frac{\partial \; [Pr(\tau_{\hat{\theta},\mathcal{H}_{n}} \; < x_{i})]}{\partial \hat{\theta}} \; \leq 0    
\end{equation*}
\end{enumerate}

Asymptotic efficiency is immediate from the standard regret guarantees of exponential weights.\footnote{Strictly speaking this requires that all $\tau \in \mathcal{T}$ are finite to provide an upper bound on the per-arrival loss.  In practice, the difference between a large but finite $\tau$ and an infinite one is negligible, so we make use of $\tau = \infty$ when convenient.  Furthermore it relies on appropriately setting a learning-rate parameter that we omit for brevity in our exposition and also do not use in our experiments because it is unnecessary for good performance. (This omission can be interpeted as setting it to 1.)}  
It provides a strong guarantee of optimizing the cumulative loss approximately as well as the best fixed window length in hindsight.
Such asymptotic guarantees provide reliable performance for applications that have a reasonable number of arrivals.  In practice, it may be desirable to use heuristic policies for the first few arrivals to limit costs due to applications that are only ever invoked a handful of times.  See related discussion in Section~\ref{sec:Azure}.

Monotonicity is a natural property in its own right in terms of acting in a manner that is interpretable by customers: if they report a higher cost of cold starts they will receive (weakly) fewer cold starts.
It is also a desirable technical property we use when establishing properties related to incentives in Section~\ref{sec:payments}.
One of our key insights is that the structure of our setting makes algorithms like exponential weights inherently monotone, as the following lemma captures.
 
\begin{lemma} \label{lem:monotone}
The exponentially weighted average custom policy learned from a mixture of fixed keep-alive policies is monotone in $\hat{\theta}$.
\end{lemma}

The proof of Lemma \ref{lem:monotone} is provided in the Appendix. To prove that the expected rate of a cold start of the exponentially weighted average keep-alive policy is monotonically non-increasing in $\hat{\theta}$ we first construct the expression for the probability of a cold start. Then, we examine the sign of the partial derivative of the probability of the cold start with respect to $\hat{\theta}$. Our proof of Lemma~\ref{lem:monotone} does not rely on the full details of our settings, but rather only that (a) the loss used is affine in $\hat{\theta}$ and (b) the coefficient on $\hat{\theta}$ is monotone in $\tau$.  In our setting, this corresponds to saying that the number of cold starts is decreasing in $\tau$.  But this analysis would apply more broadly to other related challenges in cloud resource allocation (e.g.~the number of timed-out requests is decreasing in the number of servers allocated to processing them) and challenges in other market domains (when implementing an auto-bidder for an auction, increasing the bid results in more auction wins).  Thus, our analysis is more broadly relevant to a variety of domains, particularly when implemented by a platform owner who has complete information about counterfactual outcomes.

\section{Payments}
\label{sec:payments}

Our prior analysis has focused on properties that solely depend on the policy.  We now examine two ways of charging customers and their effects on incentives and the ability of the cloud provider to recover its costs.  As we will see, the two end up being quite similar in practice although they provide contrasting guarantees about customer incentives and provider cost recovery.

In order to discuss customer incentives, we must first define their objective.  Customers strive to minimize their total cost, which combines both their cold starts and the amount they are charged by the provider.  We can express this for the \textit{i}-th inter-arrival for a customer with type $\theta$ as
\begin{equation*}
\text{cost}_{i}(p_{i}, \tau_i, x_{i}, \theta) = p_{i} + \theta \cdot cs(\tau_{i}, x_{i})   
= \begin{cases}
p_{i} , \quad \text{if} \quad x_{i} \; \leq \tau_{i}\\
p_{i} + \theta, \quad \text{if} \quad x_{i} > \tau_{i}
\end{cases}  
\end{equation*}  

We can now formally define our desiderata, which consist of two incentive properties for the customer and a revenue concern for the provider.

\begin{enumerate}
    \item Incentive Compatibility (IC). A custom keep-alive policy is said to be \textit{incentive compatible} if reporting the true cost of a cold start is a dominant strategy for the customer. That is, for all reports $\hat{\theta}$,
    $$\sum_i \text{cost}_{i}(p_{i}(\theta), \tau_{\theta}, x_{i}, \theta) \; \leq \sum_i \text{cost}_{i}(p_{i}(\hat{\theta}), \tau_{\hat{\theta}},x_{i}, \theta).$$  We say the payment rule is $\varepsilon$-IC if $$\sum_i \text{cost}_{i}(p_{i}(\theta), \tau_{\theta}, x_{i}, \theta) \; \leq \sum_i \text{cost}_{i}(p_{i}(\hat{\theta}), \tau_{\hat{\theta}},x_{i}, \theta)+n\varepsilon.\footnote{We allow $\varepsilon$ to scale with $n$ as we believe it is adequate that the incentive to deviate per arrival is small.}$$
    \item Individual Rationality (IR). A policy is said to be \textit{individually rational} if the customer does not pay more than her cost of a cold start to prevent one. Formally, we want for all $\hat{\theta}$ that $\sum_i p_{i}(\hat{\theta}) \; \leq \sum_i \hat{\theta} \cdot cs(\tau_{i}, x_{i})$. We say the payment rule is $\varepsilon$-IR if $\sum_i p_{i}(\hat{\theta}) \; \leq \sum_i\hat{\theta} \cdot cs(\tau_{i}, x_{i})+n\varepsilon$
    \item Cost Recovery (CR). A payment rule is said to be \textit{cost recovering} if the customer's payment covers the wasted memory cost incurred by the cloud provider. Formally, $p_{i} \; \geq wm(\tau_{i}, x_{i})$. 
\end{enumerate}

In general, individual rationality is important to ensure the willingness of customers to participate in a marketplace.  Here, it perhaps has somewhat less bite because cold starts are only one of the considerations faced by serverless customers.  Still, both our payment schemes satisfy it, at least approximately.

More concerningly, if we wish customers to truthfully reveal their cold start cost our payment scheme should lead to an incentive compatible mechanism.  Our allocation rule is an approximation of VCG's, but it is well known that approximate versions of VCG need not be incentive compatible in general.  Our two payment schemes represent different approaches to this challenge.  The first builds on our observation from Section~\ref{properties-custom} that our allocation rule is monotone in $\hat{\theta}$.  Thus, we can follow Myerson's approach to computing appropriate payments to guarantee ex-post incentive compatibility.  The second simply charges VCG-style payments and explores whether in this setting the result is approximately incentive compatible.

One reason to consider the second scheme, even if it has weaker incentive guarantees, is that it provides a stronger guarantee to the cloud provider that the revenue will cover the cost of caching.  While certainly beneficial, this may not be an essential feature for the marketplace.  Current systems like AWS Lambda and Azure Functions provide some amount of caching for free as part of their overall service~\citep{shahrad2020serverless}, making up the cost in some of the other fees they charge the customer.  Still, being able to correctly attribute this cost to customers to the extent they incur it is likely to result in a fairer set of charges to customers and a more efficient marketplace overall.

In the remainder of this section, we introduce our payment schemes formally, analyze what theoretical guarantees they provide for our desiderata, and supplement this with empirical analysis where their guarantees are only approximate.

\subsection{Myerson Payments}

Lemma~\ref{lem:monotone} shows that our allocation rule is monotone in $\hat{\theta}$, so we can directly apply Myerson's machinery to derive a payment function which guarantees incentive compatibility.  The following proposition summarizes this derivation, which applies to {\em any} monotone policy, not merely the one we study.

\begin{proposition} \label{myerson_lemma}
Given the reported cost of a cold start $\hat{\theta}$ by a customer, a monotone custom keep-alive cache policy is incentive compatible and individually rational if the payment after the \textit{i}-th arrival is
\begin{equation*}
p_{i}(\hat{\theta}) = \int_{0}^{\hat{\theta}} Pr( \tau_{\hat{\theta} = y,\mathcal{H}_{i}} < x_{i}) \; dy - \hat{\theta} \cdot \Big( Pr( \tau_{\hat{\theta}, \mathcal{H}_{i}} < x_{i} ) \Big).   
\end{equation*}
\end{proposition}

Instantiating this with the exponential weights rule for computing the probability of each window length yields a concrete formula.

\begin{corollary} \label{exponential_myerson}
For the exponential weighted average policy, the following payment rule is IC and IR: 
\begin{align*}
\begin{split}
p_{i} (\hat{\theta}) & = \int_{0}^{\hat{\theta}} \sum_{\tau_{j} \in \mathcal{T}| \tau_{j} < x_{i}}\frac{ e^{-  L_{i}(\tau_{j}, \mathcal{H}_{i}, y)}}{\sum_{\tau_{k} \in \mathcal{T}} e^{-  L_{i}(\tau_{k}, \mathcal{H}_{i}, y)}} \; dy 
\\ & - \hat{\theta} \cdot  \frac{\sum_{\tau_{j} \in \mathcal{T} | \tau_{j} < x_{i}} e^{-  L_{i}(\tau_{j}, \mathcal{H}_{i}, \hat{\theta})} }{\sum_{\tau_{k} \in \mathcal{T}} e^{-  L_{i}(\tau_{k}, \mathcal{H}_{i}, \hat{\theta})}}.   
\end{split}
\end{align*}
\end{corollary}

The function being integrated is continuous in $y$ and this integral can easily be computed numerically, but we are not aware of a closed form for it in general.  However, the special case where the mixture of fixed policies consists of only the two extreme fixed policies, that is, $\mathcal{T}  = \{\tau = 0, \tau = \infty\}$ has a closed form.  This is of interest because if the arrival process is Poisson, then one of the two is optimal (although the process parameter is needed to determine which)~\citep{narayana2023keep}.

\begin{corollary} \label{payment_poison}
The exponentially weighted average custom policy when the mixture of fixed policies is comprised of $\mathcal{T} = \{\tau = 0, \tau = \infty\}$ has a closed form payment rule of
\begin{equation*}
 p_{i}(\hat{\theta}) = \hat{\theta} \cdot  \frac{e^{ i \cdot \hat{\theta} - c_{p} \cdot \sum_{j = 1}^{i} x_{j}}}{1 + e^{i \cdot \hat{\theta} - c_{p} \cdot \sum_{j = 1}^{i} x_{j} }} +  \frac{1}{i} \cdot \log \Big( \frac{1 + e^{ -c_{p} \cdot \sum_{j = 1}^{i} x_{j}}}{1 + e^{ i \cdot \hat{\theta} - c_{p} \cdot \sum_{j = 1}^{i} x_{j}}} \Big)    
\end{equation*}
\end{corollary}

While this formula remains complicated, it can be observed that the behavior of the payment rule is dependent on the sign of the exponents in the payment formula, that is, $i \cdot \hat{\theta} - c_{p} \cdot \sum_{j = 1}^{i} x_{j}$. 
This form is suggestive because the  optimal policy in the case of the Poisson process is determined by the sign of $\frac{c_{p}}{\hat{\theta}} - \lambda$ \citep{narayana2023keep}.  In particular, $\mathbb{E}[x_{j}] = \frac{1}{\lambda}$, meaning that, at least in expectation, the two are consistent.

While we know of no guarantee that this payment rule is even approximately cost recovering, this observation provides some intuition for why we should expect reasonable performance in this regard.  In many cases the allocation rule converges toward the efficient one in the sense that the optimal window length has probability essentially 1.  In such cases, by the standard analysis of VCG, we know the average payment must converge toward the average wasted memory cost experienced by the cloud provider.  Of course, such convergence is not guaranteed, particularly with more complex arrival processes.  Nevertheless, our empirical results show it remains close in practice for key processes of interest.

\subsection{Externality Payments}

Given that the policy is an approximation of VCG, a natural alternative payment rule is to simply change the externality, which in this context is exactly the wasted memory.  This has the advantage of being simple to explain to customers (in the same way that the generalized second price auction (GSP) is easier to explain than VCG in advertising auctions), is consistent with the current market design where customers who ask to always be cached pay the full cost of that caching, and guarantees cost recovery.

However, its incentive properties are more delicate.  First, it is only approximately individually rational.  Due to exploration, even a report of $\hat{\theta}=0$ will end up with some caching during the learning process.  However, due to the asymptotic efficiency guarantee the customer will do approximately at least as well as the best fixed window length.  As long as $0 \in \mathcal{T}$, which seems a reasonable decision in practice since some customers may have no need for low latency, this guarantees $\varepsilon$-IR in expectation over the randomness of the policy for sufficiently large $n$.

For incentive compatibility, exponential weights guarantee that the expected cumulative regret is non-negative~\citep{gofer2016lower}.  If every possible report leads to a worse outcome than the fixed policies, this would be enough to guarantee approximate IC, with the approximation determined by the regret bound. 
Unfortunately, even if the regret is non-negative for every fixed policy for a given report, it is still possible that if the regrets were instead computed relative to the type rather than the report the resulting regret would have been negative.  Our experimental results show this does indeed happen in practice.  Luckily, it appears that the regret in this case is never {\em too} negative, and whenever this holds we can guarantee approximate IC.

To make this precise, let 
\begin{equation*}
{
\kappa(\hat{\theta}) = \left(\sum_i wm(\tau_{\hat{\theta},\mathcal{H}_{i}},x_i), \sum_i cs(\tau_{\hat{\theta}, \mathcal{H}_{i}},x_i)\right)
}
\end{equation*}
denote the function mapping $\hat{\theta}$ to the resulting expected combination of wasted memory and cold starts.
In a slight abuse of notation, we similarly let $\kappa(\tau) = (\sum_i wm(\tau, x_i), \sum_i cs(\tau, x_i))$ denote the combination of wasted memory and cold starts that result from a fixed policy $\tau \in \mathcal{T}$.

\begin{lemma}
\label{lem:IC}
Let $R_n$ be an upper bound on the policy's expected regret after $n$ arrivals and $$D(\theta)=\max(0,\max_{\hat{\theta}} \min_{\tau \in \mathcal{T}} ((\kappa(\tau) - \kappa(\hat{\theta})) \cdot (1,\theta))).$$  Externality Payments are $1/n(R_n+\max_\theta \cdot D(\theta))$-IC.
\end{lemma}

\begin{proof}
With externality payments, the customer's cost is exactly the social cost.  The expected loss is worse than the ex-post optimal fixed policy loss by at most $R_n$.  In turn, the most any point achievable by any misreport can improve on this for a given $\theta$ is, by definition, $D(\theta)$.
\end{proof}

The tightness bound in Lemma~\ref{lem:IC} is determined by how much $\kappa(\hat{\theta})$ dips ``below'' the Pareto frontier of the $\kappa(\tau)$.  If it is always above the frontier then $D(\theta)= 0$ and the bound reduces to simply $R_n/n$ which vanishes asymptotically.  If it is ever below it, $D(\theta)$ calculates how much improvement is possible for a given $\theta$.  Furthermore, $D(\theta)$ is bounded because of the ordering on $\tau \in \mathcal{T}$. Nothing can have less wasted memory than the smallest choice of $\tau$ while nothing can have fewer cold starts than the largest.

Figure~\ref{fig_externality} illustrates a hypothetical scenario giving rise to a positive, but small value of $D(\theta)$.  It illustrates a $\kappa(\theta)$ curve which enables slightly better tradeoffs than the convex hull of the $\kappa(\tau)$,  Thus, for some values of $\theta$ that have no regret relative to any fixed $\tau$ (all the $\kappa(\tau)$ points are above the line of tradeoffs among which the customer is indifferent to $\kappa(\theta)$), there are choices of $\hat{\theta}$ below the line that the customer would have preferred to report.  $D(\theta)$ measures the largest that this gap gets.  In contrast, if the $D(\theta)$ curve were entirely inside the convex hull of $D(\tau)$, every choice of $\hat{\theta}$ would be dominated by some $\tau \in \mathcal{T}$ and $D(\theta)$ would be uniformly zero.

While Lemma~\ref{lem:IC} provides good incentive guarantees when $\kappa$ is well behaved (and weaker ones in general), we do not have a theoretical guarantee of such good behavior. Nevertheless, our empirical results show it typically occurs in practice.

\begin{figure}[t!] 
    \centering
     \includegraphics[height=2.1 in]{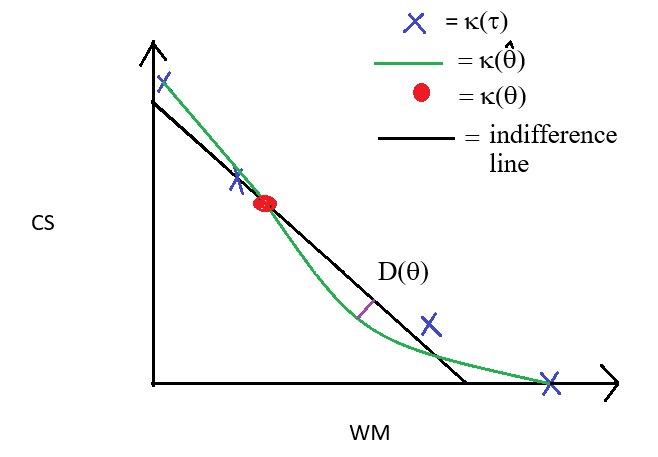}
    \caption{Illustration of IC for Externality Payments}
    \label{fig_externality}
\end{figure}

\subsection{Simulations} \label{simulations}

Thus far we have seen that Myerson payments guarantee IC while externality payments guarantee CR.  While neither guarantees the other's property, we have presented intuitive explanations for why we may still expect good results.  In the remainder of this section, we examine their behavior empirically and show that this does indeed seem to occur in practice.


Based on prior work, we simulated applications that follow two different arrival processes: Poisson and Hawkes.  We use a uniform random distribution to select the parameter of the Poisson process.  We use \textit{Ogata's modified thinning algorithm} \citep{ogata1981lewis} to generate the samples of the Hawkes process.  We generate 200 sample points of application arrivals in a single realization of the Hawkes process. The parameters of the Hawkes process were chosen from uniform random distributions. More specifically, the parameter $\lambda_0$ was sampled from uniform random distribution $\mathcal{U}(0.0,1.0)$, alpha and beta from uniform random distribution $\mathcal{U}(0.0,5.0)$. We use a set of fixed window experts $\mathcal{T} = \{0.0, 1.0, 2.0, 4.0, 8.0, 16.0, 32.0, 64.0 \}$, and customers have a set of possible values of $\theta$ and $\hat{\theta}$ for their true and reported cold start costs of $\{0.0, 0.125, 0.25, 0.50, 1.0, 2.0, 4.0, 8.0, 16.0, \\ 32.0, 64.0 \}$, which are powers of 2 covering a wide range. 
We normalize the cost for the provider to keep the application image in the cache per unit time to be 1 unit. That is, $c_{p} = 1$. Overall, 100 simulation runs were performed with Poisson process arrivals and Hawkes process arrivals.



Table~\ref{tab:regret} examines the incentives of the two payment rules, as measured by the customer's regret $\rho = \sum_i \text{cost}_{i}(p_{i}, \tau_{\theta}, x_{i}, \theta) - \max_{\hat{\theta}} \sum_i \text{cost}_{i}(p_{i}, \tau_{\hat{\theta}},x_{i}, \theta)$.\footnote{This notion of the customer's regret is a standard way to measure violations of IC and should not be confused with the {\em policy's regret} relative to fixed policies.}  A regret that is always 0 means the payment rule is IC while positive regrets quantify the incentive of the customer to misreport.  The table reports both how often the regret is non-zero and how large it is on average in comparison to the customer's total cost.  All values are averaged over 100 runs and all values of $\theta$.  For $\overline{\rho}$, we report the average only of the {\em non-negative} regrets to focus on the strength of the incentive to misreport when one exists.  For the total cost,  we report the standard deviation to illustrate the substantial variation across runs and types.

Since Myerson's rule is IC, the regrets are uniformly zero.  For the Externality rule, non-zero regrets do occur a meaningful fraction of the time, but when they do are small on average (less than two percent of the cost).  This shows that the circumstances discussed in the context of Lemma~\ref{lem:IC} that lead to good incentives tend to occur in practice.  

Table~\ref{tab:cost_rec} examines the same simulation to analyze cost recovery.  If  $\sum_i wm_i - \sum_i p_i = 0$, the provider exactly recovers the cost of the wasted memory.  This time, it is the Externality rule that, by construction, exactly covers its costs.   Consistent with the intuition from Corollary~\ref{payment_poison}, on average the extent to which Myerson's rule covers its costs is essentially indistinguishable from zero.  So while any given customer may end up paying more or less than their true wasted memory cost, at the scale of the entire market the provider essentially achieves cost recovery.

\begin{table*}[t]
\centering
\begin{tabular}{| p{20mm}| p{15mm}| p{15mm}|p{15mm}| p{15mm}|p{15mm}|p{15mm}|}
\hline
Payment Scheme &  \multicolumn{3}{c|}{Poisson Process} &  \multicolumn{3}{c|}{Hawkes Process}\\
\cline{1-7}
& \%  $\rho >0$ & $\overline{\rho}$ & $\sum_i cost_i$ &\%  $\rho >0$ & $\overline{\rho}$ & $\sum_i cost_i$\\
\hline
Myerson's Rule &   0.0  &  $0.0 \pm 0.0$ &  $347.83 \pm 265.71 $  & 0.0  &  $0.0 \pm 0.0$ &  $260.89 \pm 232.88 $\\
\hline
Externality Rule &    22.15 & $5.15 \pm 2.48$   &  $353.50 \pm 267.14$  &  19.90  & $2.93 \pm 1.00 $    &   $262.93 \pm 231.33$\\
\hline
\end{tabular}
\caption{Summary statistics of positive regret  with both payment schemes for 100 simulation runs \vspace{-5mm}}
\label{tab:regret}
\end{table*}

\begin{table*}[t]
\centering
\begin{tabular}{| p{20mm}| p{15mm}| p{15mm}| p{15mm}| p{15mm}| p{15mm}| p{15mm}|}
\hline
 &  \multicolumn{3}{c|}{Poisson Process} &  \multicolumn{3}{c|}{Hawkes Process}\\
\cline{1-7}
& $\sum_i wm_i - \sum_i p_i$ & $\sum_i p_i$ & $\sum_i wm_i$ &  $\sum_i wm_i - \sum_i p_i$ & $\sum_i p_i$ & $\sum_i wm_i$\\
\hline
Myerson's Rule &   $-3.93 \pm 20.58$ &  $236.27 \pm 152,84 $ & $240.20 \pm 154.11$  & $-0.005 \pm 18.01$    & $177.00 \pm 128.90$  &  $177.01 \pm 126.40$   \\
\hline
Externality Rule &   $0.0  \pm 0.0$ &   $240.20 \pm 154.11$ & $240.20 \pm 154.11$   &   $0.0  \pm 0.0$  & $177.01 \pm 126.40$  & $177.01 \pm 126.40$     \\
\hline
\end{tabular}
\caption{Summary statistics for Cost Recovery  with both payment schemes for 100 simulation runs \vspace{-5mm}}
\label{tab:cost_rec}
\end{table*}

\begin{figure*}[t!]
    \centering
    \begin{subfigure}[t]{0.48\textwidth}
        \includegraphics[height=1.73in]{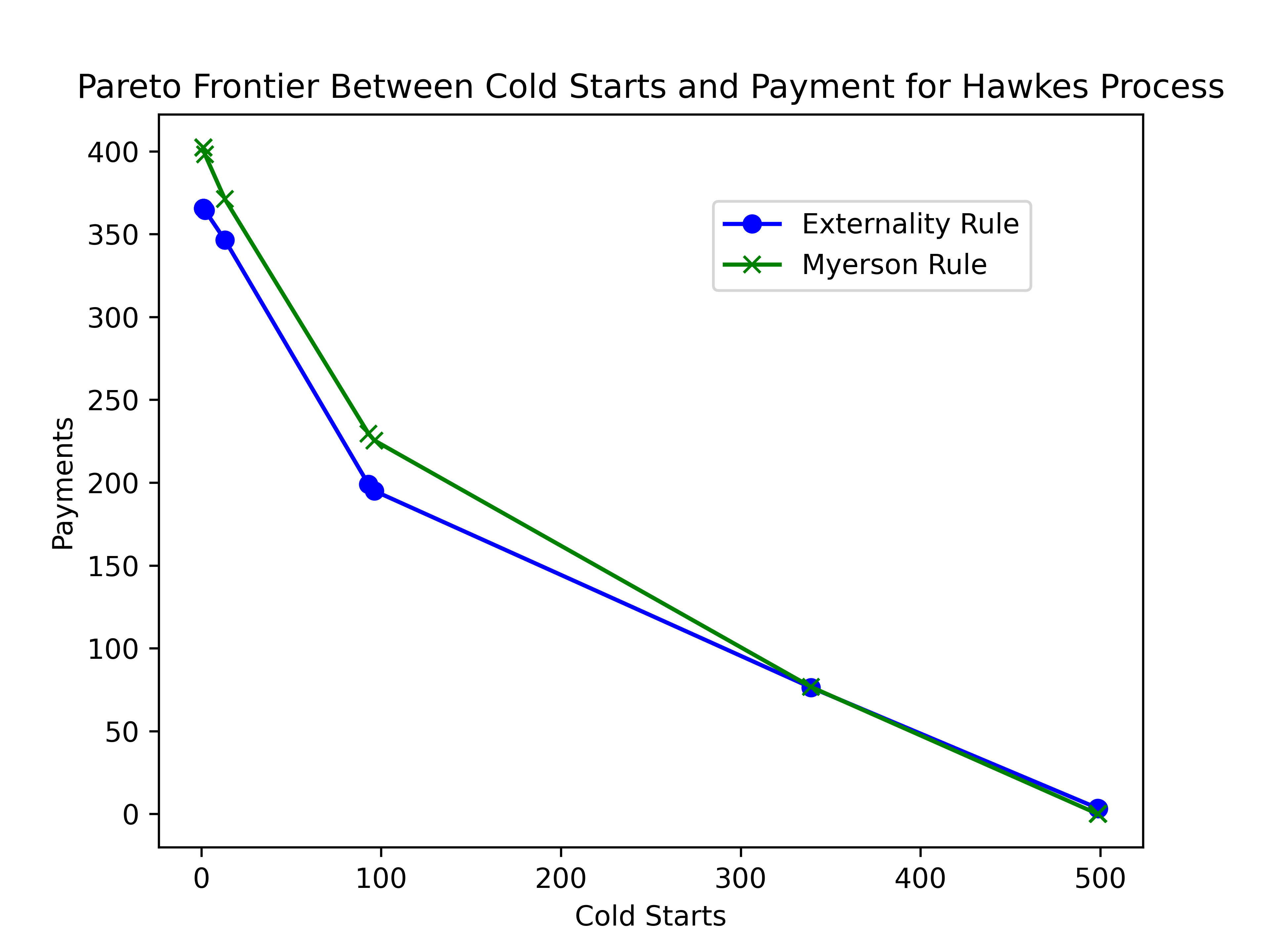}
        \caption{Hawkes process application}
    \end{subfigure}%
    \hfill 
    \begin{subfigure}[t]{0.48\textwidth}
        \includegraphics[height=1.73in]{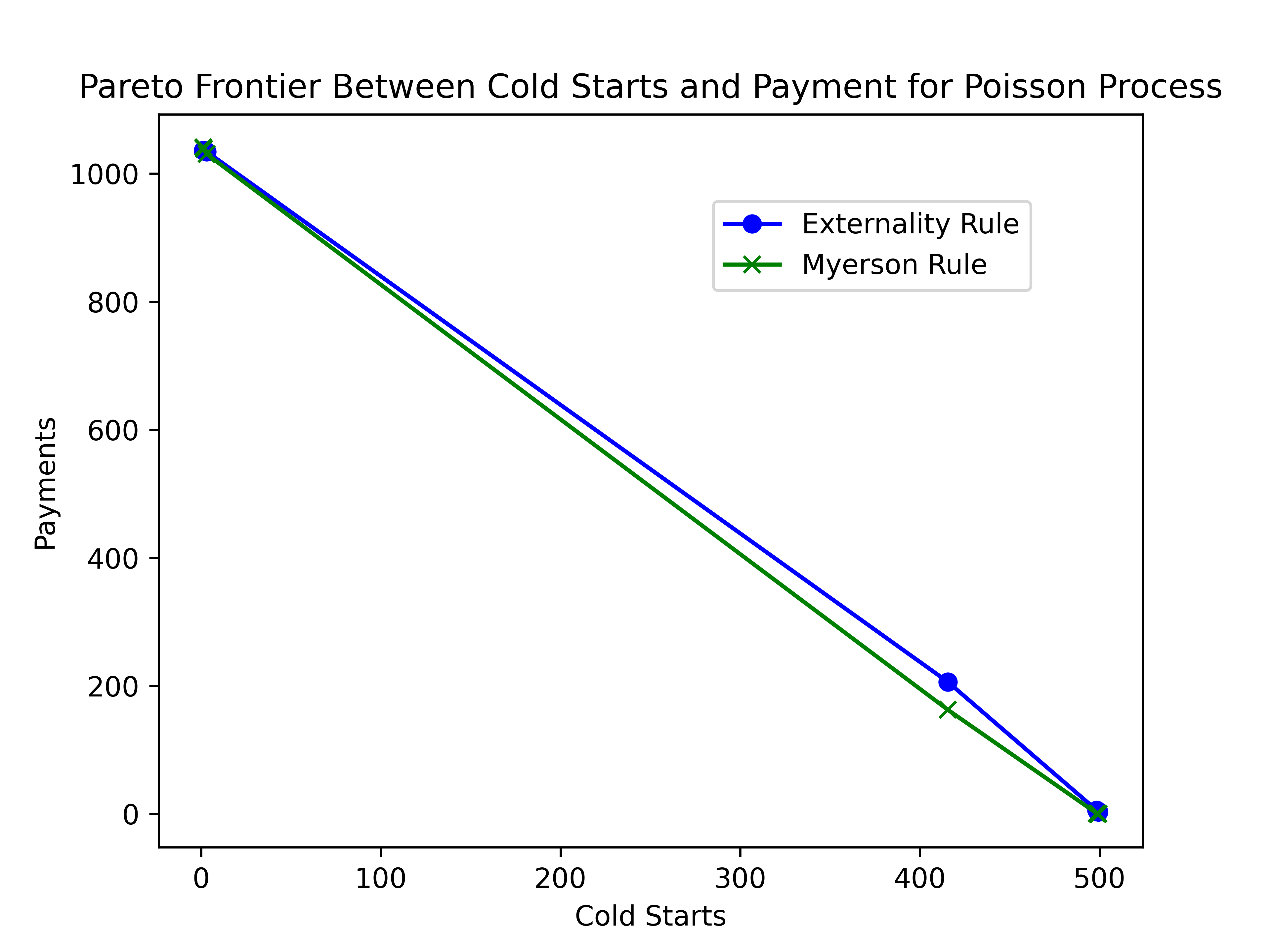}
        \caption{Poisson process applications}
    \end{subfigure}    
   \caption{Trade-off curve of Payments vs number of cold starts for Hawkes and Poisson Process applications}
    \label{poisson_fig}
\end{figure*}



Overall, despite contrasting theoretical guarantees the two payment rules achieve similar performance in practice.  As Figure~\ref{poisson_fig} shows in two example runs, this is because the payments charged by the two rules tend to be quite similar in practice despite their very different approaches.  In some cases, as on the left, Myerson charges slightly more, while in others, as on the right, it charges slightly less.  (See the supplemental material for additional examples which demonstrate this behavior.)  Thus, the choice between the two in practice may come down to which property is considered more important to strictly guarantee.  Additionally, the simpler nature of the Externality rule may be attractive.

\subsection{Experiments on Azure Data trace}
\label{sec:Azure}

Our simulation results so far have been derived on synthetic scenarios using natural choices of arrival process.  We conclude our analysis using a more realistic scenario based on a subset of the Microsoft Azure Data trace released by \citet{shahrad2020serverless}.\footnote{These traces are available at \url{https://github.com/Azure/AzurePublicDataset}}. In the Azure Data trace, an application consists of multiple functions with each function responsible for completing a specific task. The Azure Data traces collect invocation counts of functions that belong to an application binned in 1-minute intervals. We examine the regret and the cost recovery of both the payment rules for applications in a subset of the Azure Data trace identified as good fits to Hawkes processes using the methodology of \cite{narayana2023keep}.  This restricts to a set of applications where our underlying fixed policies are known to be a good choice.

Within this set of applications, we limit to applications with at most 180 arrivals.
This avoids excessive computation from applications with a very large number of arrivals as well as avoids underflow issues since the computation of weights, and Myerson payments involve exponents of large negative numbers, which can be problematic even with the log-sum-exp trick. After these two filtering steps, 3173 applications remained. We use as our fixed keep-alive window lengths $\mathcal{T}=\{ 5 \; \text{minutes},  10 \;\text{minutes}, 20 \;\text{minutes}, 30 \;\text{minutes},\\ 45 \;\text{minutes}, 60 \;\text{minutes}, 90 \;\text{minutes}, 120 \;\text{minutes}.\}$, which match those considered in prior work on this trace\cite{shahrad2020serverless,narayana2023keep}. We assume that the customers have a set of possible values of $\theta$ and $\hat{\theta}$ for their true and reported cold start costs of $\{0.0, 5.0, 10.0, 20.0, 30.0, 45.0, 60.0\}$, which again match the costs of cold starts used in prior work~\cite{narayana2023keep}. 

Table \ref{azure_table} examines the regret and the cost recovery of the two payment schemes on a subset of the Azure Data trace. First, we observe that the results from Table \ref{azure_table} are broadly consistent with those obtained from the performance of the payment scheme on simulated data.
That is, Myerson's rule has a cost recovery close to zero and much smaller than the variation. While the externality rule yields a somewhat higher IC violation, it is still on the order of ten percent of the overall cost, which remains modest.

\begin{table*}[t]
\centering
\begin{tabular}{| p{20mm}| p{18mm}| p{15mm}| p{15mm}| p{15mm}| p{15mm}| p{15mm}|}
\hline
 &  \multicolumn{6}{c|}{Azure Data}\\
\cline{1-7}
& $\sum_i wm_i  - \sum_i p_i$ & $\sum_i p_i$ & $\sum_i wm_i$ &\%  $\rho >0$ & $\overline{\rho}$ & $\sum_i cost_i$\\
\hline
Myerson's Rule &  $-4.15 \pm 76.75 $   & $165.45 \pm 225.75$  &   $169.59 \pm 223.40$ &   $0.0$  &  $0.0 \pm 0.0$ & $377.98 \pm 303.47$    \\
\hline
Externality Rule & $0.0 \pm 0.0$ &  $169.59 \pm 223.40$ & $169.59 \pm 223.40$  & $16.09$   & $47.55 \pm 6.48$   & $452.99 \pm 271.87$  \\
\hline
\end{tabular}
\caption{Summary statistics for Positive Regret and Cost Recovery  with both payment schemes for Azure Data trace \vspace{-5mm}}
\label{azure_table}
\end{table*}

\begin{figure*}[t!]
    \centering
    \begin{subfigure}[t]{0.48\textwidth}
        \centering
        \includegraphics[height=1.73in]{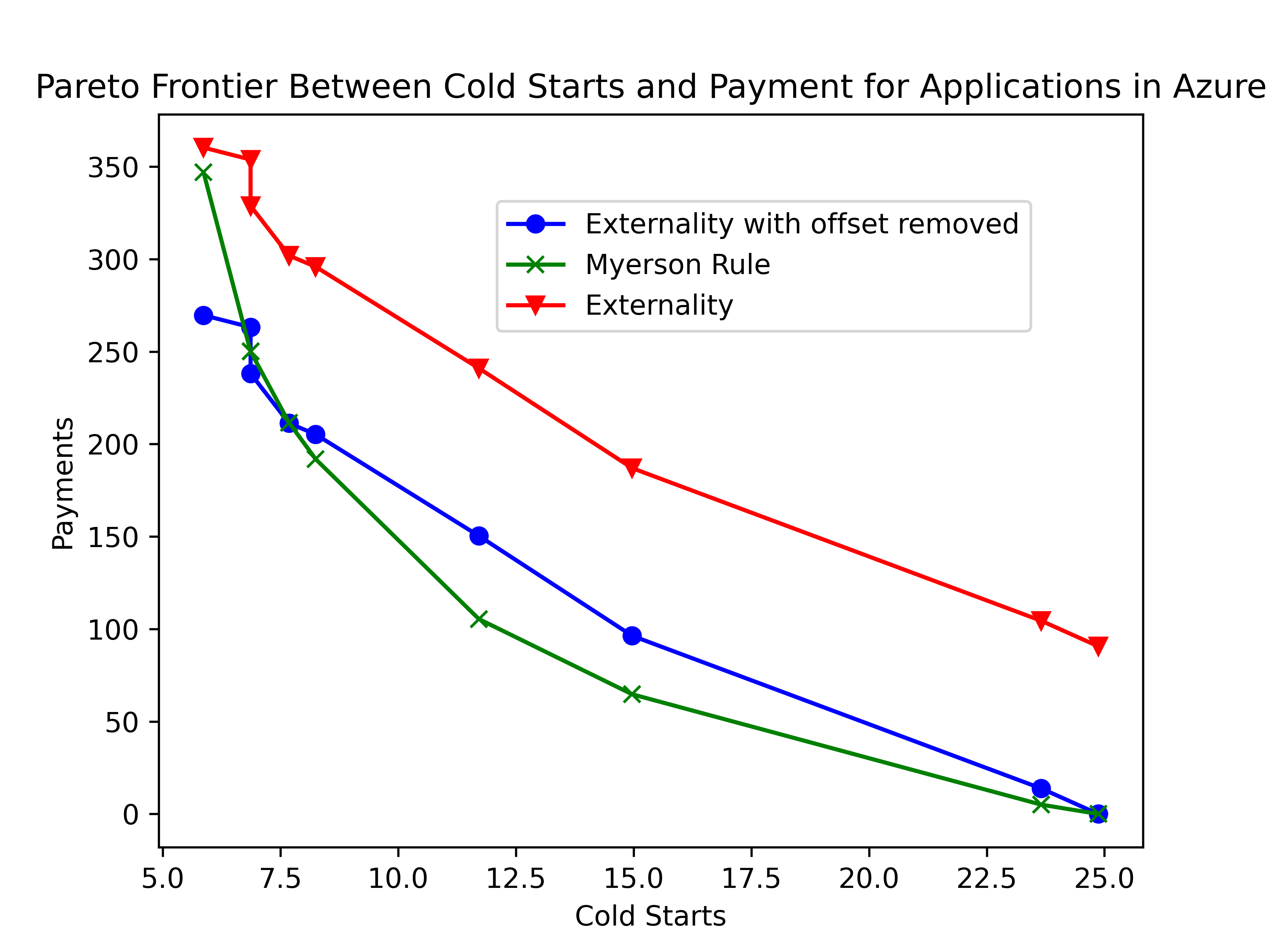}
        \caption{Azure application 1}
    \end{subfigure}%
    \hfill 
    \begin{subfigure}[t]{0.48\textwidth}
        \centering
        \includegraphics[height=1.73in]{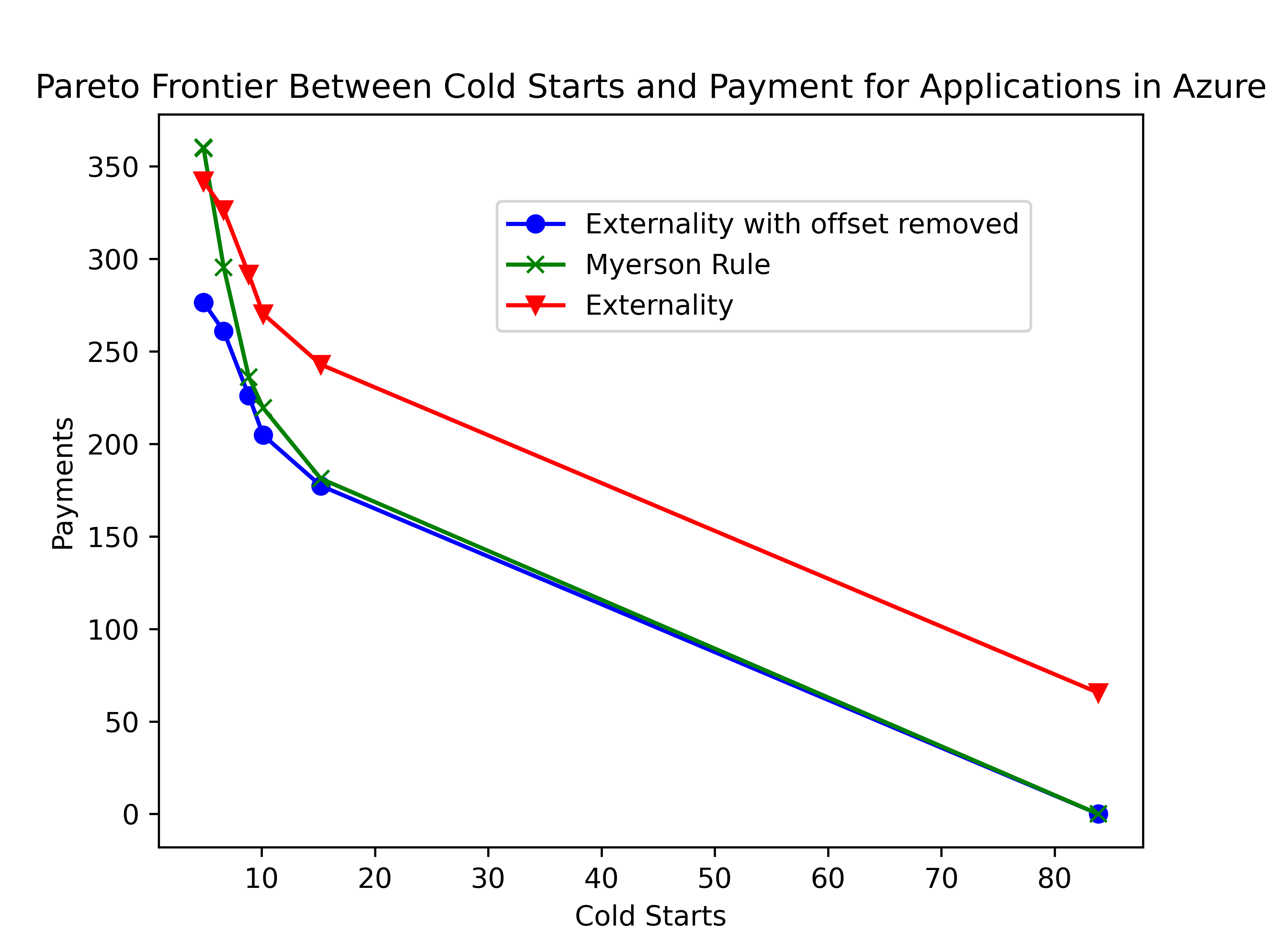}
        \caption{Azure application 2}
    \end{subfigure}    
   \caption{Trade-off curve of payments vs cold starts for Azure applications}
    \label{fig_azure}
\end{figure*}

Figure \ref{fig_azure} shows the trade-off curves for the payments charged and the number of cold starts for randomly sampled applications from the Data set. In some applications, we observe an initial high externality payment in the first arrival (where we explore uniformly at random), which sets a high cumulative externality price for the remaining arrivals. To get a better comparison between the two payment schemes, we subtract this initial externality payment (termed as "offset" in Figure \ref{fig_azure}) from all the externality payments.  This yields a result much closer to what we saw in our synthetic simulations, which lacked such large inter-arrivals.  This result suggests that in practice it may be better to adopt some default policy for the first few arrivals (e.g. no or limited caching) and only using the learned policy and charging customers subsequently.  This may also address the somewhat larger regrets in processes with a small number of arrivals driven by regret that a more extreme value of $\hat{\theta}$ was not reported, which would have led to faster learning during this initial period.

\section{Discussion}

This work contributes to the theory of cloud market design by identifying it as one special case where monotonicity of allocation comes ``for free'' with a standard online learning technique, in contrast to prior general work which required specialized algorithms.  We also establish that, in addition to the standard way of using monotonicity to achieve incentive compatibility, simply charging externality payments achieves good incentive guarantees as long as the online learning algorithm was ``well behaved'' ex post, a property that does not generally hold for approximations of VCG.

Our work on custom cache policy allocation with Myerson payment rule and VCG payment rule seems to side-step Myerson-Satterthwaite impossibility (\citeyear{myerson1983efficient}) of achieving incentive compatibility with efficiency without violating individual rationality or budget balance. 
As studied, our mechanism achieves only asymptotic efficiency, and our claims of cost recovery are only empirical.
However, if we consider a variant of our setting where the arrivals are all known {a priori} and we can calculate the exact optimal threshold instead of learning it, then our mechanism would be exactly VCG and so satisfy efficiency and budget balance exactly.
A reason this is possible is because the cloud provider is acting as both the auctioneer and the owner of the resource. This means the budget balance constraint is relaxed to cost recovery.  
This changes the structure of the mechanism that allows it to evade the Myerson-Satterthwaite impossibility result. \citet{cramton1987dissolving} give another example where the impossibility can be evaded due to a different mechanism structure.

On the practical side, we supplement this theory by simulations of customers with Poisson and Hawkes arrival processes and a trace of Microsoft Azure data.  The main takeaway here is that despite contrasting approaches and guarantees the two ways of calculating payments end up being quite similar in practice.  This perhaps supports the use of the simpler externality payments which make it substantially easier to communicate to customers how their charges were determined and, as our results show, surprisingly good incentive properties in the form of having low customer regret. Indeed, payments based on metered usage of some resource are widespread in cloud markets, making this market design particularly natural.  On the other hand, if exact incentive compatibility is considered essential, our results show it does surprisingly well at covering the cost of providing caching despite the lack of guarantees, and we provide some theoretical intuition for why this should be the case.

While we have focused on a practical market design for a particular application, our approach is more general.  It relies on the domain being single parameter and also having a single parameter to control allocations that behaves in a monotone manner.  We have argued that this structure exists both in aspects of the cloud marketplace and in other marketplaces such as auction platforms. Thus, our technical approach is relevant beyond the particular marketplace we study.

Our results suggest a number of interesting open questions.  We studied a particular regret minimization algorithm, but the way exponential weights is naturally monotone in our setting seems likely to apply to other natural regret minimization algorithms. Characterizing the class of algorithms to which this applies is an interesting challenge.  It would also be interesting to develop algorithms that exploit the particular structure of our domain, for example, that avoid discretizing the set of possible window lengths in the manner of \citet{kleinberg2008multi} or that achieve better regret, particularly when there are only a small number of arrivals by avoiding large values of $\tau$ when $\theta$ is small.  Finally, it would be interesting to exploit the substantial structure in both $\kappa(\hat{\theta})$ and $\kappa(\tau)$ to provide tighter bounds on the incentive and cost recovery properties of our payment schemes.

\section*{Acknowledgments}

This material is based upon work supported by the National Science Foundation under Grant No. 2110707.

%
%
%
%
\printbibliography

\clearpage
\newpage


\appendix

\section*{Supplementary Material}


\section{Proof of Lemma 5.1}
\label{app:lem}

\begin{lemma}
The exponentially weighted average custom policy learned from a mixture of fixed keep-alive policies is monotone in $\hat{\theta}$.
\end{lemma}

\begin{proof}
We prove that the expected rate of a cold start of the exponentially weighted average keep-alive policy is monotonically non-increasing in $\hat{\theta}$ by examining the partial derivative of the probability of the cold start with respect to $\hat{\theta}$. Without loss of generality, we assume that the fixed keep-alive policies in the mixture $\mathcal{T}$ satisfy $\tau_{1} \leq \tau_{2} \leq \cdots \leq \tau_{j} \leq \tau_{j+1} \leq \cdots$. In an exponentially weighted average keep-alive policy, the probability of a cold start during the $i$-th arrival is expressed as the sum of the probabilities over all the fixed policies that encounter a cold start loss during the $i$-th arrival.  The probability of a cold start of the exponentially weighted average keep-alive policy is given by the following expression,

\begin{align*}
Pr( \tau_{\hat{\theta}, \mathcal{H}_{i}} < x_{i} ) =  \sum_{\tau_{j} \in \mathcal{T} | \tau_{j} < x_{i}} \; \frac{ e^{-  L_{i}(\tau_{j} , \mathcal{H}_{i}, \theta)} }{\sum_{\tau_{k} \in \mathcal{T}} e^{-  L_{i}( \tau_{k}, \mathcal{H}_{i}, \theta)}}   
\end{align*}

The cumulative loss for a fixed keep-alive window $\tau_{j}$ can be computed as the sum of the cumulative wasted memory accumulated, and the cumulative loss for cold starts.  The cumulative loss $L_{i}(\tau_j, \mathcal{H}_{i}, \theta)$, or social cost, over \textit{i} arrivals is  
\begin{equation*}
\begin{split}
 L_{i}(\tau_{j}, \mathcal{H}_{i}, \theta) & = \sum_{j=1}^{i} wm(\tau_{j}, x_{j}) + \sum_{j=1}^{i} \theta \cdot cs(\tau_{j}, x_{j}) \\
 & = WM(\tau_j, \mathcal{H}_{i})  + \theta \cdot CS(\tau_{j}, \mathcal{H}_{i}) 
\end{split}
\end{equation*}

\noindent where $\hat{\theta}$ is the submitted report, $\theta$ is the true cost of a cold start for the customer, $WM(\tau_j,  \mathcal{H}_{i})  = \displaystyle \sum_{j=1}^{i} wm(\tau_{j}, x_{j})$ denotes the cumulative wasted memory, and $CS(\tau_j,  \mathcal{H}_{i}) = \displaystyle \sum_{j=1}^{i}  cs(\tau_{j}, x_{j})$  denotes the cumulative cold starts.

Expanding the cumulative loss of the fixed policy $\tau_{j}$ we have,    

\begin{align*}
\begin{split}
Pr( \tau_{\hat{\theta}, \mathcal{H}_{i}} < x_{i} )
&= \sum_{\tau_{j} \in \mathcal{T} | \tau_{j} < x_{i}} \; \frac{ e^{- WM(\tau_{j}, \mathcal{H}_{i}) - \hat{\theta} \cdot CS(\tau_{j}, \mathcal{H}_{i}) }}{\sum_{\tau_{k} \in \mathcal{T}} e^{-  L_{i}(\tau_{k}, \mathcal{H}_{i}, \theta)}}\\
\end{split}
\end{align*}

\noindent We take the partial derivative of the probability of a cold start with respect to $\hat{\theta}$ using the quotient rule to obtain,

\begin{align*}
\begin{split}
& \frac{\partial \; [Pr(\tau_{\hat{\theta}, \mathcal{H}_{i}} \; < x_{i})]}{\partial \hat{\theta}} \\ & = \sum_{j | \tau_{j} < x_{i}} \; \Bigg( \frac{\big(\sum_{\tau_{k} \in \mathcal{T}} e^{-  L_{i}(\tau_{k}, \mathcal{H}_{i}, \theta )} \big) \cdot \big( e^{- WM(\tau_{j}, \mathcal{H}_{i}) - \hat{\theta} \cdot CS(\tau_{j}, \mathcal{H}_{i})} \big) \cdot (-CS(\tau_{j}, \mathcal{H}_{i}))}{\big(\sum_{\tau_{k} \in \mathcal{T}} e^{-  L_{i}(\tau_k, \mathcal{H}_{i}, \theta)} \big)^{2}} \\
& - \frac{e^{- WM(\tau_{j}, \mathcal{H}_{i}) -  \hat{\theta} \cdot CS(\tau_{j}, \mathcal{H}_{i})} \cdot \big( \sum_{\tau_{k} \in \mathcal{T}} e^{-  WM(\tau_k, \mathcal{H}_{i})  - \hat{\theta} \cdot CS(\tau_{k}, \mathcal{H}_{i})} \cdot (-CS(\tau_{k}, \mathcal{H}_{i}))\big)}{\big( \sum_{\tau_{k} \in \mathcal{T}} e^{-  L_{i}(\tau_{k}, \mathcal{H}_{i}, \theta)} \big)^{2}}\Bigg)
\end{split}
\end{align*}

\noindent Since, we need to determine the sign of the partial derivative $\dfrac{\partial \; [Pr(\tau_{\hat{\theta}, \mathcal{H}_{i}} \; < x_{i})]}{\partial \hat{\theta}}$, and as the denominator is always positive, $\displaystyle \big( \sum_{\tau_{k} \in \mathcal{T}} e^{-  L_{i}(\tau_{k}, \mathcal{H}_{i},\theta)} \big)^{2} \; \geq 0$, we examine the numerator as follows,

\begin{align*}
\begin{split}
& \text{Numerator of } \frac{\partial \; [Pr(\tau_{\hat{\theta}, \mathcal{H}_{i}} \; < x_{i})]}{\partial \hat{\theta}} \\
& = \sum_{j| \tau_{j} < x_{i}} \; \sum_{\tau_{k} \in \mathcal{T}} e^{-  L_{i}(\tau_k, \mathcal{H}_{i}, \theta)}  \cdot  e^{- WM(\tau_{j}, \mathcal{H}_{i}) -  \hat{\theta} \cdot CS(\tau_{j}, \mathcal{H}_{i}) }  \cdot (-CS(\tau_{j}, \mathcal{H}_{i})) \\
& -  \sum_{j | \tau_{j} < x_{i}} \; e^{- WM(\tau_{j}, \mathcal{H}_{i}) -   \hat{\theta} \cdot CS(\tau_{j}, \mathcal{H}_{i})} \cdot \sum_{\tau_{k} \in \mathcal{T}} e^{-  WM(\tau_{k}, \mathcal{H}_{i}) -  \hat{\theta} \cdot CS(\tau_{k}, \mathcal{H}_{i})} \cdot (-CS(\tau_{k}, \mathcal{H}_{i})) \\
& = \sum_{ j| \tau_{j} < x_{i}} \; \sum_{\tau_{k} \in \mathcal{T}} e^{-  WM(\tau_{k}, \mathcal{H}_{i}) - \hat{\theta} \cdot CS(\tau_{k}, \mathcal{H}_{i})}  \cdot  e^{- WM(\tau_{j}, \mathcal{H}_{i})  - \hat{\theta} \cdot CS(\tau_{j}, \mathcal{H}_{i}) }  \cdot (-CS(\tau_{j}, \mathcal{H}_{i})) \\
& +  \sum_{ j | \tau_{j} < x_{i}} \; \cdot \sum_{\tau_{k} \in \mathcal{T}} e^{-  WM(\tau_{k}, \mathcal{H}_{i})  -  \hat{\theta} \cdot CS(\tau_{k}, \mathcal{H}_{i})} \cdot e^{- WM(\tau_{j}, \mathcal{H}_{i}) - \hat{\theta} \cdot CS(\tau_{j}, \mathcal{H}_{i})} \cdot CS(\tau_{k}, \mathcal{H}_{i}) \\
& = \sum_{ j| \tau_{j} < x_{i}} \; \sum_{\tau_{k} \in \mathcal{T}} e^{-  WM(\tau_{k}, \mathcal{H}_{i})  - \hat{\theta} \cdot CS(\tau_{k}, \mathcal{H}_{i})}  \cdot  e^{- WM(\tau_{j}, \mathcal{H}_{i}) -  \hat{\theta} \cdot CS(\tau_{j}, \mathcal{H}_{i})}  \cdot (-CS(\tau_j, \mathcal{H}_{i}) + CS(\tau_{k}, \mathcal{H}_{i})) \\
\end{split}    
\end{align*}

Let us assume there are $m$ fixed policies in the mixture of fixed policies. That is, $\mathcal{T} = \{\tau_{1}, \tau_{2}, \cdots, \tau_{m} \}$.  Without loss of generality, we also assume that the largest length of the fixed keep-alive window to encounter a cold- start for the $i$-th arrival is $\tau_{u}$. Since, a keep-alive policy $\tau$ encounters a cold start if $\tau < x_{n}$. Hence, if a fixed policy with length $\tau_{u}$ encounters a cold start $\tau_{u} < x_{i}$, then all the fixed policies with keep-alive windows lower than $\tau_{j} \leq \tau_{u}$ also encounter a cold start. Hence, as per our assumption the fixed policy with length $\tau_{u+1}$ encounters a warm start $\tau_{u+1} \;\geq  x_{i}$. This implies all the fixed policies with keep-alive windows larger than $\tau_{j} \geq \tau_{u+1}$ also encounter a warm start for the $i$th arrival. Now, the numerator of the partial derivative can be further simplified as follows,

\begin{align*}
\begin{split}
\text{Numerator of } \frac{\partial \; [Pr(\tau_{\hat{\theta}, \mathcal{H}_{i}} \; < x_{i})]}{\partial \hat{\theta}} & \quad \qquad \\
= \sum_{j \in [1,u]} \; \sum_{k \in [1,m]} e^{-  WM(\tau_{k}, \mathcal{H}_{i})  - \hat{\theta} \cdot CS(\tau_{k},\mathcal{H}_{i})}  \cdot  e^{- WM(\tau_{j}, \mathcal{H}_{i})  - \hat{\theta} \cdot CS(\tau_{k}, \mathcal{H}_{i})} & \\ \cdot (-CS(\tau_{j}, \mathcal{H}_{i}) + CS(\tau_{k}, \mathcal{H}_{i}))  \quad \text{(3)}  &\\ 
= \sum_{j \in [1,u]} \; \sum_{k \in [u+1,m]} e^{-  WM(\tau_{k}, \mathcal{H}_{i})  -  \hat{\theta} \cdot CS(\tau_{k}, \mathcal{H}_{i})}  \cdot  e^{- WM(\tau_{j}, \mathcal{H}_{i})  - \hat{\theta} \cdot CS(\tau_{j}, \mathcal{H}_{i})} & \\ \cdot (-CS(\tau_{j}, \mathcal{H}_{i}) + CS(\tau_{k}, \mathcal{H}_{i})) \quad \text{(4)} & \\
\end{split}    
\end{align*}

\noindent In Equation 3, the terms of the summation for $j = k = 1,2, \cdots, u$ cancel out since $\tau_{j} = \tau_{k}$. Now, we examine the signs of the terms in Equation 4, 
\begin{align*}
\begin{split}
e^{-  WM(\tau_{k}, \mathcal{H}_{i}) - \hat{\theta} \cdot CS(\tau_{k}, \mathcal{H}_{i})} \; & \geq 0  \\
e^{- WM(\tau_{j}, \mathcal{H}_{i}) -  \hat{\theta} \cdot CS(\tau_{j}, \mathcal{H}_{i})} \; & \geq 0 \\
-CS(\tau_{j}, \mathcal{H}_{i}) + CS(\tau_{k}, \mathcal{H}_{i}) \; & \leq 0 \quad  \text{for} \; j \in [1,u], \text{and} \; k \in [u+1,m]  \quad (\text{4c})\\
\end{split}    
\end{align*}

\noindent From Equation 4c, we observe that $-CS(\tau_{j}, \mathcal{H}_{i}) + CS(\tau_{k}, \mathcal{H}_{i}) \; \leq 0 \quad  \text{for} \; j \in [1,u], \text{and} \; k \in [u+1,m]$ because the lower fixed keep-alive windows accumulate more cold starts because a cold start occurs if $\tau < x_{i}$. Hence, the sign of the numerator of the partial derivative is negative. 

Therefore, $\dfrac{\partial \; [Pr(\tau_{\hat{\theta}, \mathcal{H}_{i}} \; < x_{i})]}{\partial \hat{\theta}} \; \leq 0$.

\end{proof}

Our proof of Lemma~\ref{lem:monotone} does not rely on the full details of our settings, but rather only that (a) the loss used is affine in $\hat{\theta}$ and (b) the coefficient on $\hat{\theta}$ is monotone in $\tau$.  In our setting this corresponds to saying that the number of cold starts is decreasing in $\tau$.  But this analysis would apply more broadly to other related challenges in cloud resource allocation (the number of timed-out requests is decreasing in the number of servers allocating) and challenges in other market domains (when implementing an auto-bidder for an auction, increasing the bid results in more auction wins).  Thus our analysis is more broadly relevant to a variety of domains, particularly when implemented by a platform owner who has complete information about counterfactual outcomes.

\section{Omitted Figures}

We generated 20 plots of the trade-off between $\sum_i p_i(\hat{\theta})$, and $\sum_i cs_i(\tau_{\hat{\theta}, \mathcal{H}_{i}},x_i)$ for the Poisson process and the Hawkes process. A sample plot for the Poisson process and the Hawkes process is presented in the main manuscript. The remaining 19 plots for the Poisson process and the Hawkes process are shown below. Figures \ref{poisson_figure1}, \ref{poisson_figure2}, and \ref{poisson_figure3} are the trade-off curves generated when the arrivals of application invocations were simulated as the Poisson process. Figures \ref{hawkes_figure1}, \ref{hawkes_figure2}, and \ref{hawkes_figure3} are the trade-off curves generated when the arrivals of application invocations were simulated as the Hawkes process. We observe that most of the Pareto frontier points for the Poisson process are close to the extremes, that is, large number of cold starts with low payment or very few cold starts with large payment. This is because the optimal (in expectation) keep-alive policies for the Poisson process is $\tau_{\hat{\theta}, \mathcal{H}_{i}} = 0$ or $\tau_{\hat{\theta}, \mathcal{H}_{i}} = \infty$. Of course, ex-post occasionally some intermediate option is part of the Pareto frontier by chance.

The Hawkes process examples generally result in a more ``interesting'' Pareto frontier.  Again, many of the examples have a close match between the two payment rules.  While a few have visually larger gaps, note that the y-axis scale varies and these generally correspond to processes with relatively rapid arrival rates and so lower payments than other examples.

We generated figures of the trade-off plots between $\sum_i p_i(\hat{\theta})$, and \\ $\sum_i cs_i(\tau_{\hat{\theta}, \mathcal{H}_{i}},x_i)$ for 10 randomly sampled Azure applications from the Azure Data trace. Two sample plots are presented in the main manuscript. The remaining 8 plots are shown in Figure \ref{figure_azure_appendix1}, and Figure \ref{figure_azure_appendix2} .
They continue the key trend of the two payment rules being very similar once the offset is taken into account.

\begin{figure*}[t!]
    \centering
    \begin{subfigure}[t]{0.44\textwidth}
        \centering
        \includegraphics[height=1.7in]{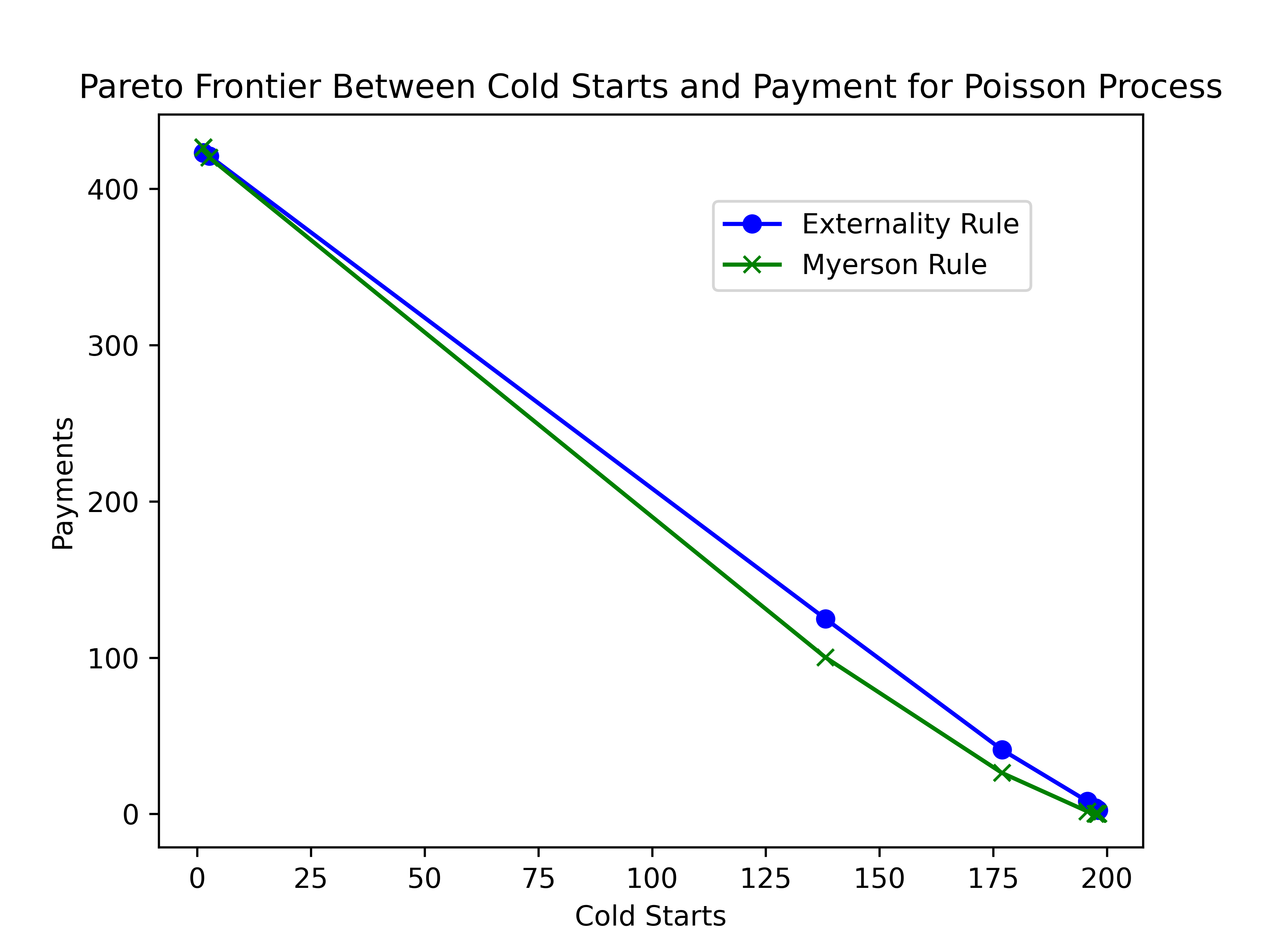}
        \caption{Poisson process $\lambda = 0.436$}
    \end{subfigure}%
    ~ 
    \begin{subfigure}[t]{0.44\textwidth}
        \centering
        \includegraphics[height=1.7in]{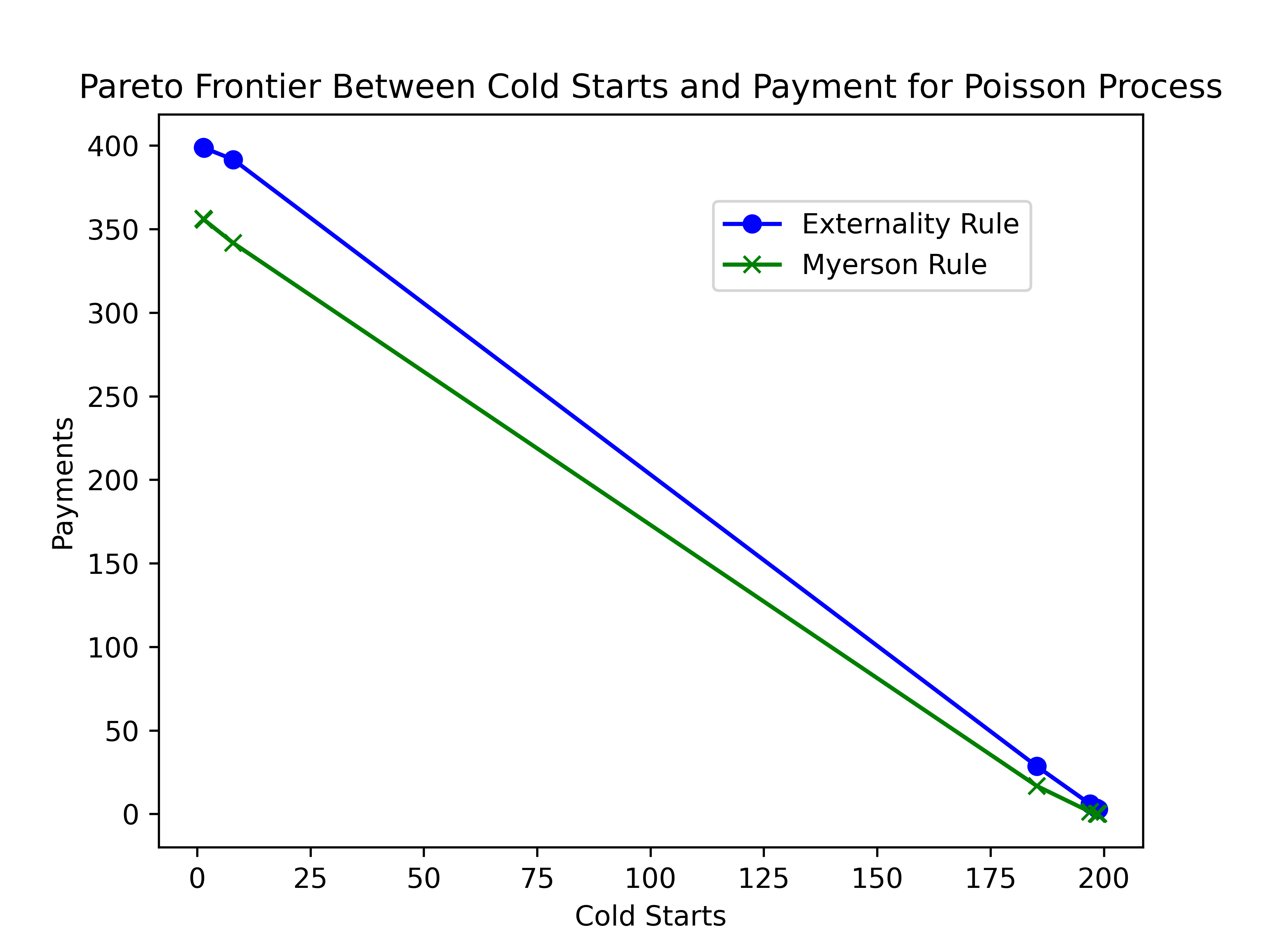}
        \caption{Poisson process $\lambda = 0.485$}
    \end{subfigure}    
     
    \begin{subfigure}[t]{0.44\textwidth}
        \centering
        \includegraphics[height=1.7in]{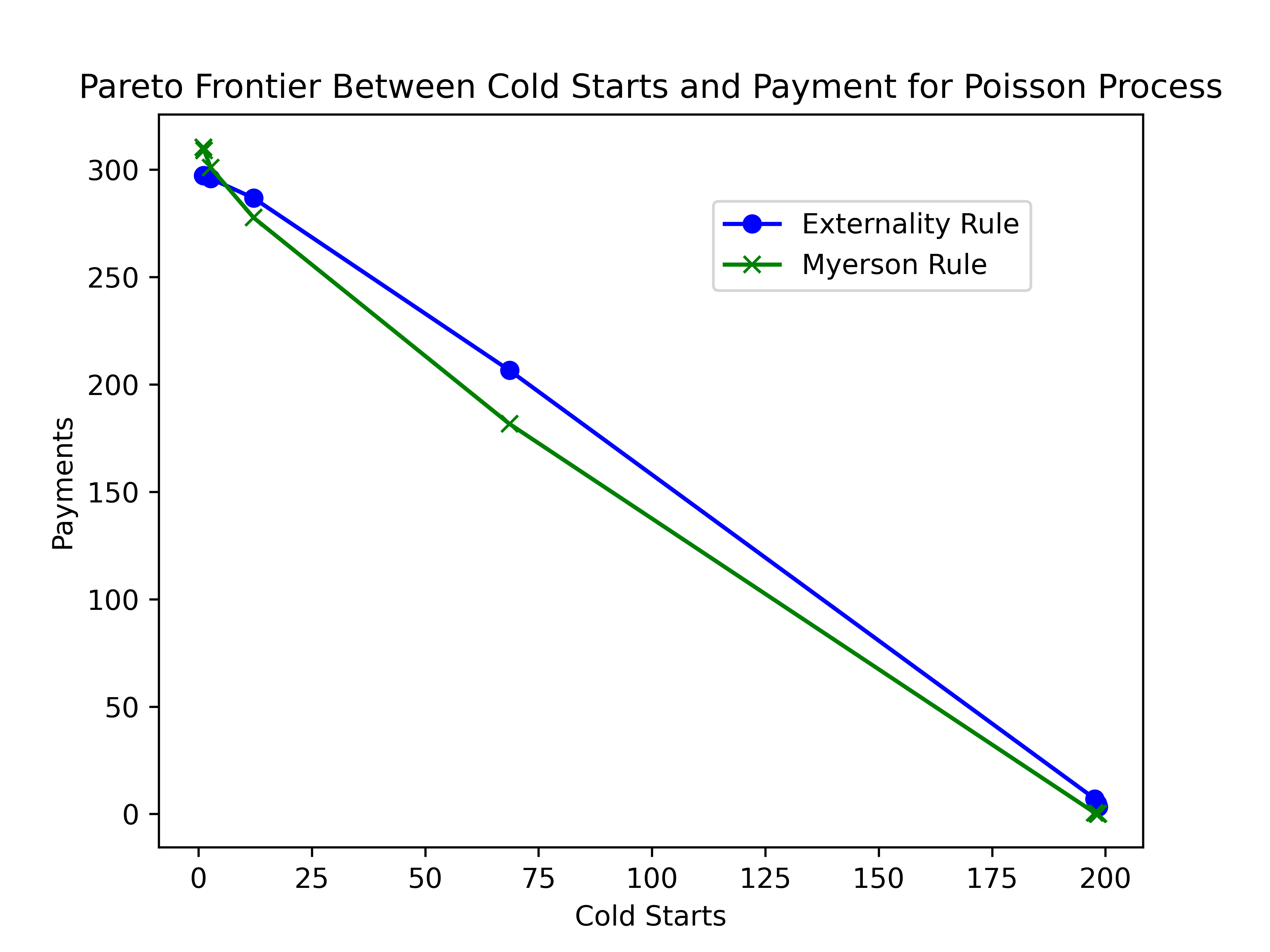}
        \caption{Poisson process $\lambda = 0.612$}
    \end{subfigure}        
    ~
    \begin{subfigure}[t]{0.44\textwidth}
        \centering
        \includegraphics[height=1.7in]{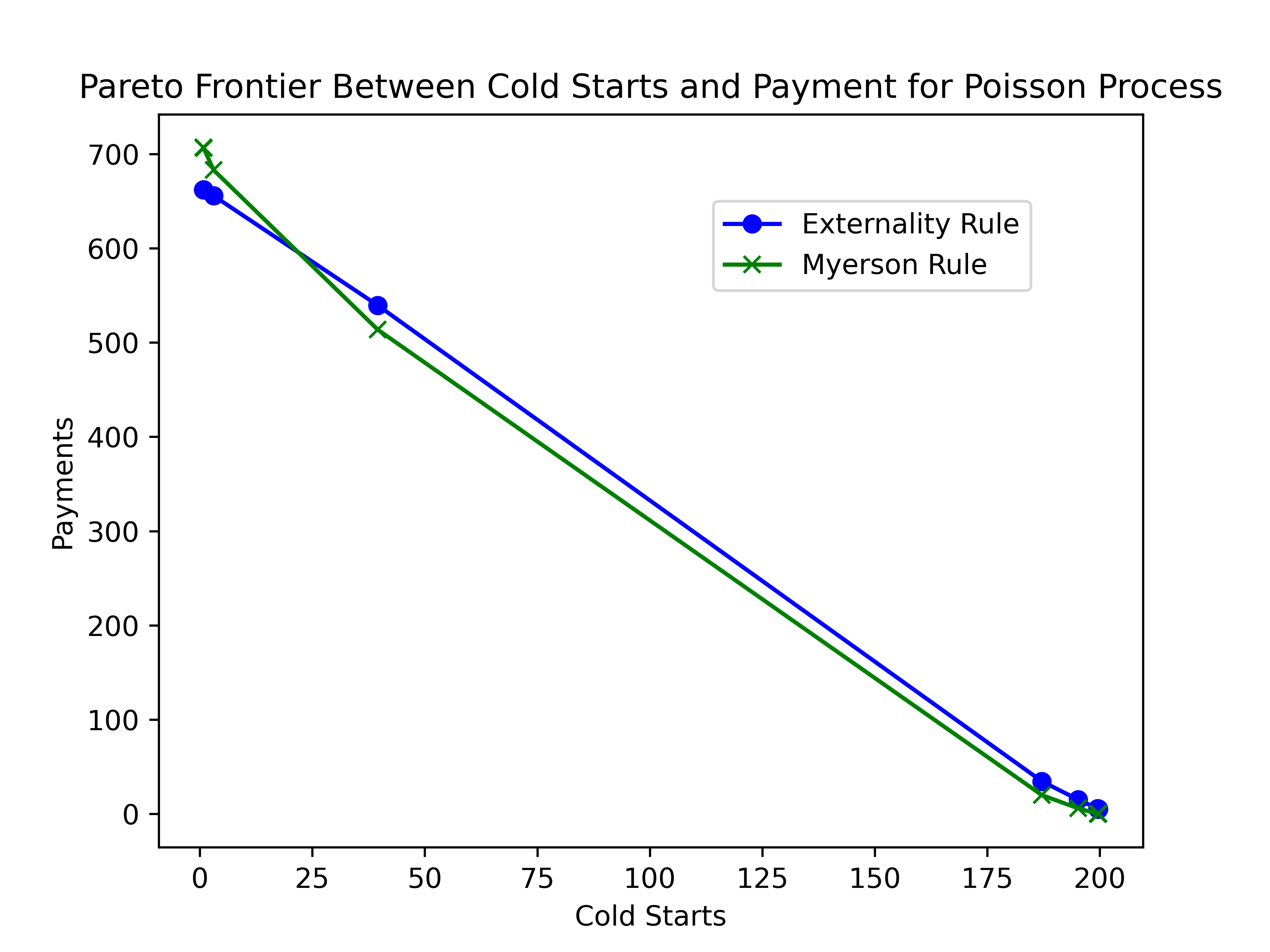}
        \caption{Poisson process $\lambda = 0.319$}
    \end{subfigure}    
     
    \begin{subfigure}[t]{0.44\textwidth}
        \centering
        \includegraphics[height=1.7in]{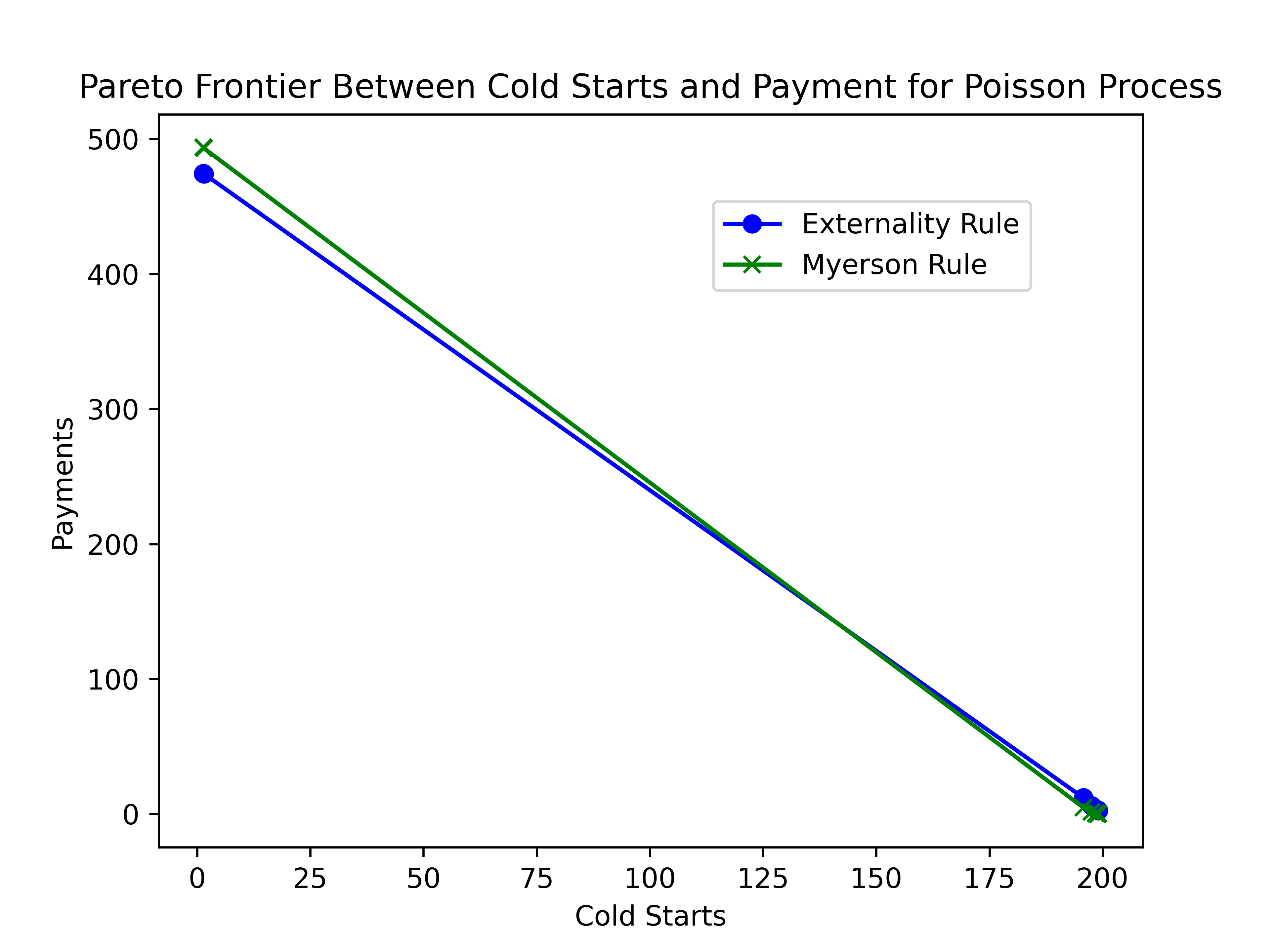}
        \caption{Poisson process $\lambda = 0.443$}
    \end{subfigure}    
    ~ 
    \begin{subfigure}[t]{0.44\textwidth}
        \centering
        \includegraphics[height=1.7in]{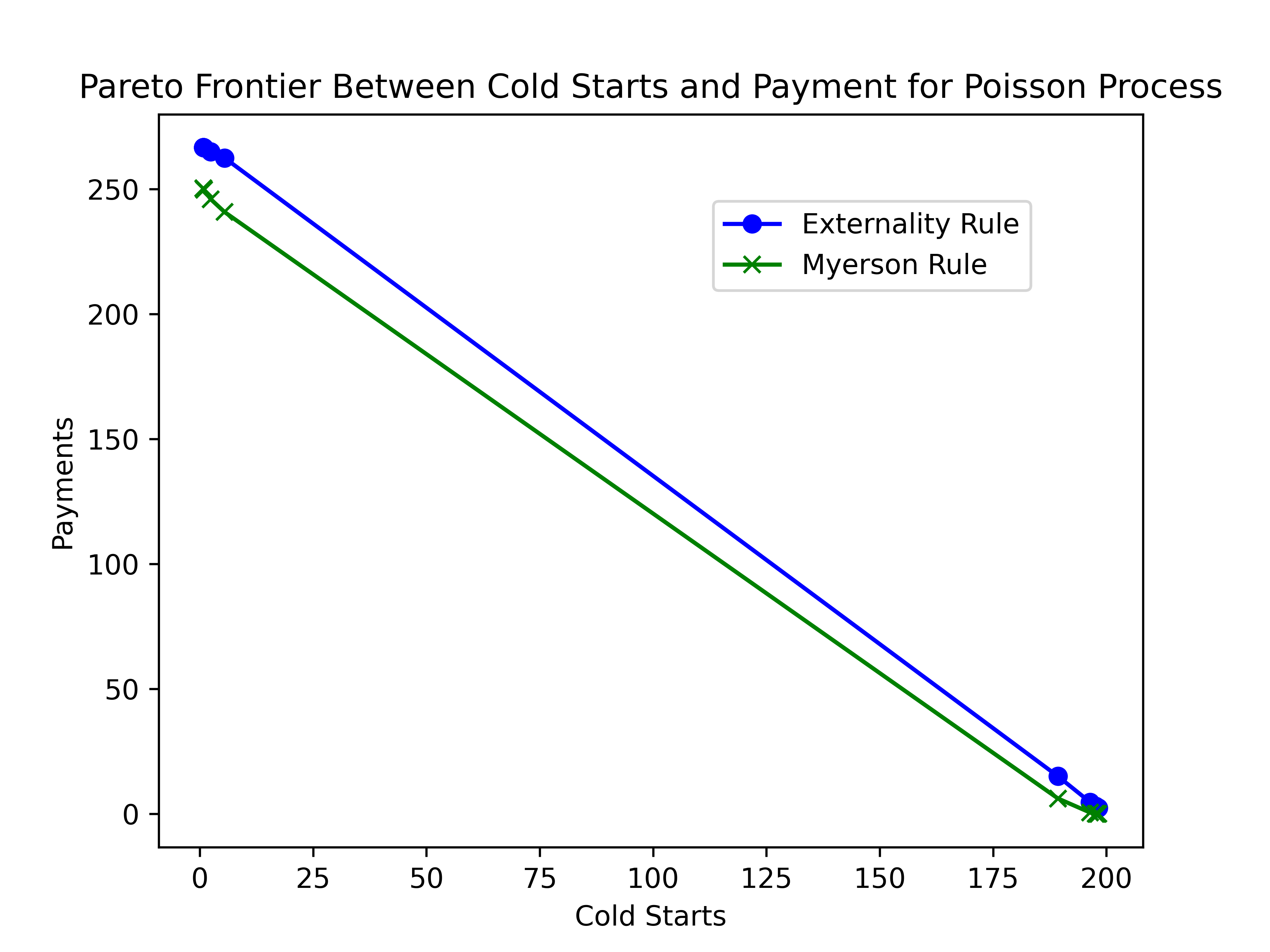}
        \caption{Poisson process $\lambda = 0.718$}
    \end{subfigure}            

    \begin{subfigure}[t]{0.44\textwidth}
        \centering
        \includegraphics[height=1.7in]{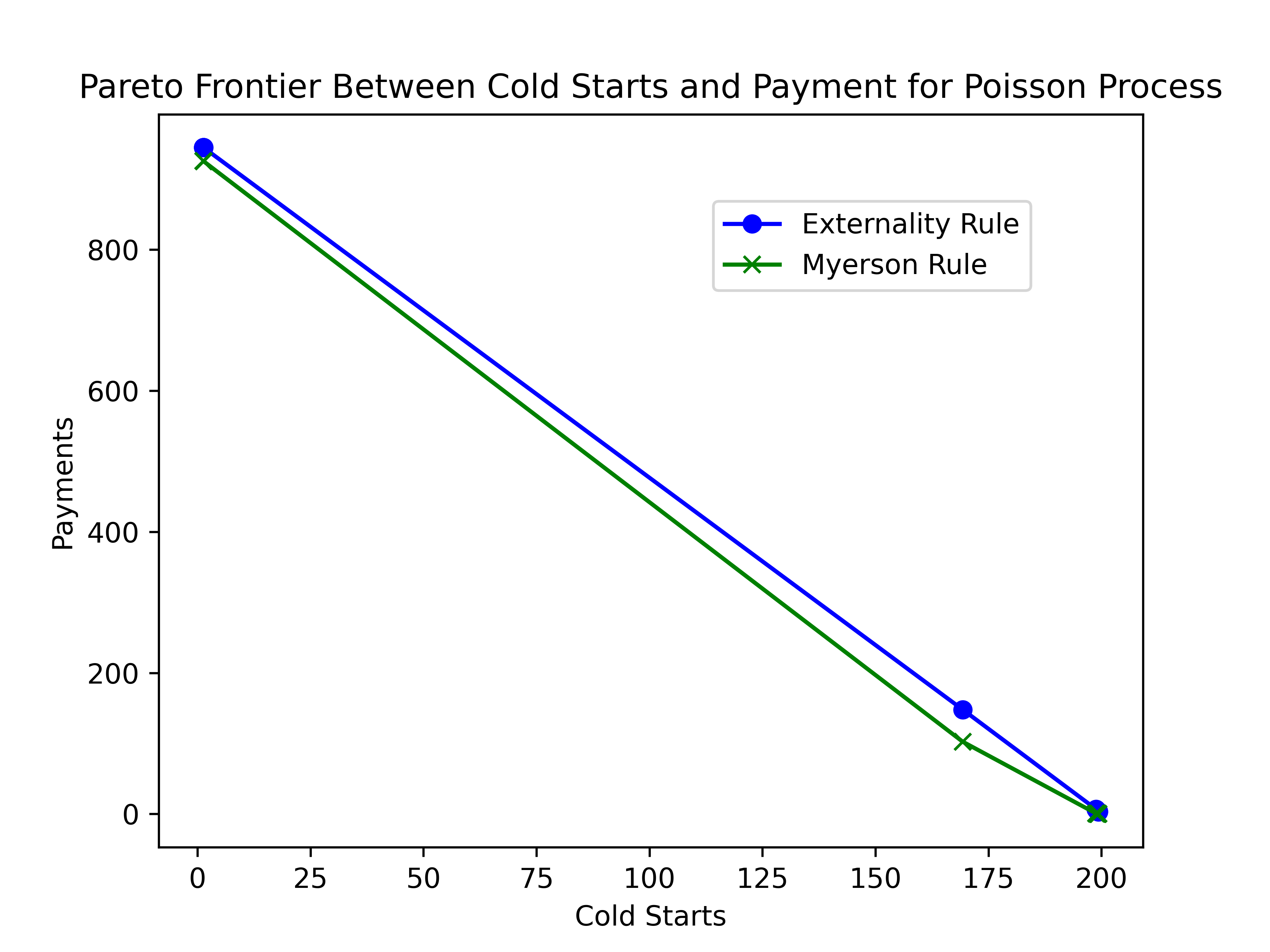}
        \caption{Poisson process $\lambda = 0.215$}
    \end{subfigure}    
    ~ 
    \begin{subfigure}[t]{0.44\textwidth}
        \centering
        \includegraphics[height=1.7in]{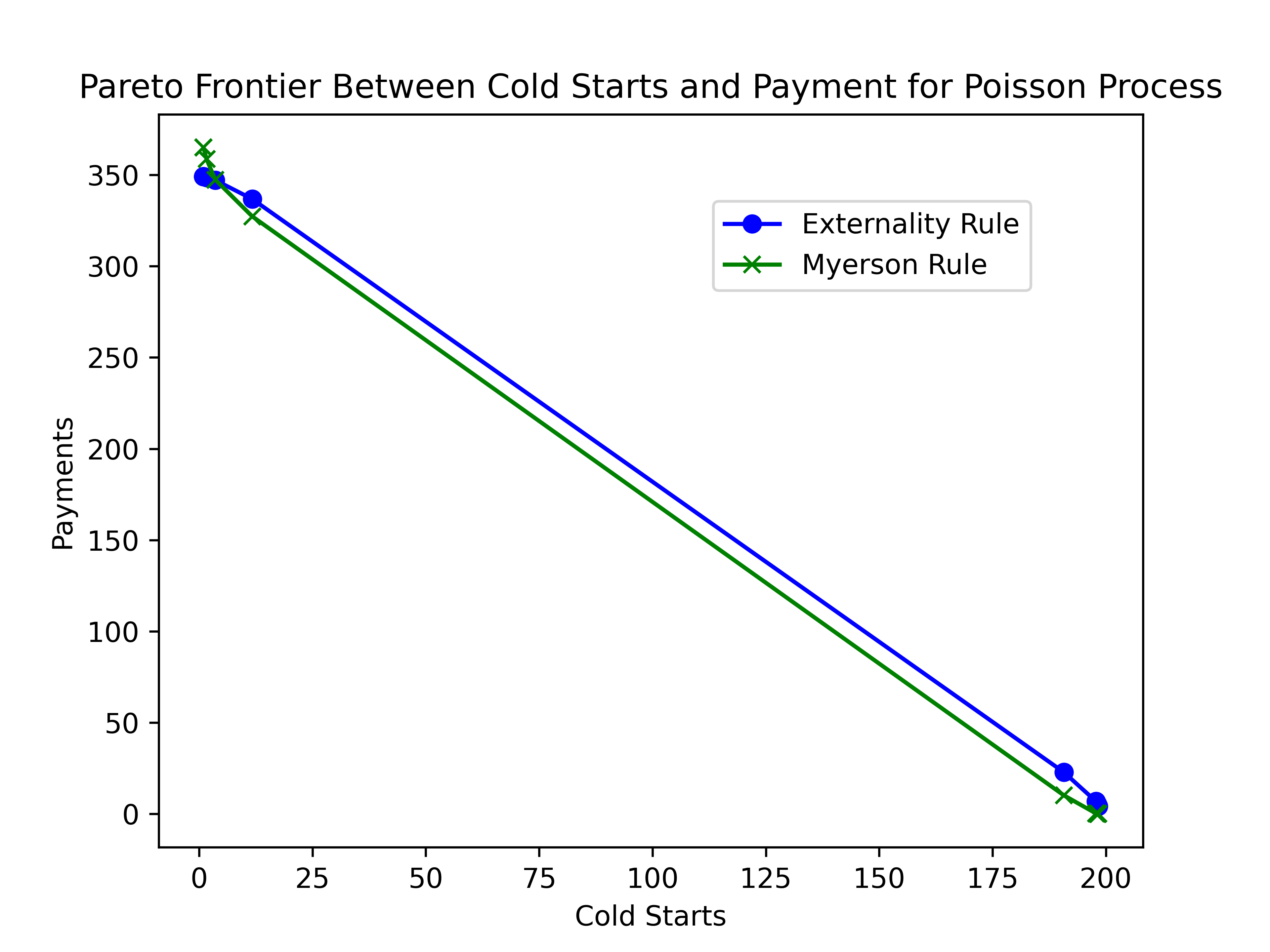}
        \caption{Poisson process $\lambda = 0.576$}
    \end{subfigure} 
   \caption{Trade-off curve of cumulative payments vs cold starts for Poisson process (1-8)}
    \label{poisson_figure1}
\end{figure*}

\begin{figure*}[t!]
    \centering           
    \begin{subfigure}[t]{0.44\textwidth}
        \centering
        \includegraphics[height=1.73in]{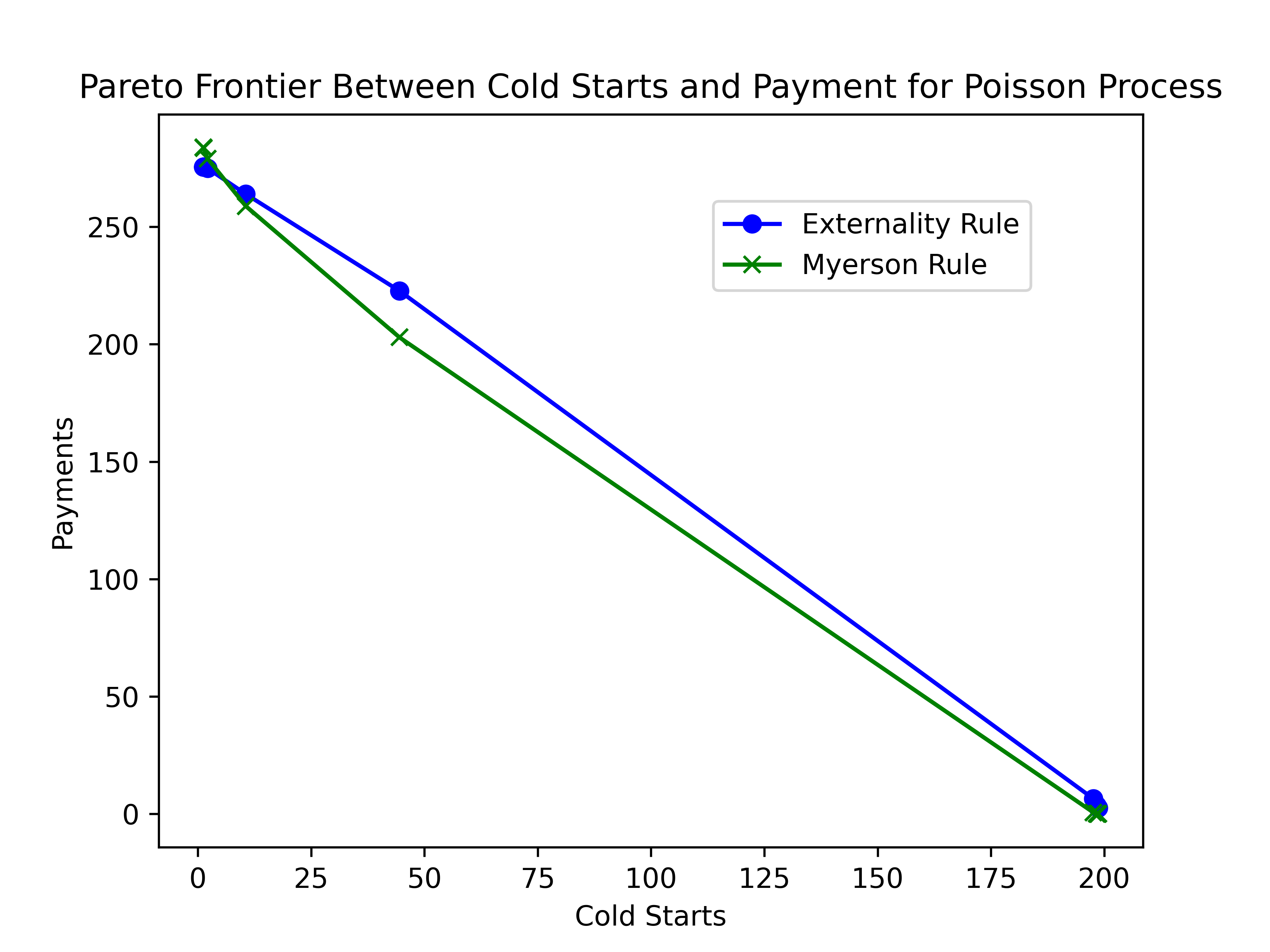}
        \caption{Poisson process $\lambda = 0.669$}
    \end{subfigure}%
    ~ 
    \begin{subfigure}[t]{0.44\textwidth}
        \centering
        \includegraphics[height=1.73in]{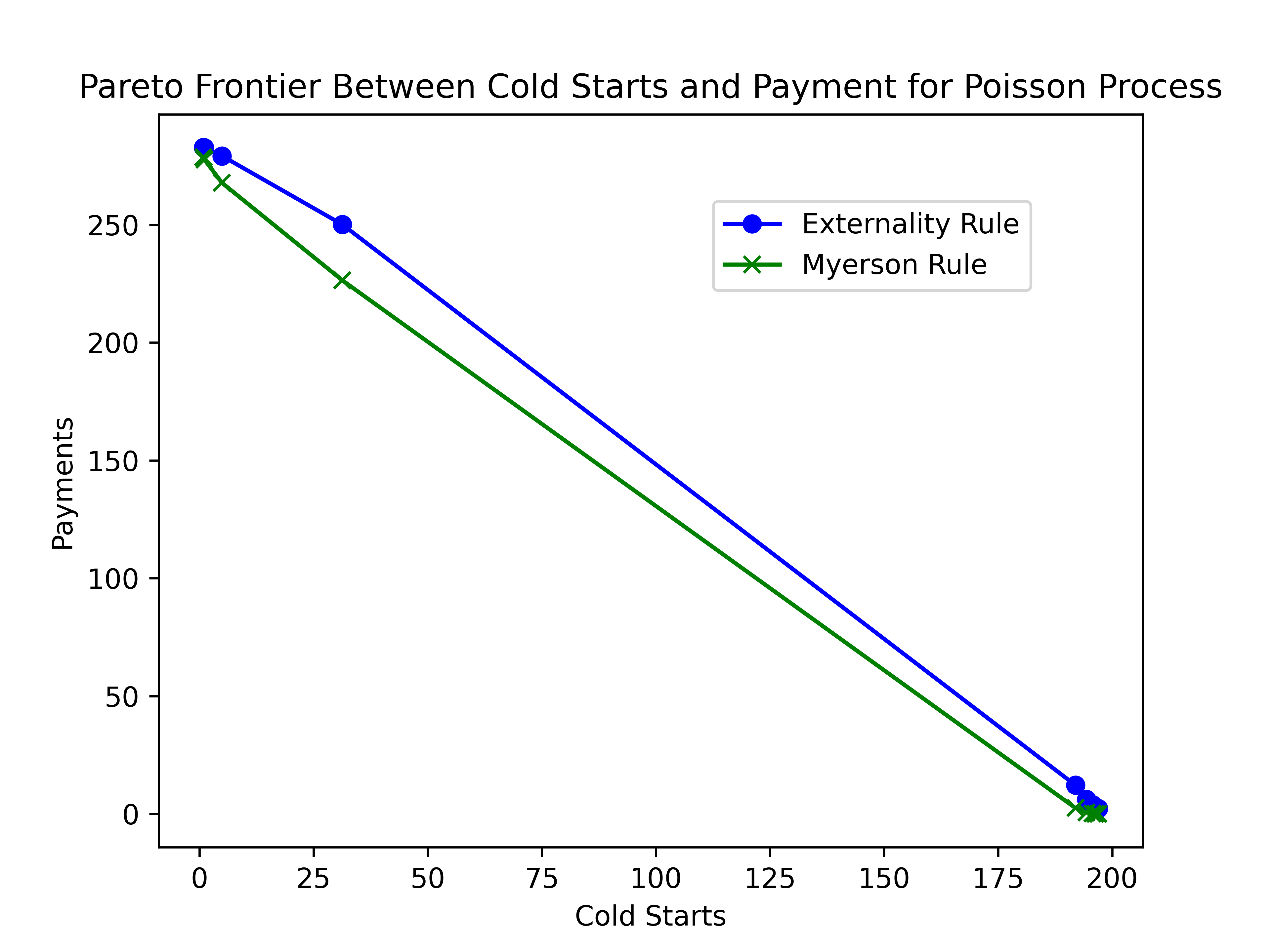}
        \caption{Poisson process $\lambda = 0.742$}
    \end{subfigure}    
     
    \begin{subfigure}[t]{0.44\textwidth}
        \centering
        \includegraphics[height=1.73in]{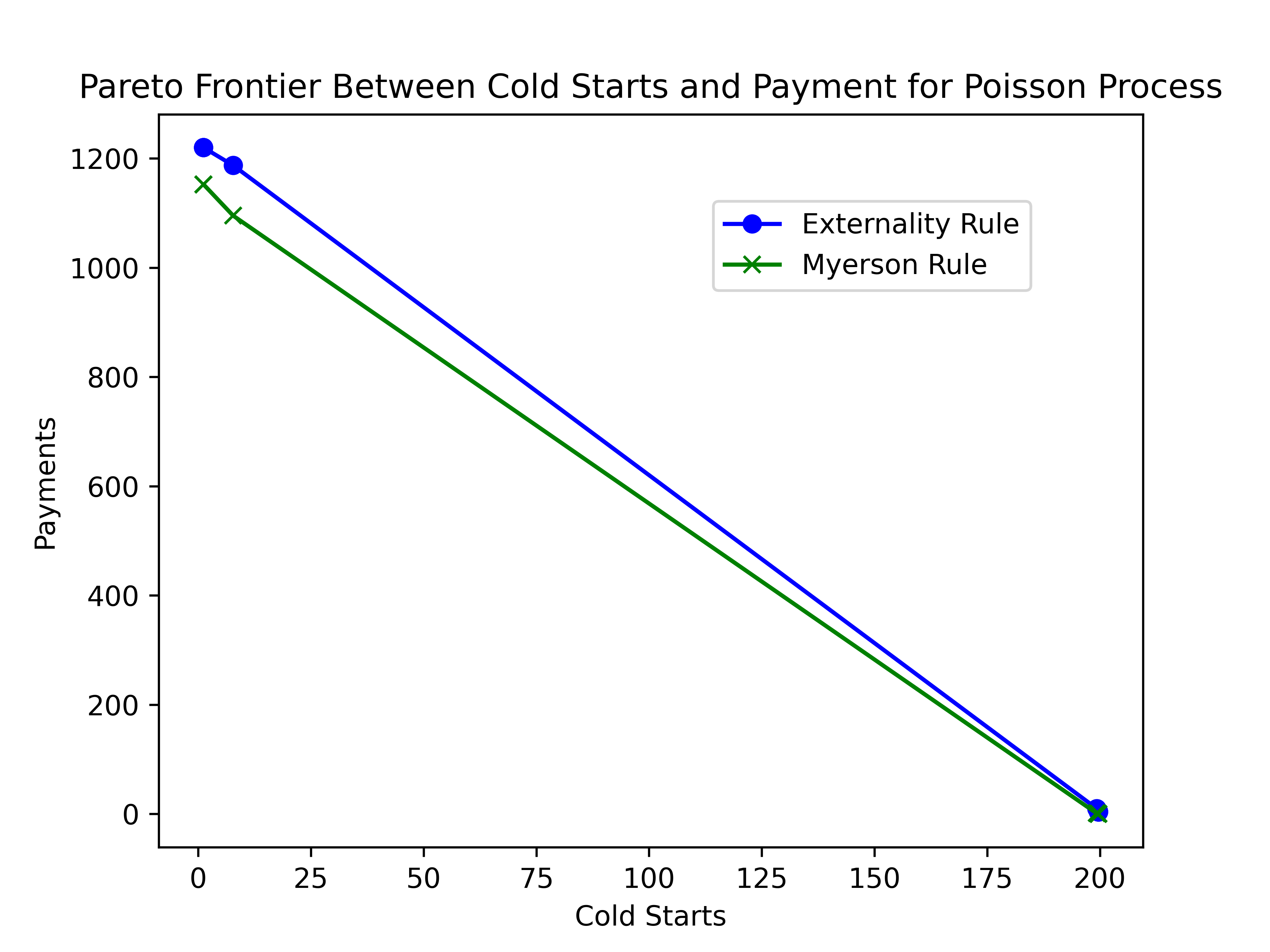}
        \caption{Poisson process $\lambda = 0.162$}
    \end{subfigure}        
    ~
    \begin{subfigure}[t]{0.44\textwidth}
        \centering
        \includegraphics[height=1.73in]{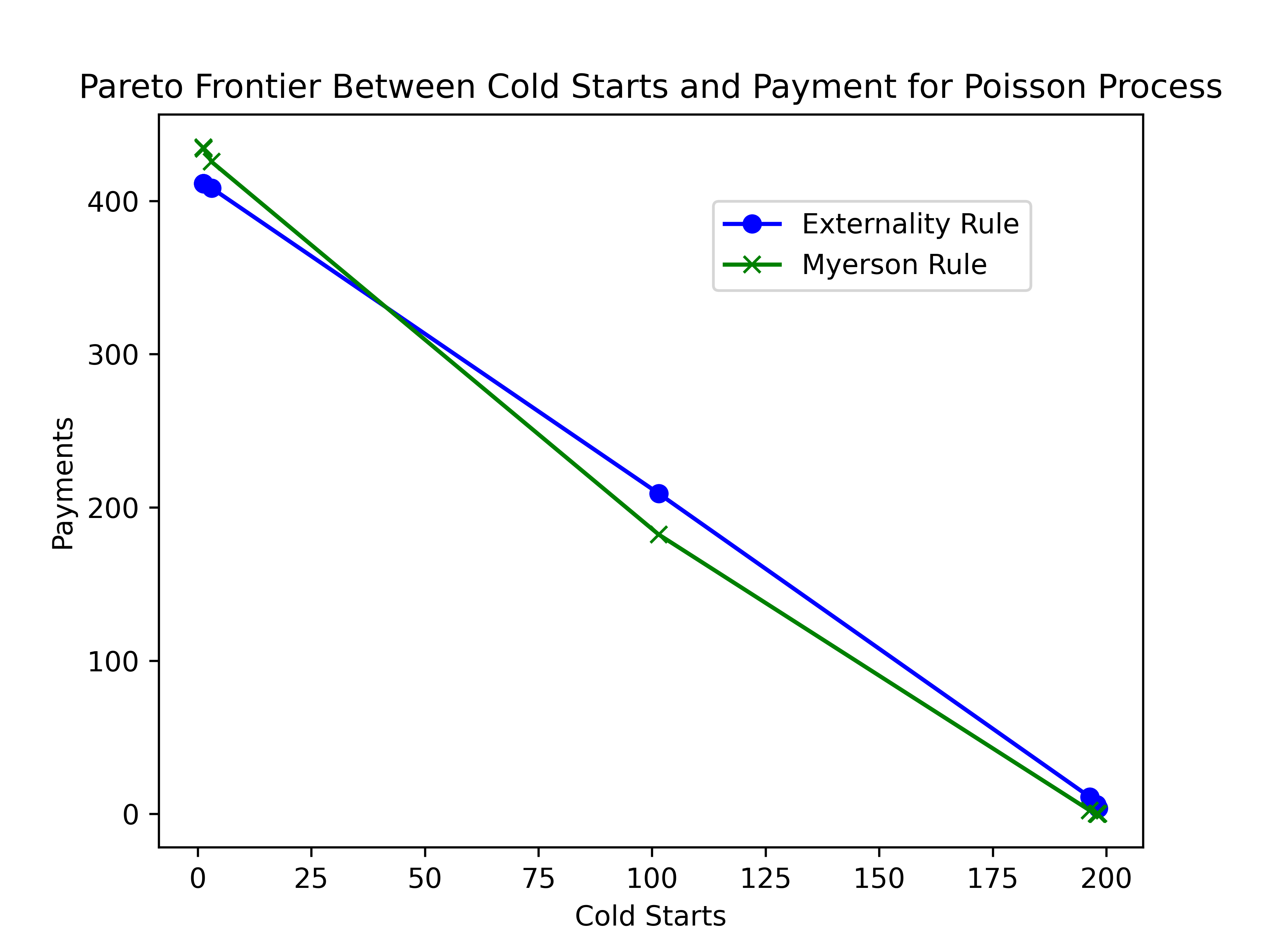}
        \caption{Poisson process $\lambda = 0.439$}
    \end{subfigure}   

    \begin{subfigure}[t]{0.44\textwidth}
        \centering
        \includegraphics[height=1.73in]{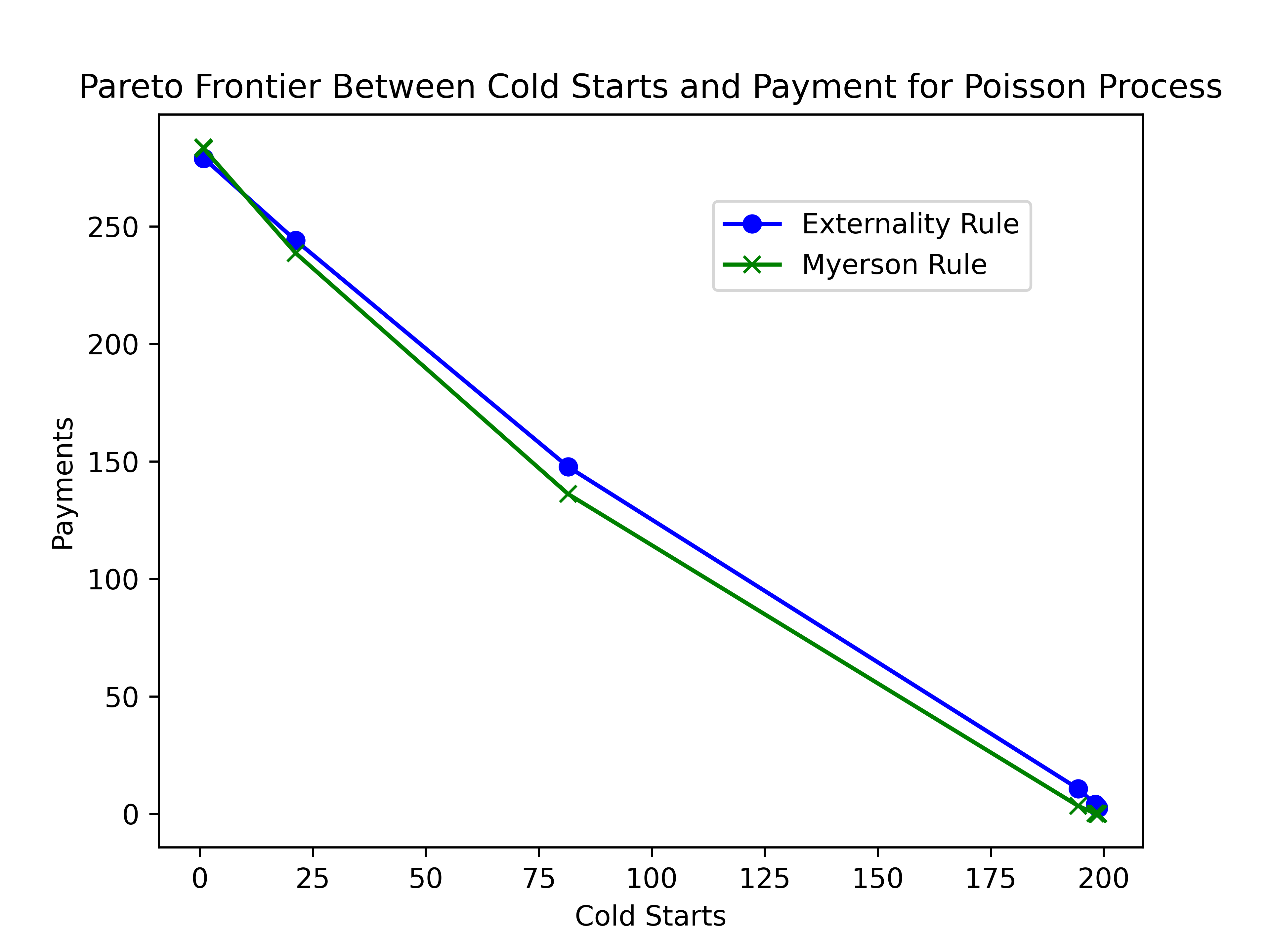}
        \caption{Poisson process $\lambda = 0.725$}
    \end{subfigure}    
    ~ 
    \begin{subfigure}[t]{0.44\textwidth}
        \centering
        \includegraphics[height=1.73in]{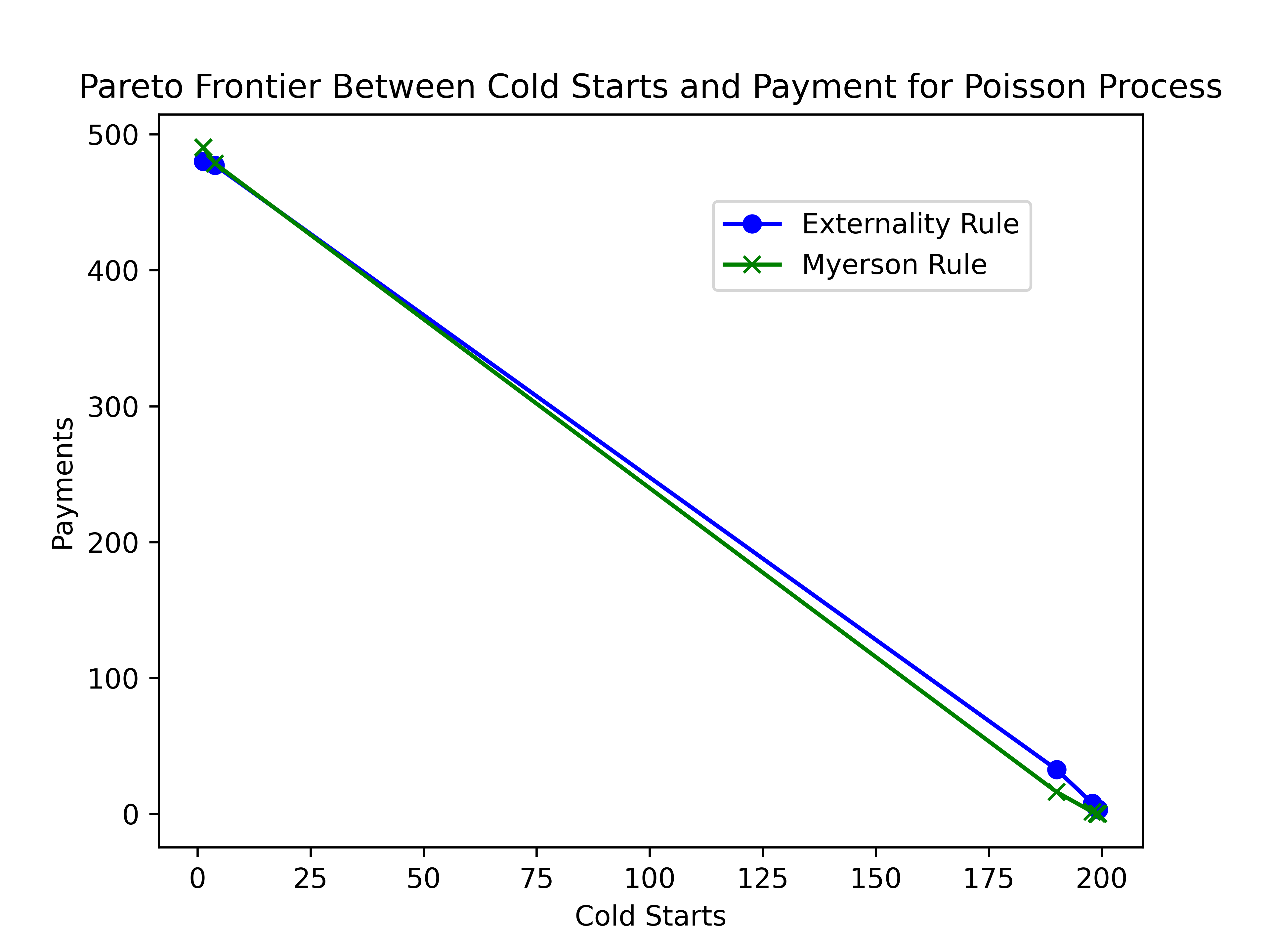}
        \caption{Poisson process $\lambda = 0.413$}
    \end{subfigure}            

    \begin{subfigure}[t]{0.44\textwidth}
        \centering
        \includegraphics[height=1.73in]{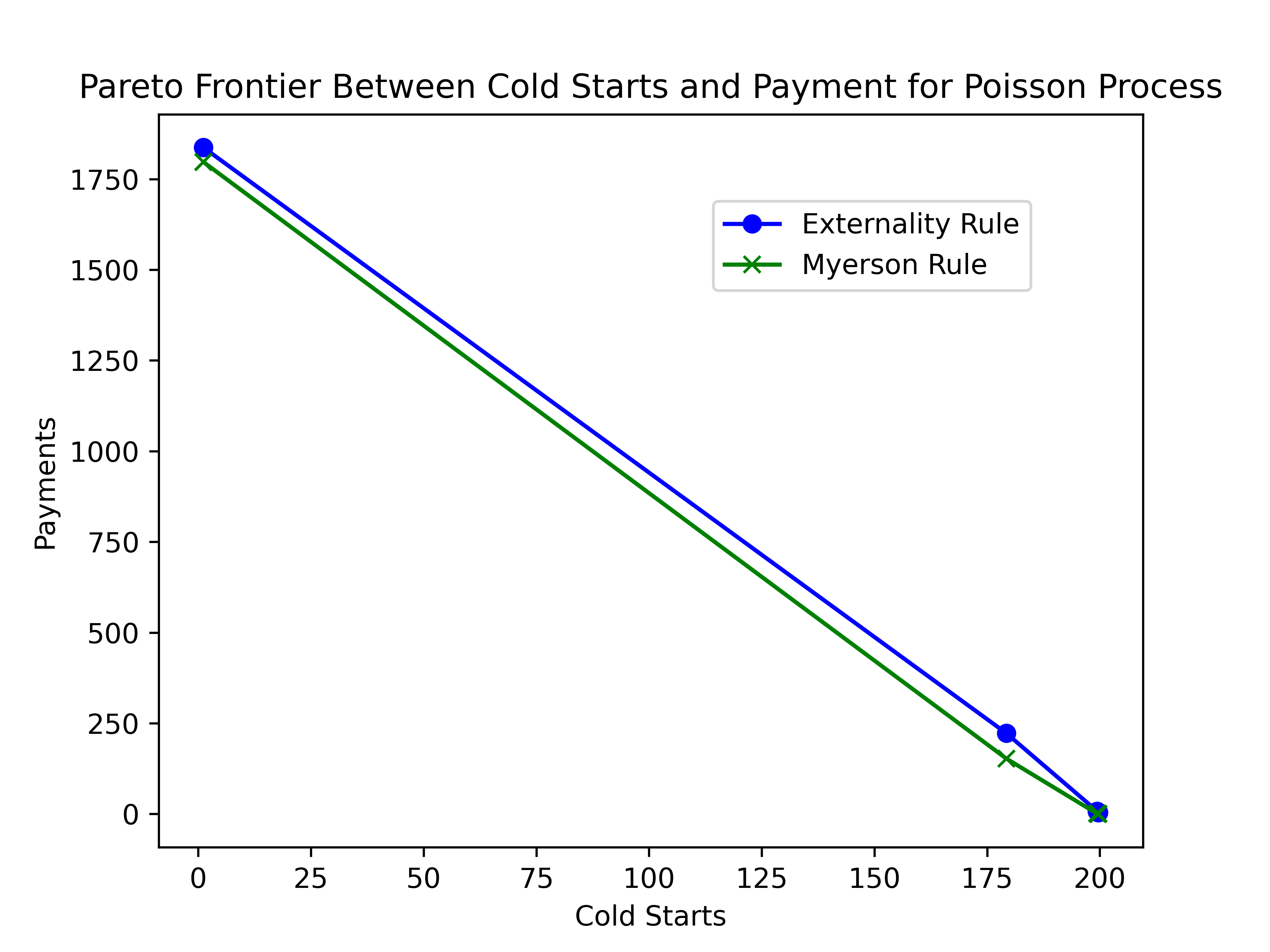}
        \caption{Poisson process $\lambda = 0.100$}
    \end{subfigure}    
    ~ 
    \begin{subfigure}[t]{0.44\textwidth}
        \centering
        \includegraphics[height=1.73in]{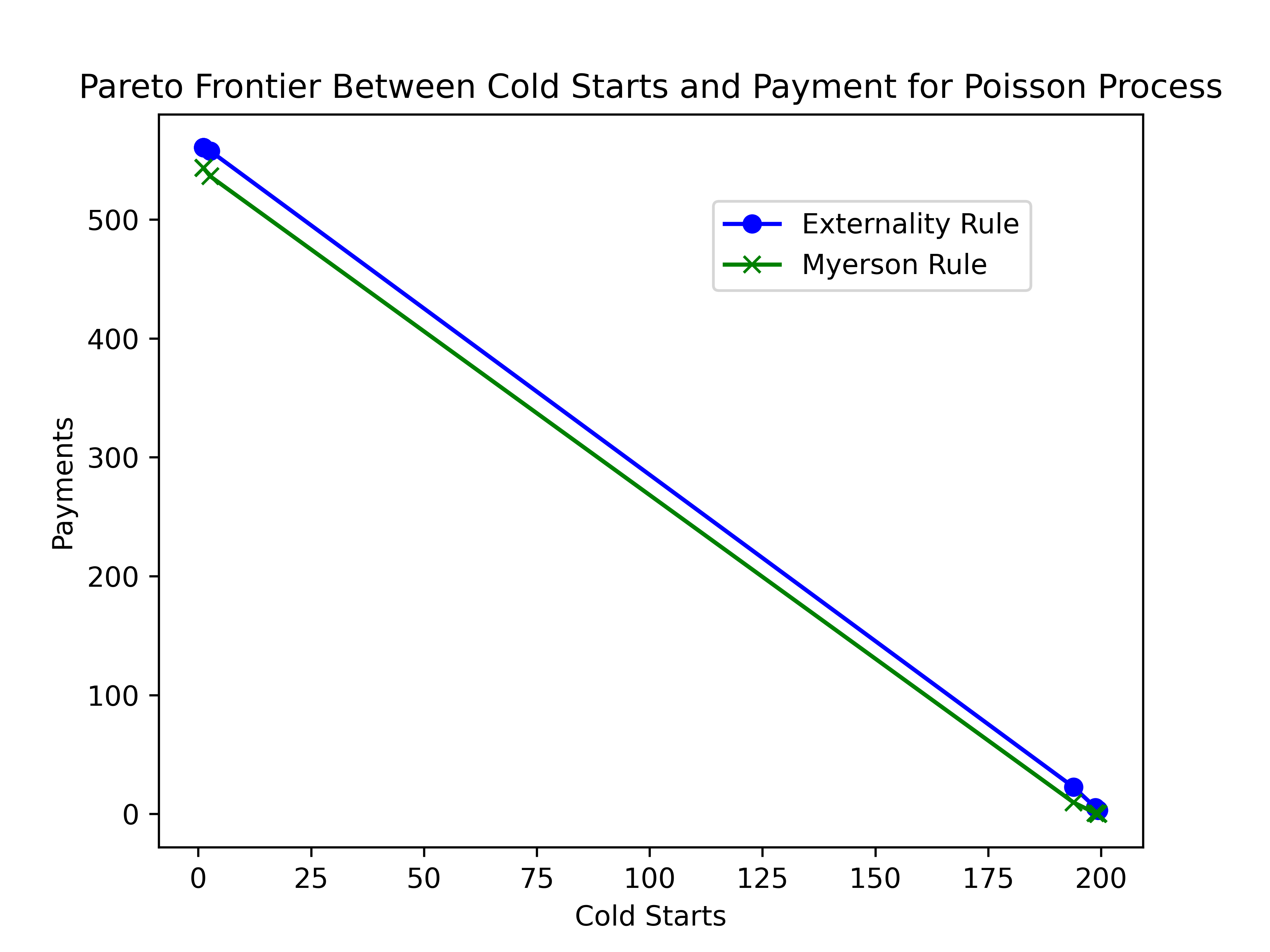}
        \caption{Poisson process $\lambda = 0.374$}
    \end{subfigure}            

   \caption{Trade-off curve of cumulative payments vs cold starts for Poisson process (9-16)}
    \label{poisson_figure2}
\end{figure*}

\begin{figure*}[t!]
    \centering
    \begin{subfigure}[t]{0.44\textwidth}
        \centering
        \includegraphics[height=1.73in]{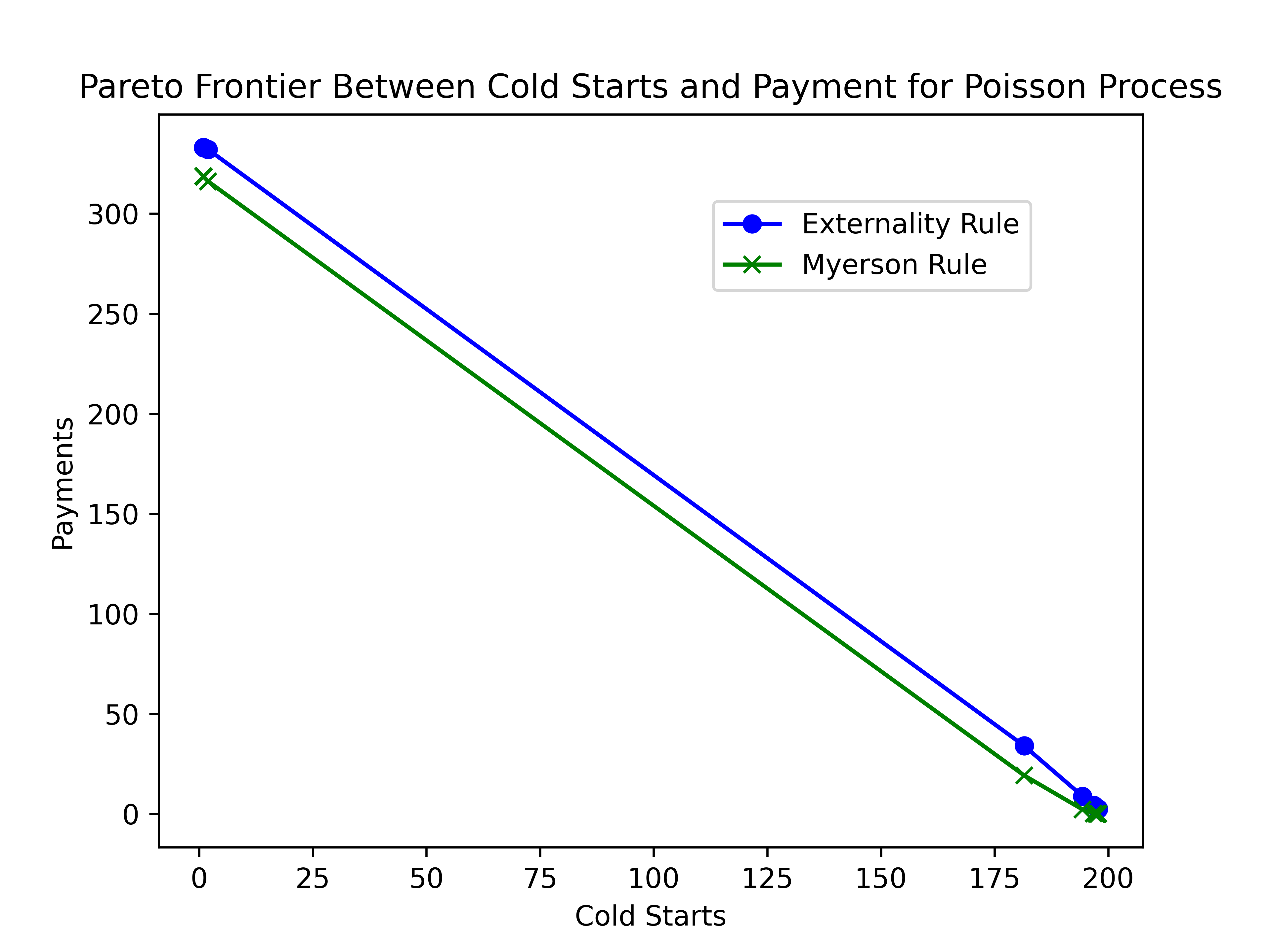}
        \caption{Poisson process $\lambda = 0.584$}
    \end{subfigure}%
    ~ 
    \begin{subfigure}[t]{0.44\textwidth}
        \centering
        \includegraphics[height=1.73in]{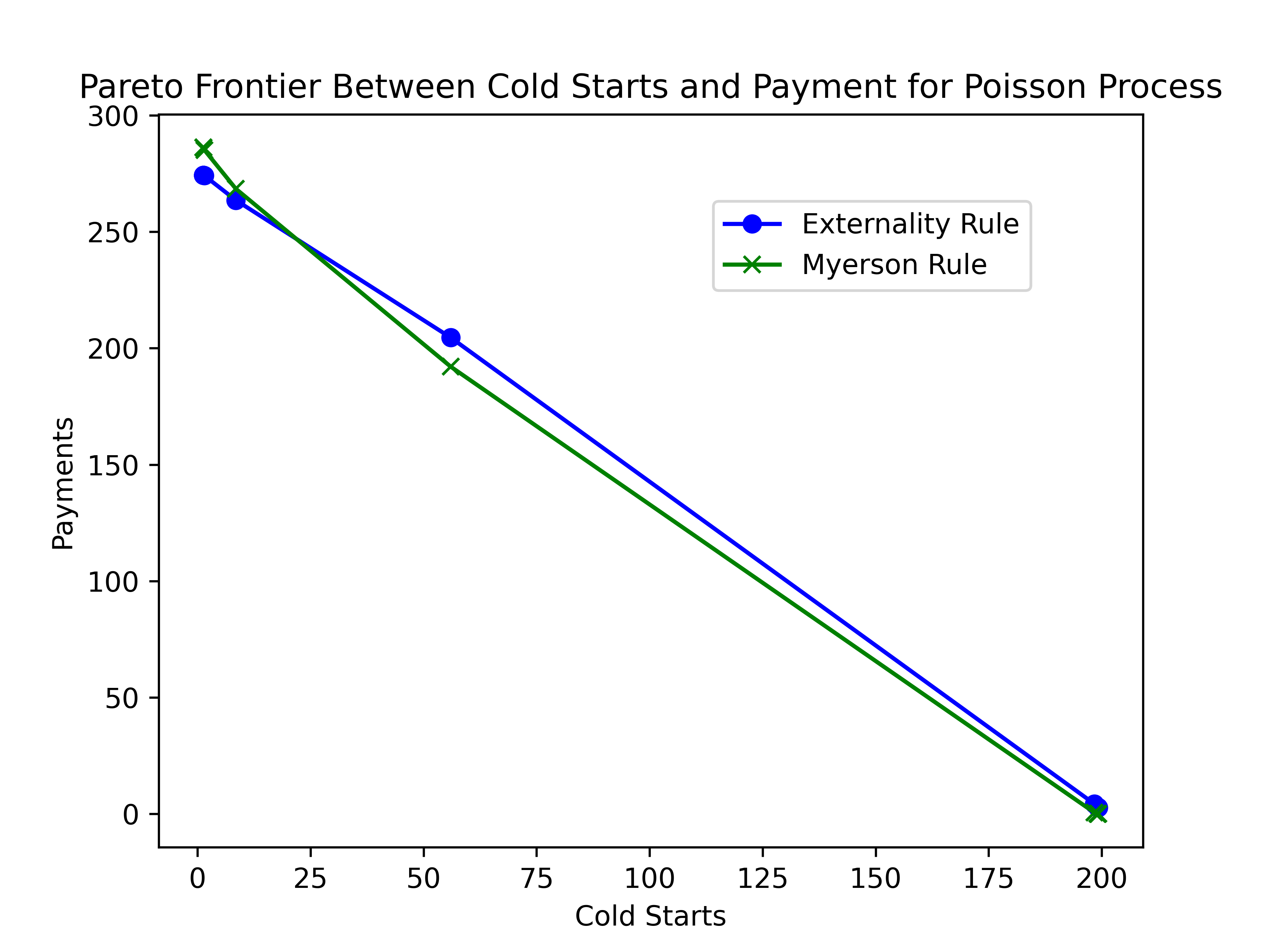}
        \caption{Poisson process $\lambda = 0.671$}
    \end{subfigure}    
    
    \begin{subfigure}[t]{0.8\textwidth}
        \centering
        \includegraphics[height=2.03in]{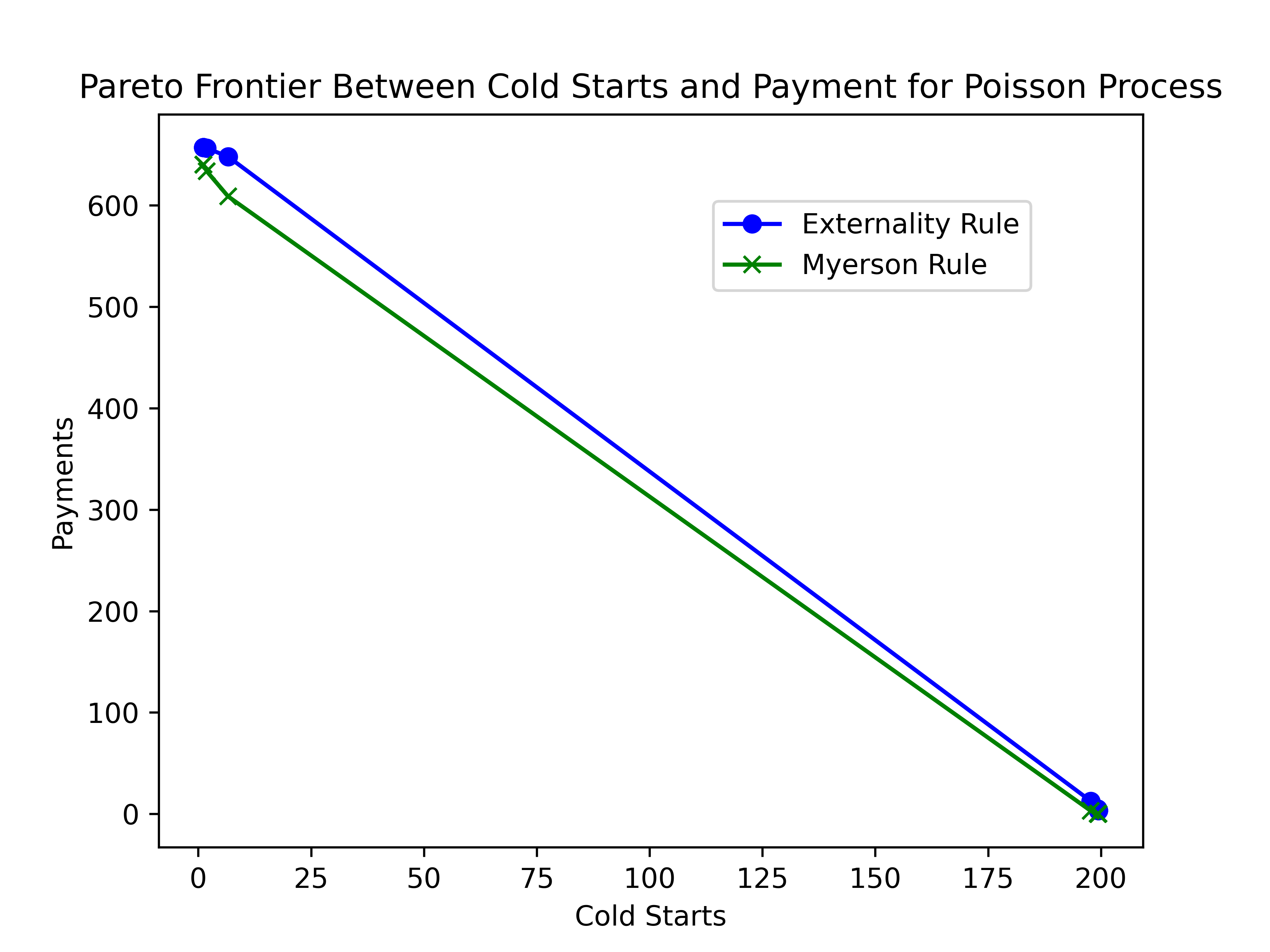}
        \caption{Poisson process $\lambda = 0.323$}
    \end{subfigure}%
    \caption{Trade-off curve of cumulative payments vs cold starts for Poisson process (17-19)}
    \label{poisson_figure3}
\end{figure*}

\begin{figure*}[t!]
    \centering
    \begin{subfigure}[t]{0.44\textwidth}
        \centering
        \includegraphics[height=1.73in]{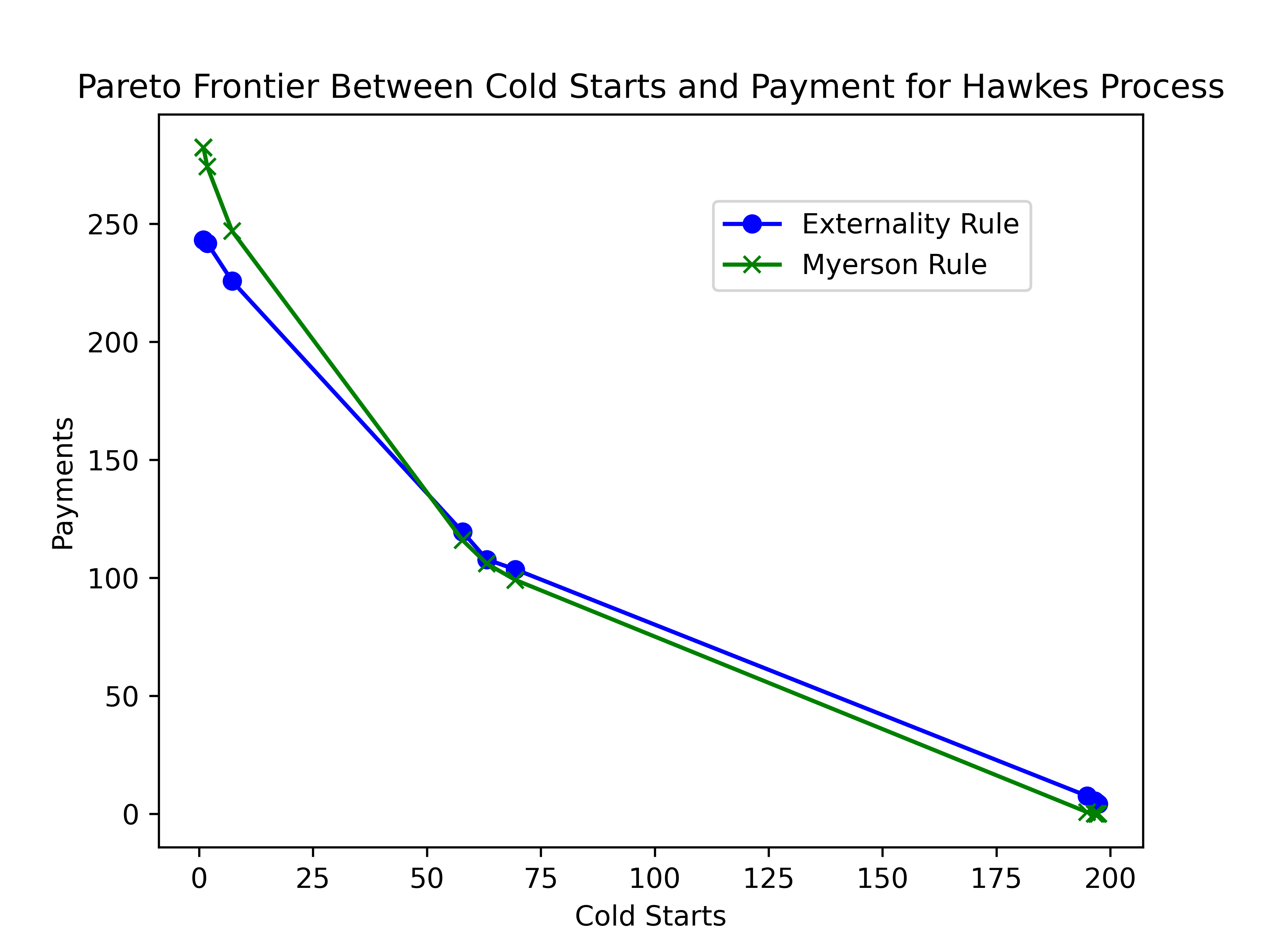}
        \caption{$\lambda_0 = 0.40, \alpha = 0.72, \beta = 1.75$}
    \end{subfigure}%
    ~ 
    \begin{subfigure}[t]{0.44\textwidth}
        \centering
        \includegraphics[height=1.73in]{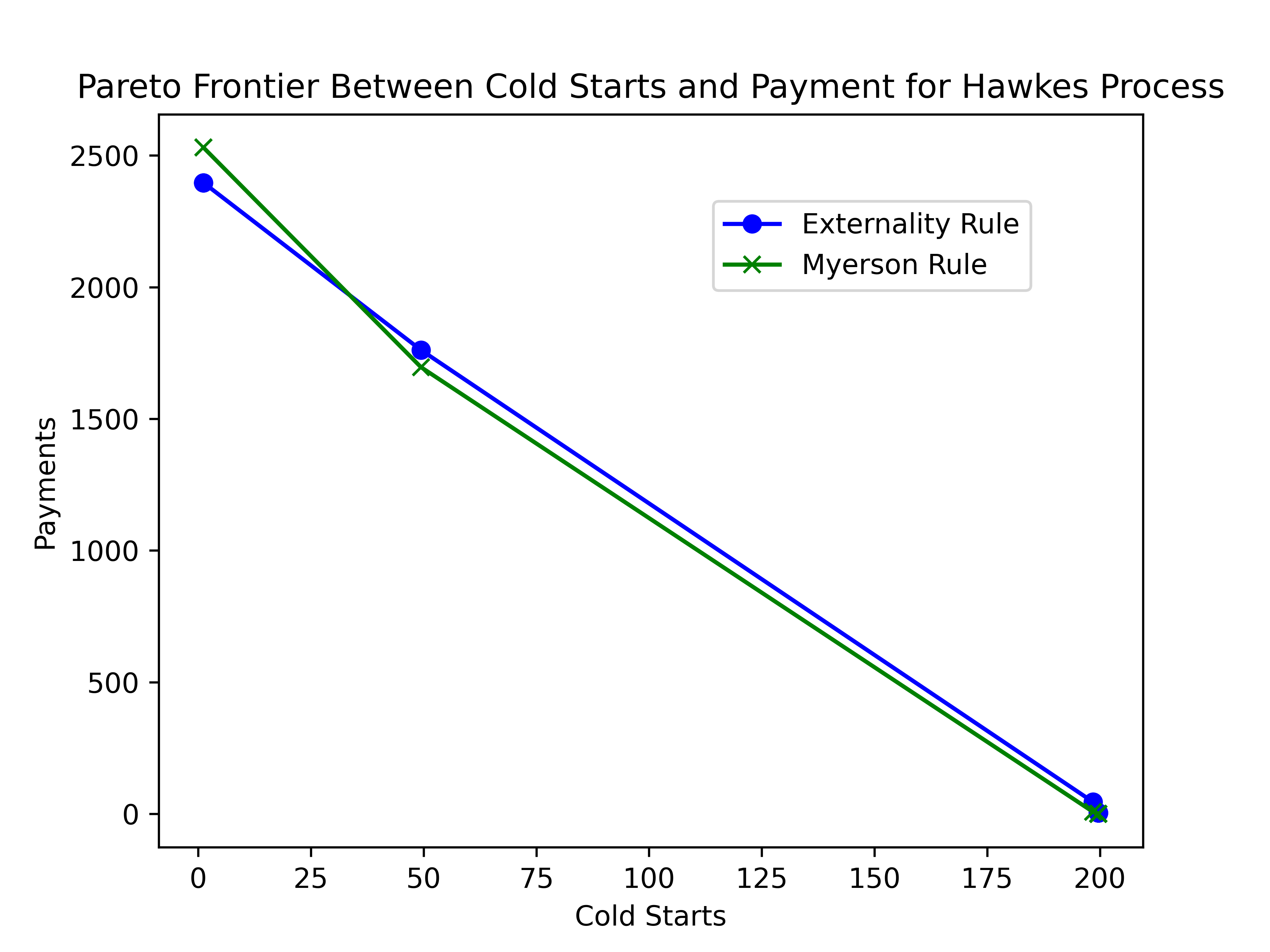}
        \caption{$ \lambda_0 = 0.09, \alpha = 0.029, \beta = 1.76$}
    \end{subfigure}    
     
    \begin{subfigure}[t]{0.44\textwidth}
        \centering
        \includegraphics[height=1.73in]{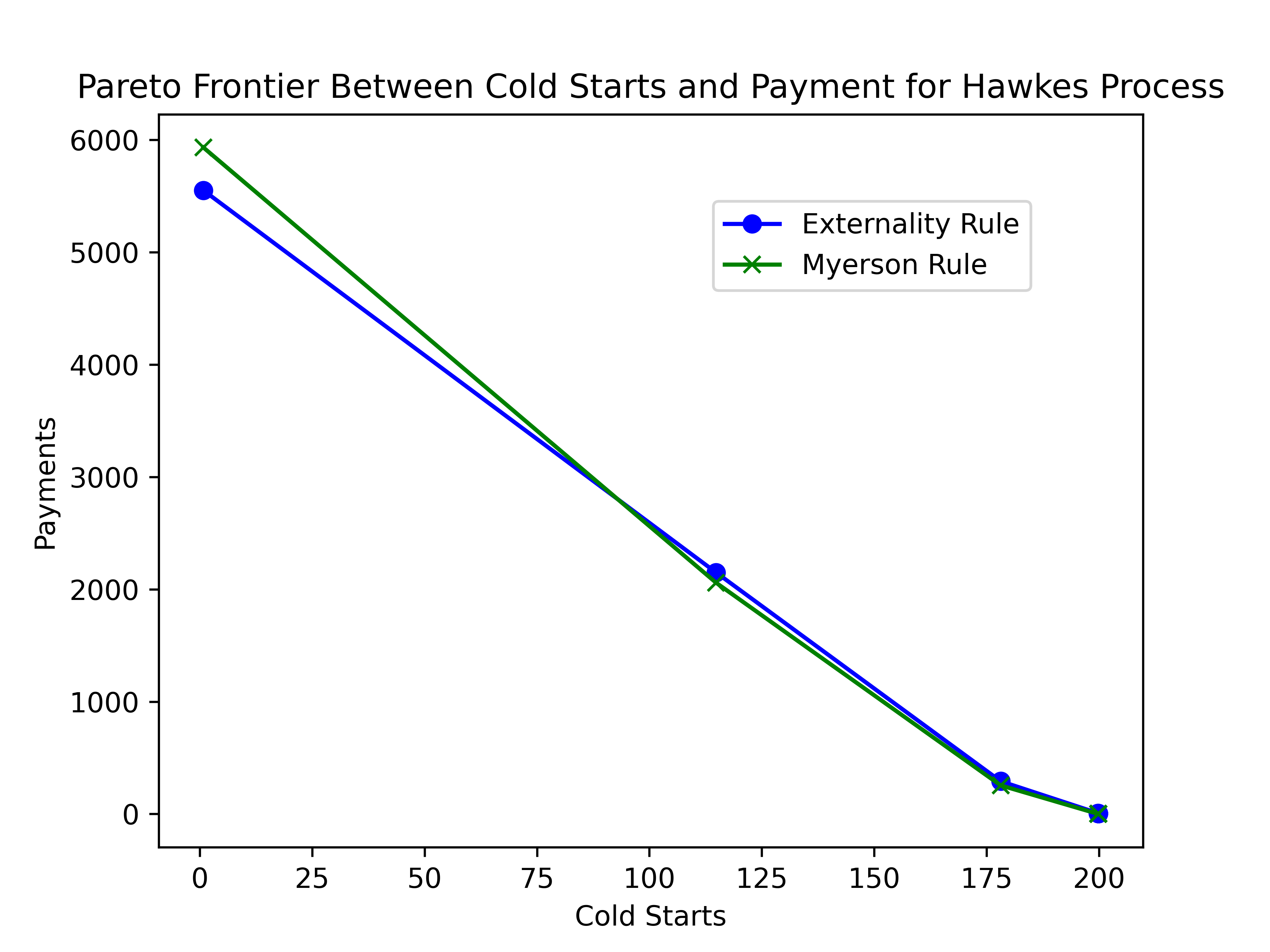}
        \caption{$\lambda_0 = 0.03, \alpha = 0.06, \beta = 2.45$}
    \end{subfigure}        
    ~
    \begin{subfigure}[t]{0.44\textwidth}
        \centering
        \includegraphics[height=1.73in]{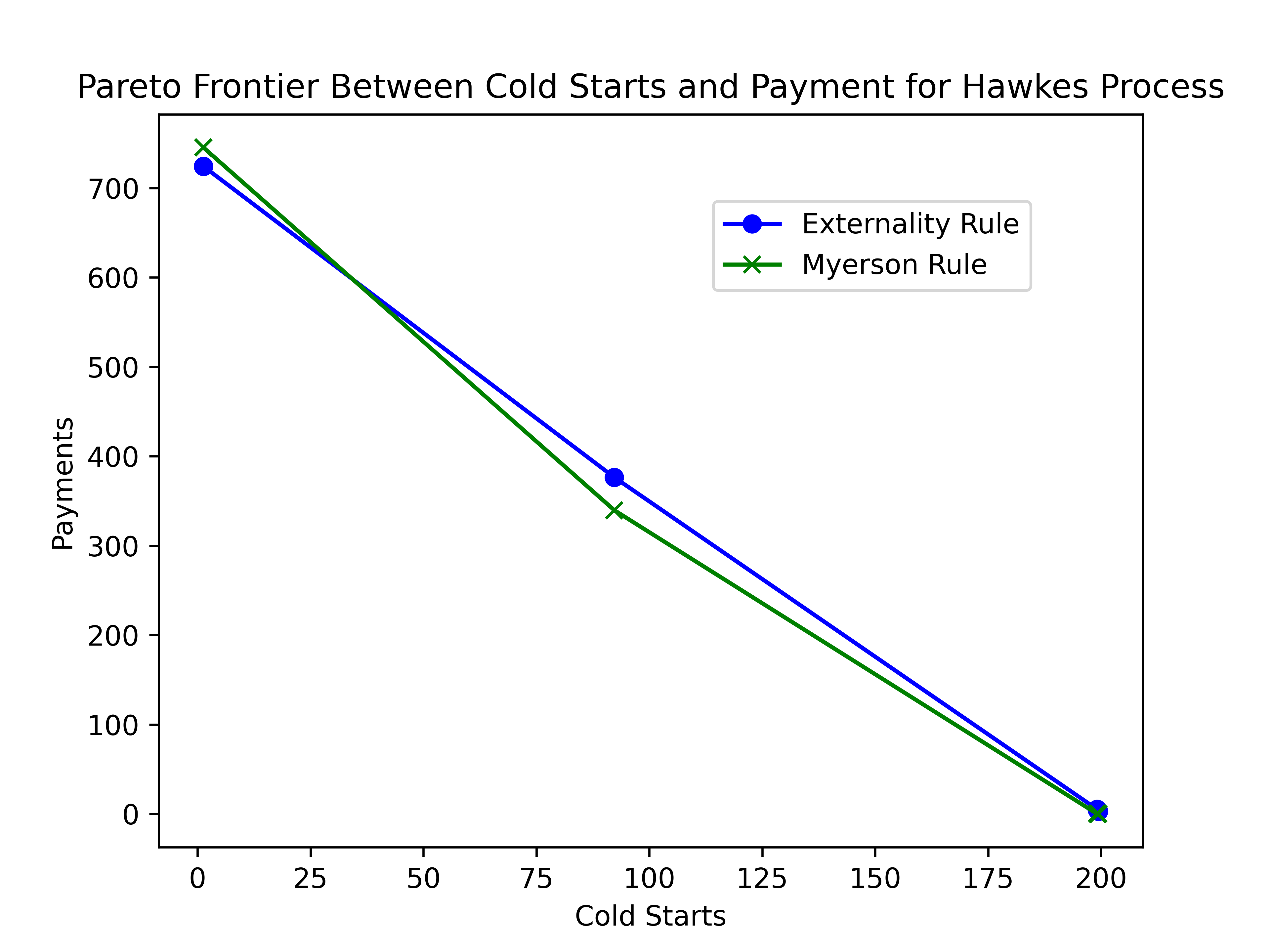}
        \caption{$ \lambda_0 = 0.22, \alpha = 0.54, \beta = 2.92$}
    \end{subfigure}    
     
    \begin{subfigure}[t]{0.44\textwidth}
        \centering
        \includegraphics[height=1.73in]{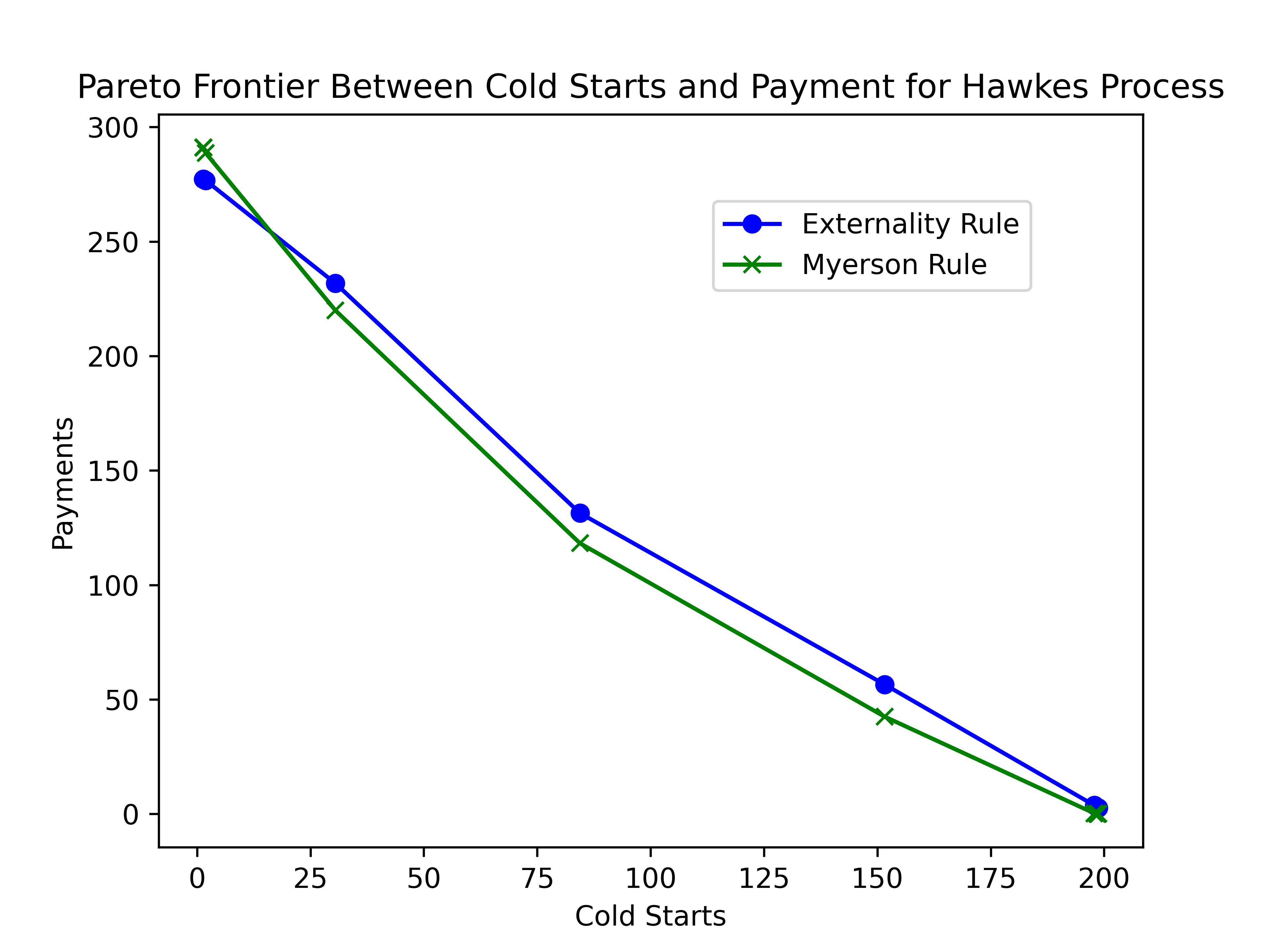}
        \caption{$ \lambda_0 = 0.63, \alpha = 0.53, \beta = 2.91$}
    \end{subfigure}    
    ~ 
    \begin{subfigure}[t]{0.44\textwidth}
        \centering
        \includegraphics[height=1.73in]{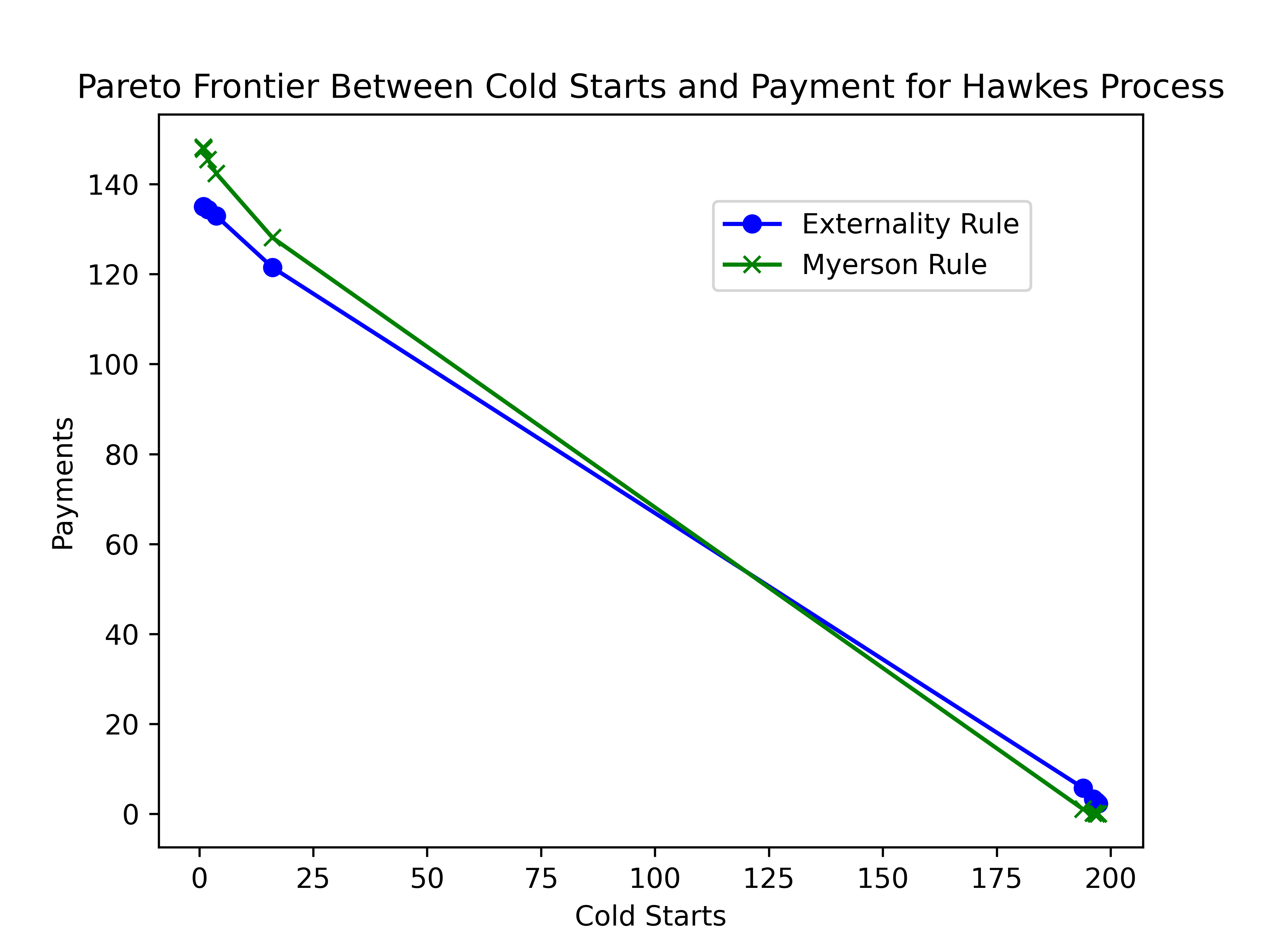}
        \caption{$ \lambda_0 = 0.95, \alpha = 0.88, \beta = 2.04$}
    \end{subfigure}  
    
    \begin{subfigure}[t]{0.44\textwidth}
        \centering
        \includegraphics[height=1.73in]{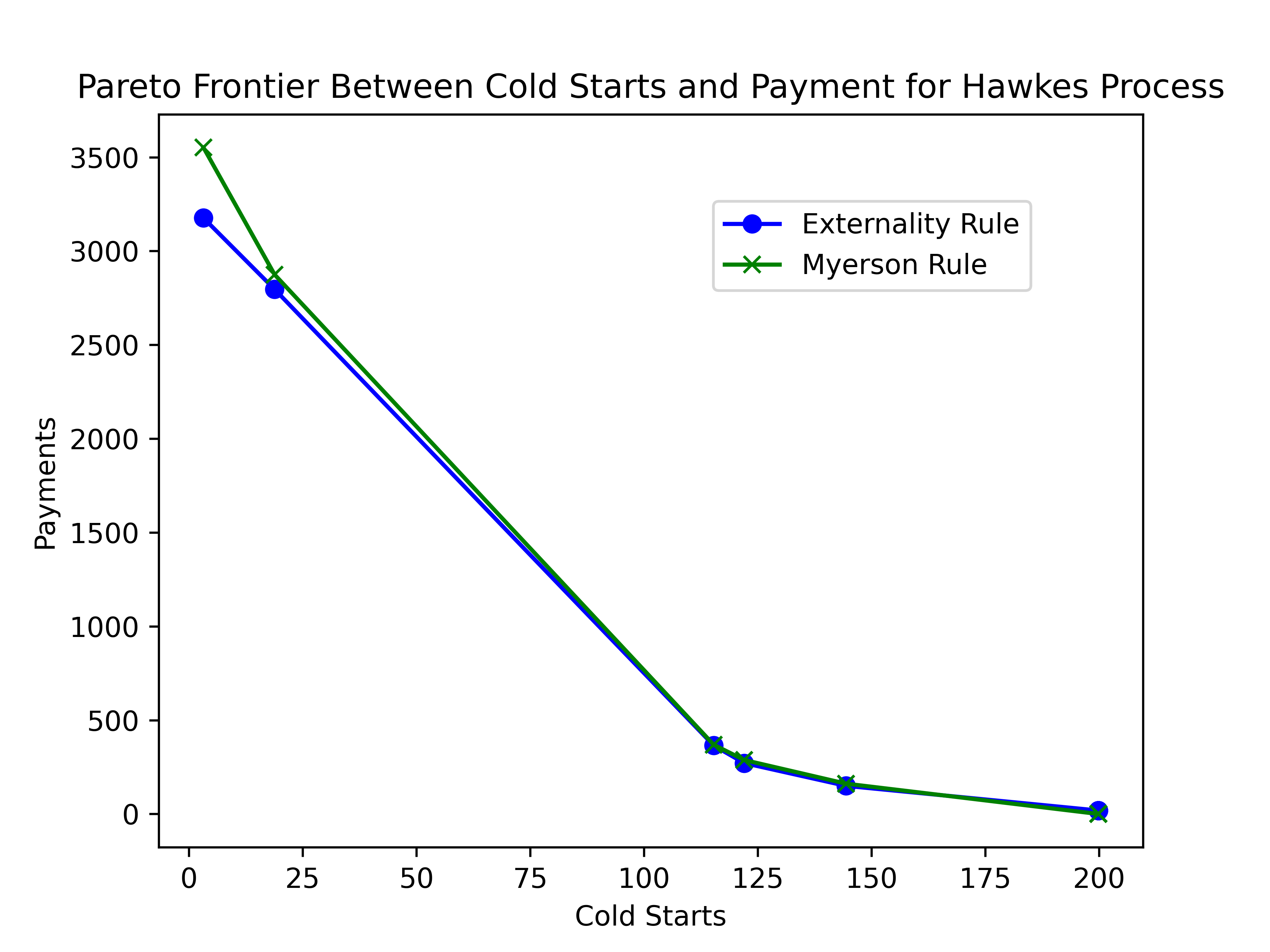}
        \caption{$ \lambda_0 = 0.03, \alpha = 0.56, \beta = 1.25$}
    \end{subfigure}    
    ~ 
    \begin{subfigure}[t]{0.44\textwidth}
        \centering
        \includegraphics[height=1.73in]{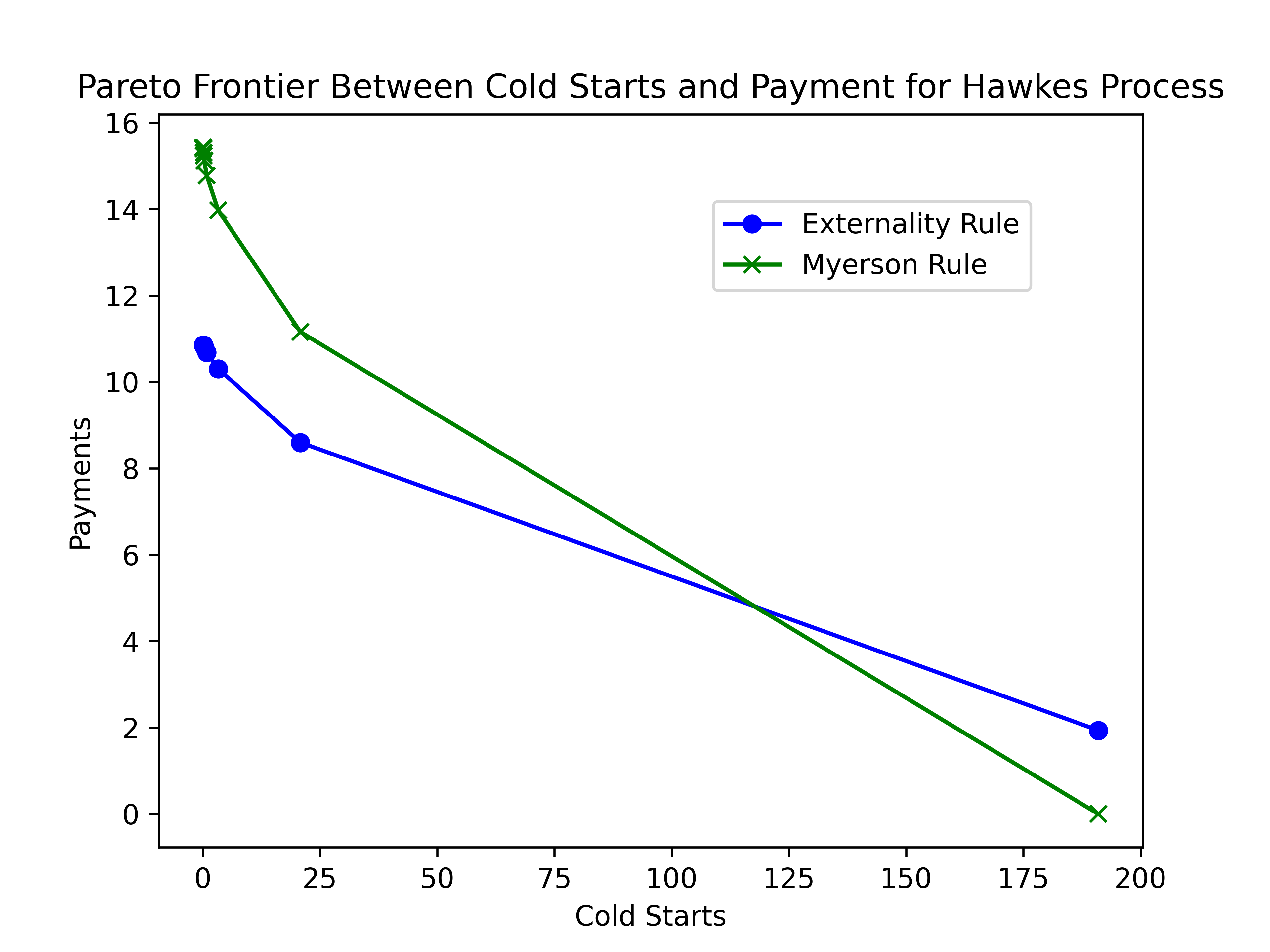}
        \caption{$ \lambda_0 = 0.98, \alpha = 0.88 , \beta = 0.64 $}
    \end{subfigure}   
   \caption{Trade-off curve of cumulative payments vs cold starts for Hawkes process (1-8)}
    \label{hawkes_figure1}
\end{figure*}

\begin{figure*}[t!]
    \centering
    \begin{subfigure}[t]{0.44\textwidth}
        \centering
        \includegraphics[height=1.73in]{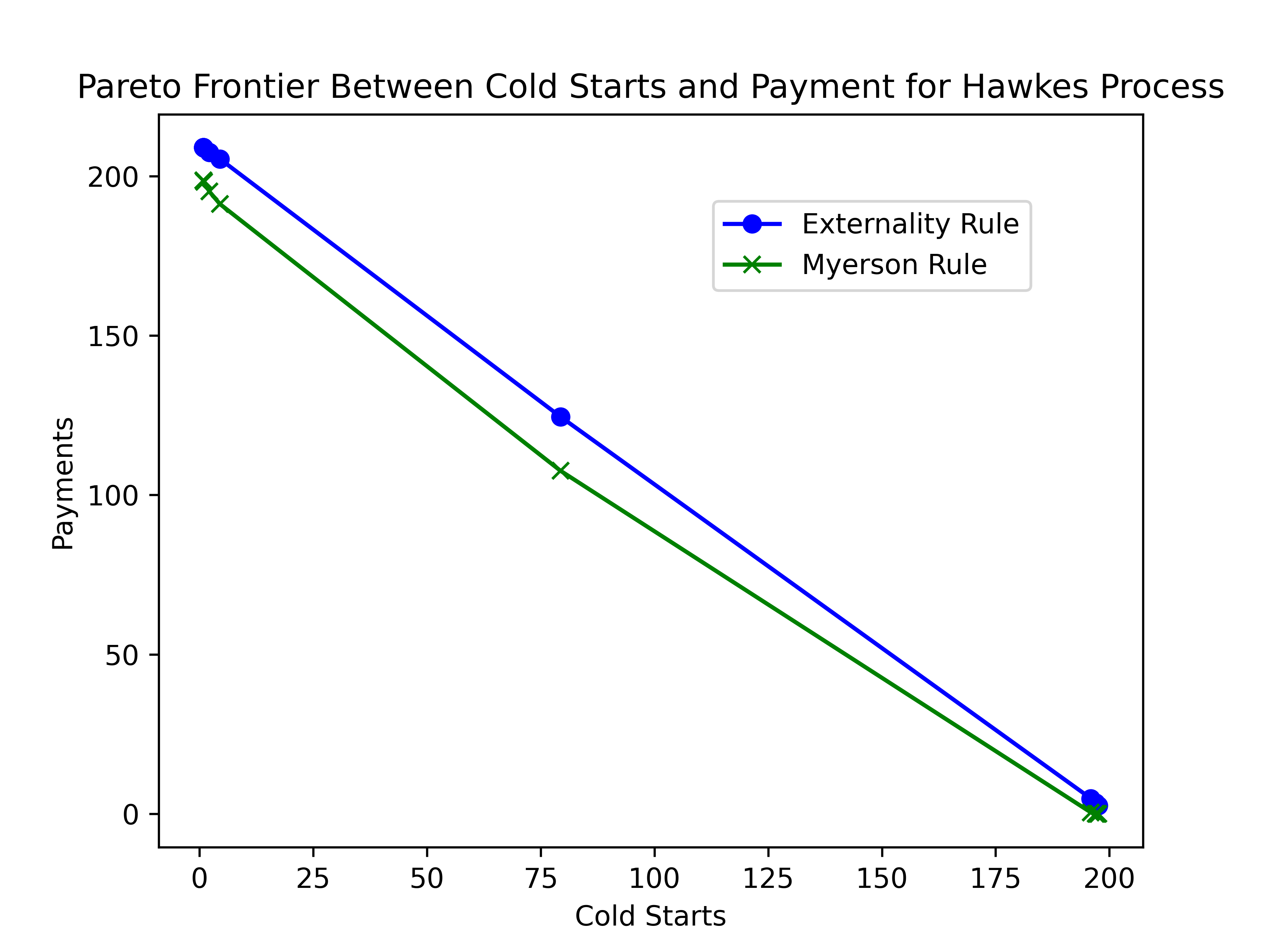}
        \caption{$\lambda_0 = 0.82 , \alpha = 0.97 , \beta = 4.43 $}
    \end{subfigure}%
    ~ 
    \begin{subfigure}[t]{0.44\textwidth}
        \centering
        \includegraphics[height=1.73in]{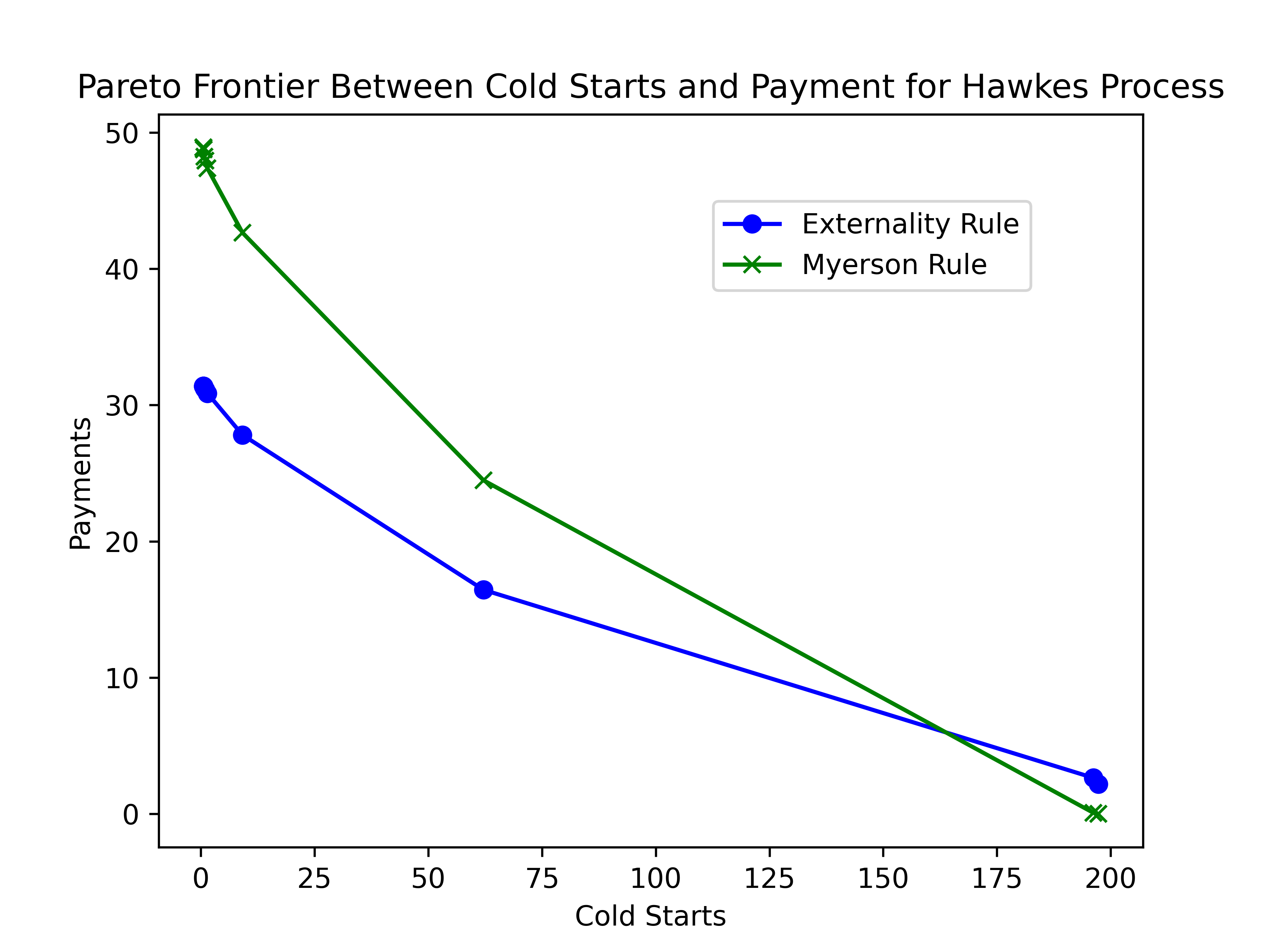}
        \caption{$ \lambda_0 = 0.99, \alpha = 0.40, \beta = 0.46$}
    \end{subfigure}    
     
    \begin{subfigure}[t]{0.44\textwidth}
        \centering
        \includegraphics[height=1.73in]{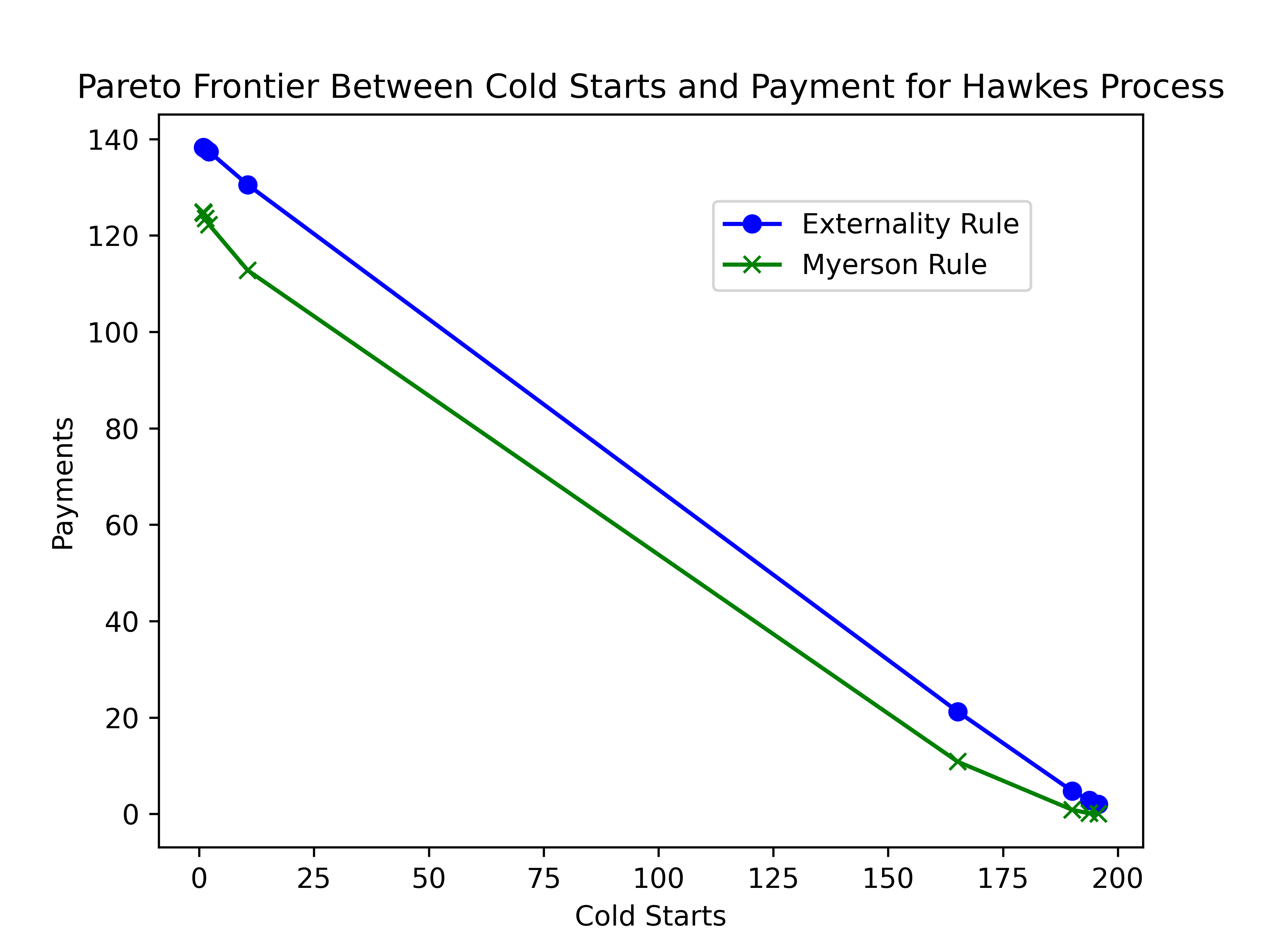}
        \caption{$ \lambda_0 = 0.89 , \alpha = 0.50 , \beta = 1.11$}
    \end{subfigure}        
    ~
    \begin{subfigure}[t]{0.44\textwidth}
        \centering
        \includegraphics[height=1.73in]{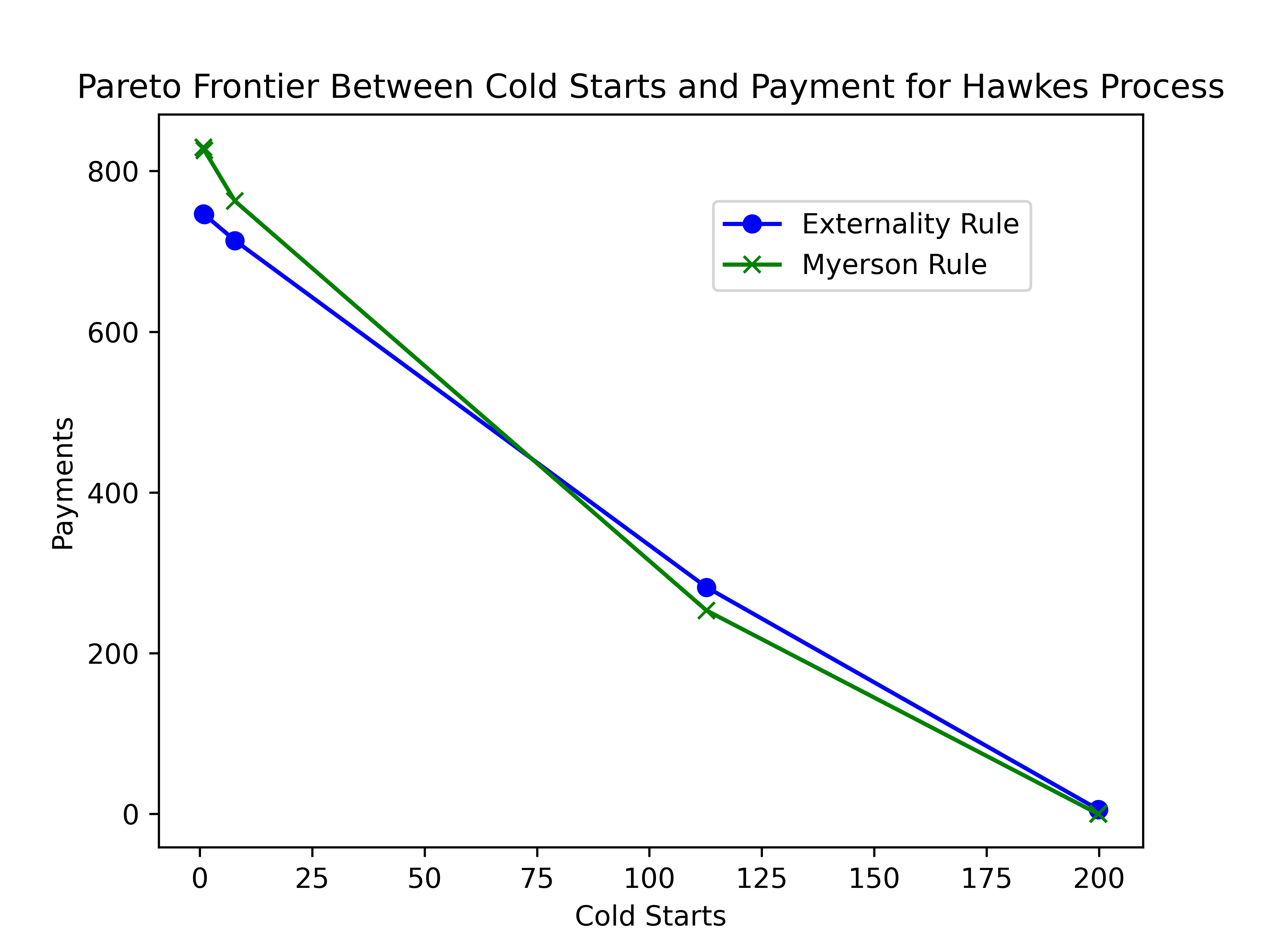}
        \caption{$ \lambda_0 = 0.26 , \alpha = 0.41 , \beta = 3.49$}
    \end{subfigure}    

    \begin{subfigure}[t]{0.44\textwidth}
        \centering
        \includegraphics[height=1.73in]{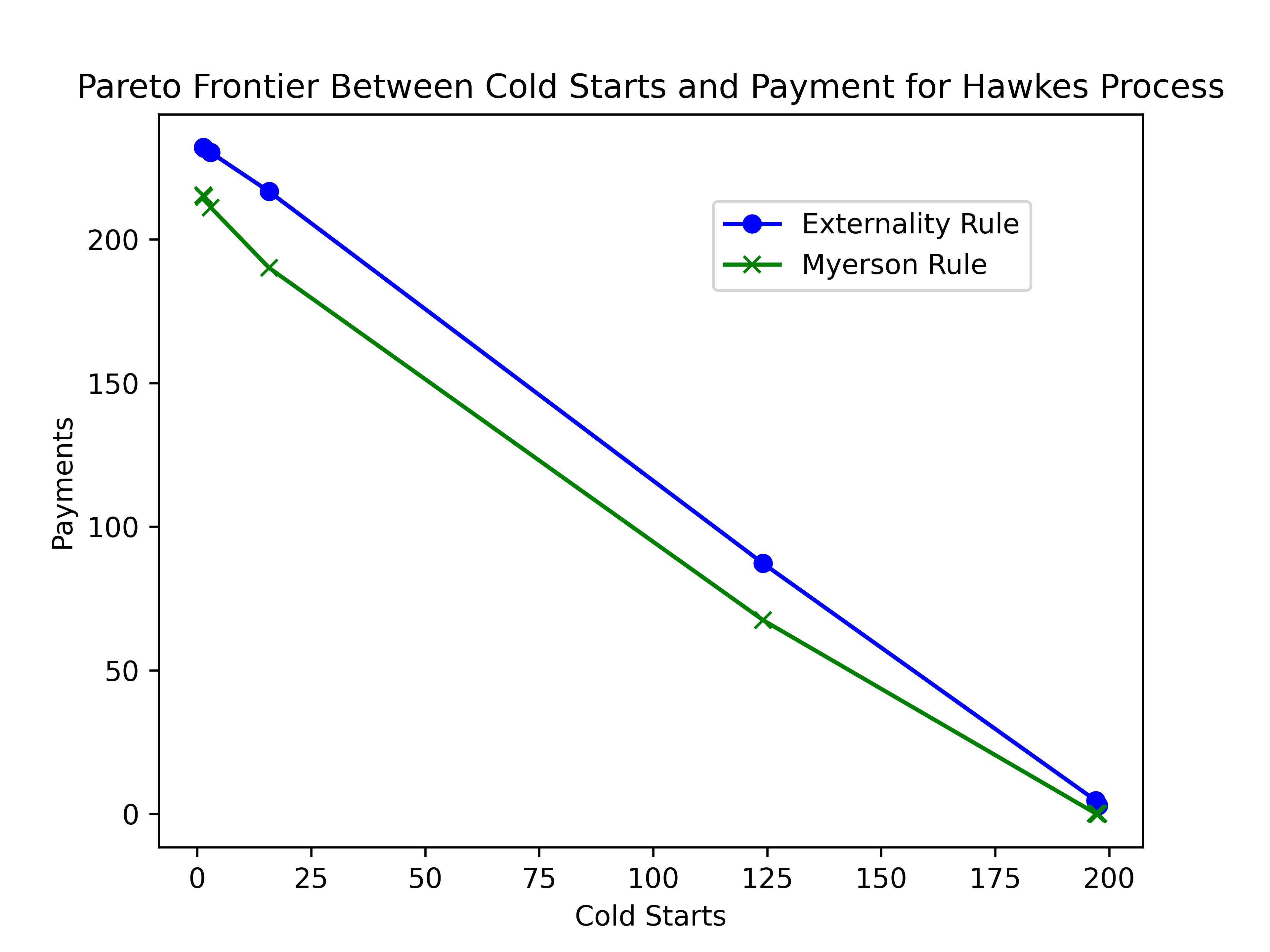}
        \caption{$ \lambda_0 = 0.74, \alpha = 0.88 , \beta = 4.71$}
    \end{subfigure}    
    ~ 
    \begin{subfigure}[t]{0.44\textwidth}
        \centering
        \includegraphics[height=1.73in]{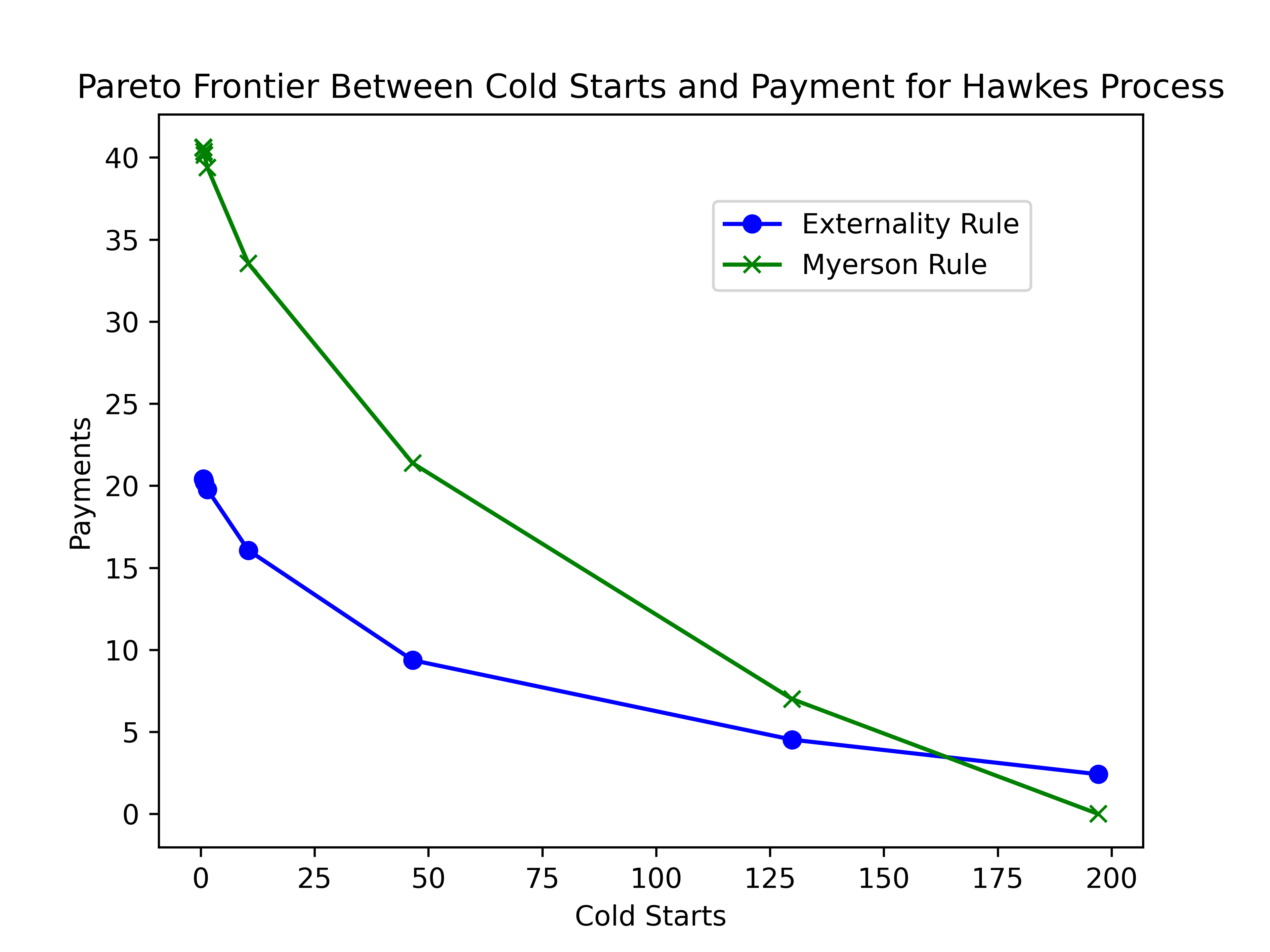}
        \caption{$\lambda_0 = 0.08, \alpha = 0.46, \beta = 0.27 $}
    \end{subfigure}            

    \begin{subfigure}[t]{0.44\textwidth}
        \centering
        \includegraphics[height=1.73in]{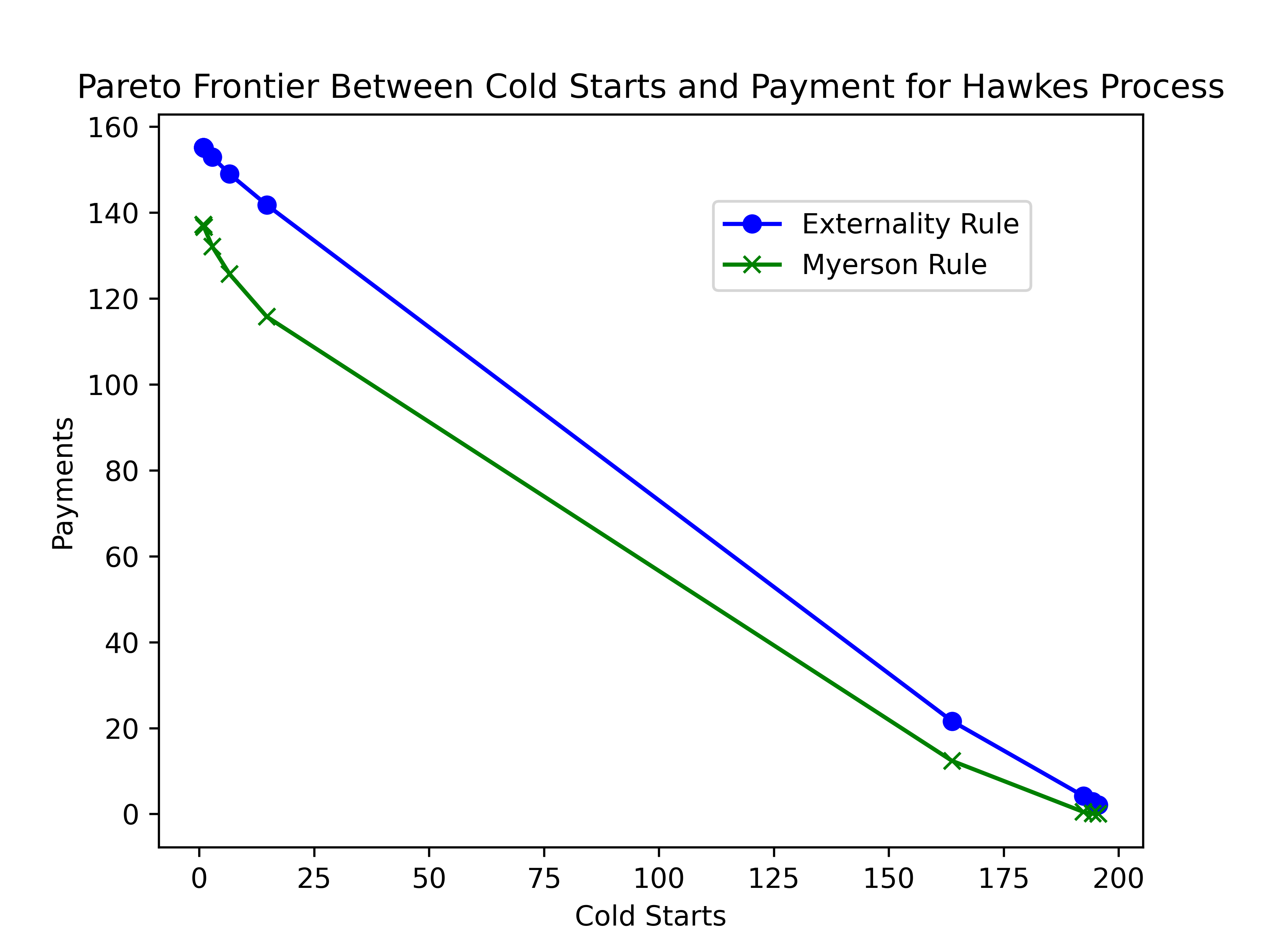}
        \caption{$ \lambda_0 = 0.90, \alpha = 0.86, \beta = 2.43$}
    \end{subfigure}    
    ~ 
    \begin{subfigure}[t]{0.44\textwidth}
        \centering
        \includegraphics[height=1.73in]{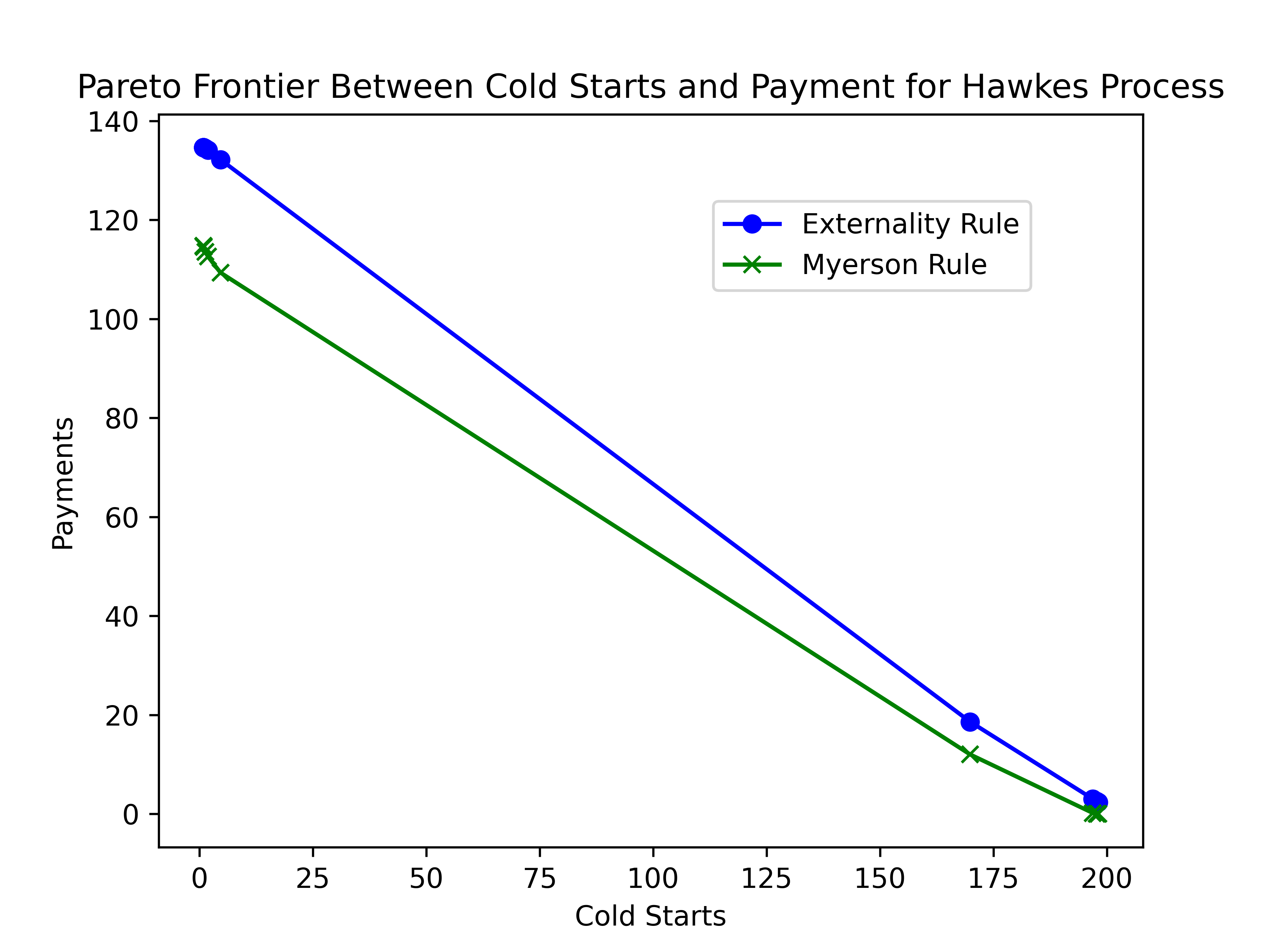}
        \caption{Hawkes  $ \lambda_0 = 0.88, \alpha = 0.55, \beta = 1.55 $}
    \end{subfigure}            

   \caption{Trade-off curve of cumulative payments vs cold starts for Hawkes process (9-16)}
    \label{hawkes_figure2}
\end{figure*}

\begin{figure*}[t!]
    \centering

    \begin{subfigure}[t]{0.44\textwidth}
        \centering
        \includegraphics[height=1.73in]{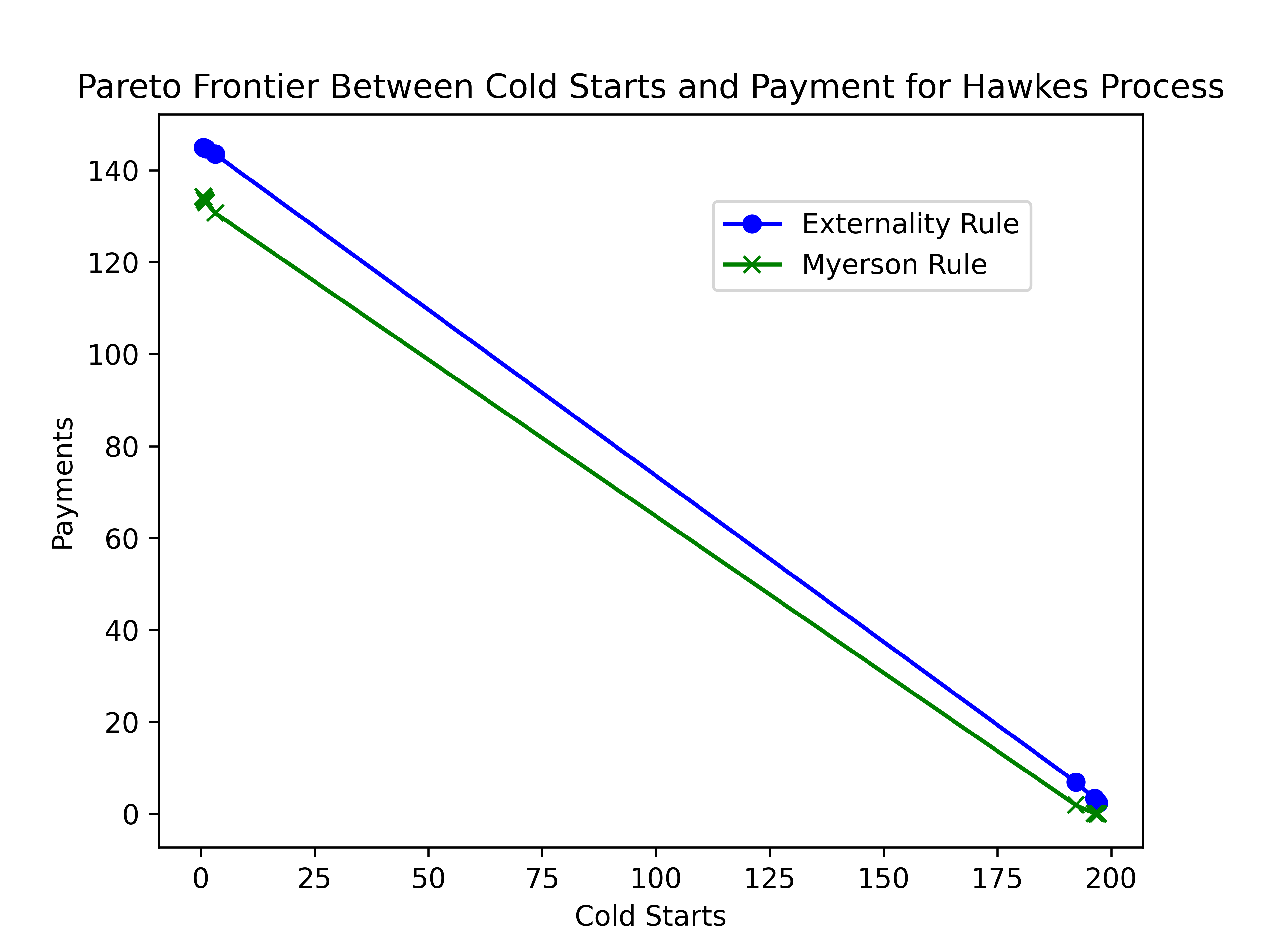}
        \caption{ $ \lambda_0=0.90, \alpha=0.39, \beta = 1.57$}
    \end{subfigure}%
    ~ 
    \begin{subfigure}[t]{0.44\textwidth}
        \centering
        \includegraphics[height=1.73in]{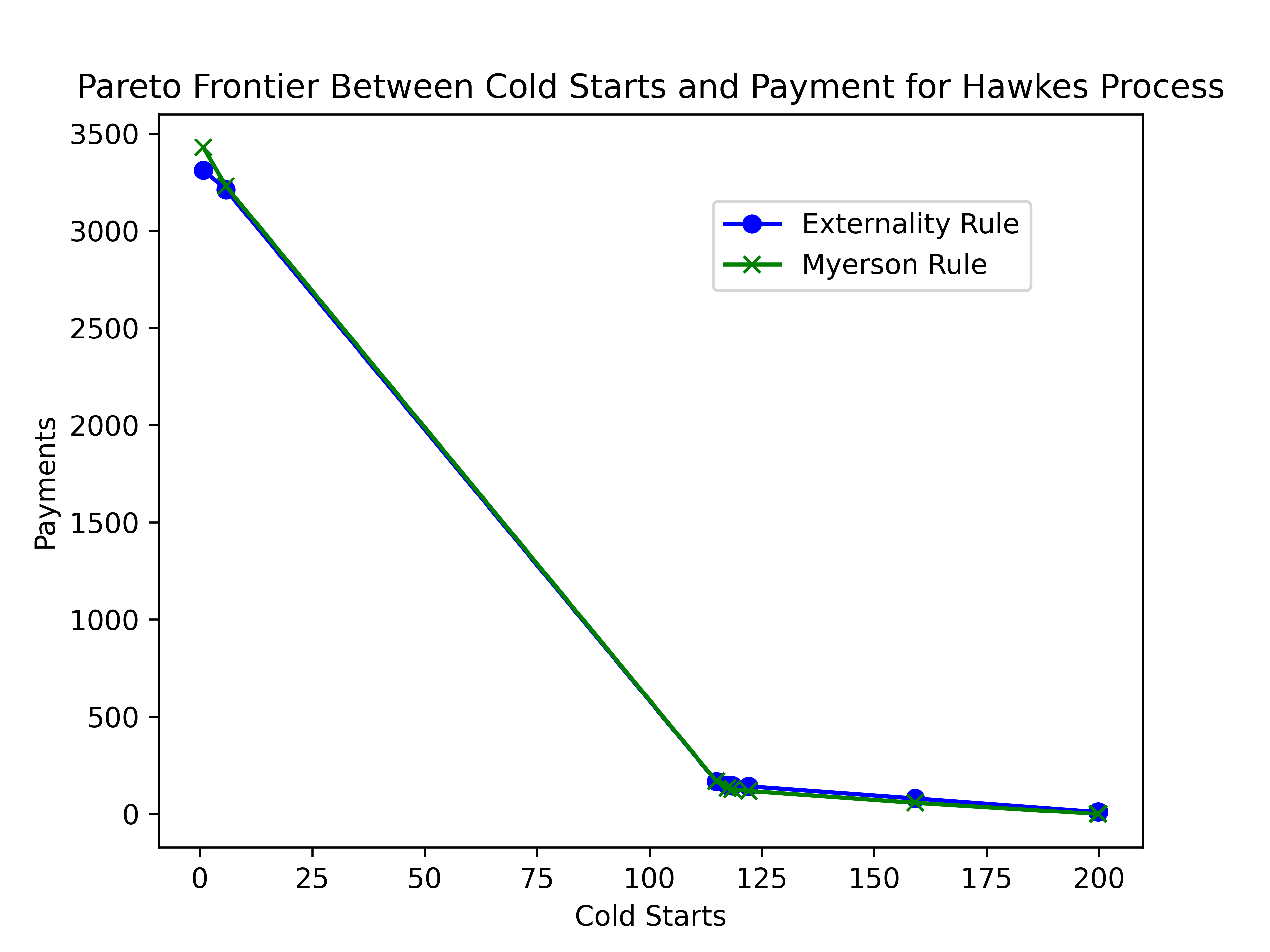}
        \caption{$ \lambda_0 = 0.03 , \alpha = 0.90 , \beta = 2.14$}
    \end{subfigure}   
    
    \begin{subfigure}[t]{0.8\textwidth}
        \centering
        \includegraphics[height=2.03in]{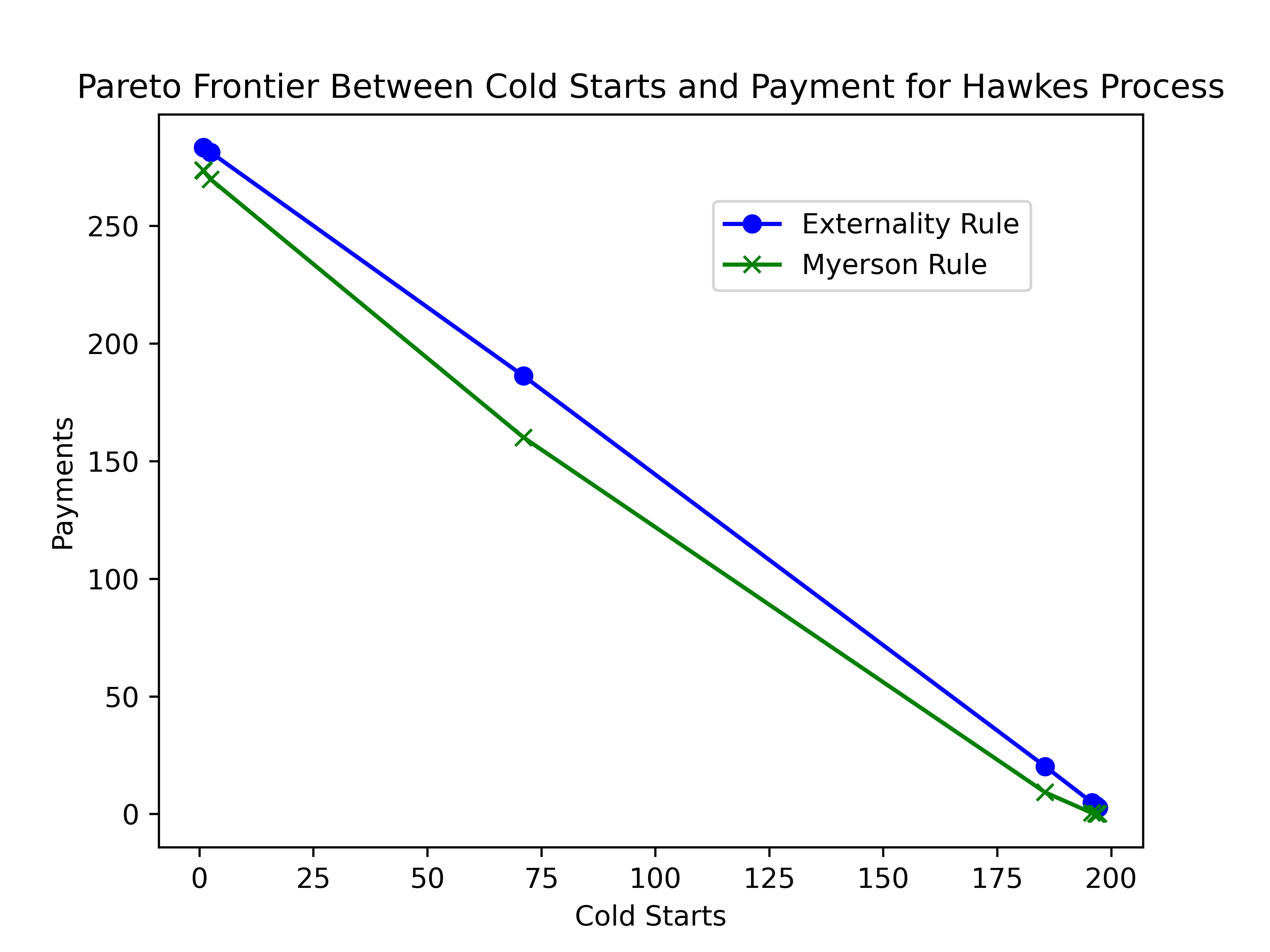}
        \caption{$ \lambda_0 = 0.63, \alpha = 0.42, \beta = 4.06$}
    \end{subfigure}       
   \caption{Trade-off curve of cumulative payments vs cold starts for Hawkes process (17 - 19)}
    \label{hawkes_figure3}
\end{figure*}

\begin{figure*}[t!]
    \centering
    \begin{subfigure}[t]{0.44\textwidth}
        \centering
        \includegraphics[height=1.73in]{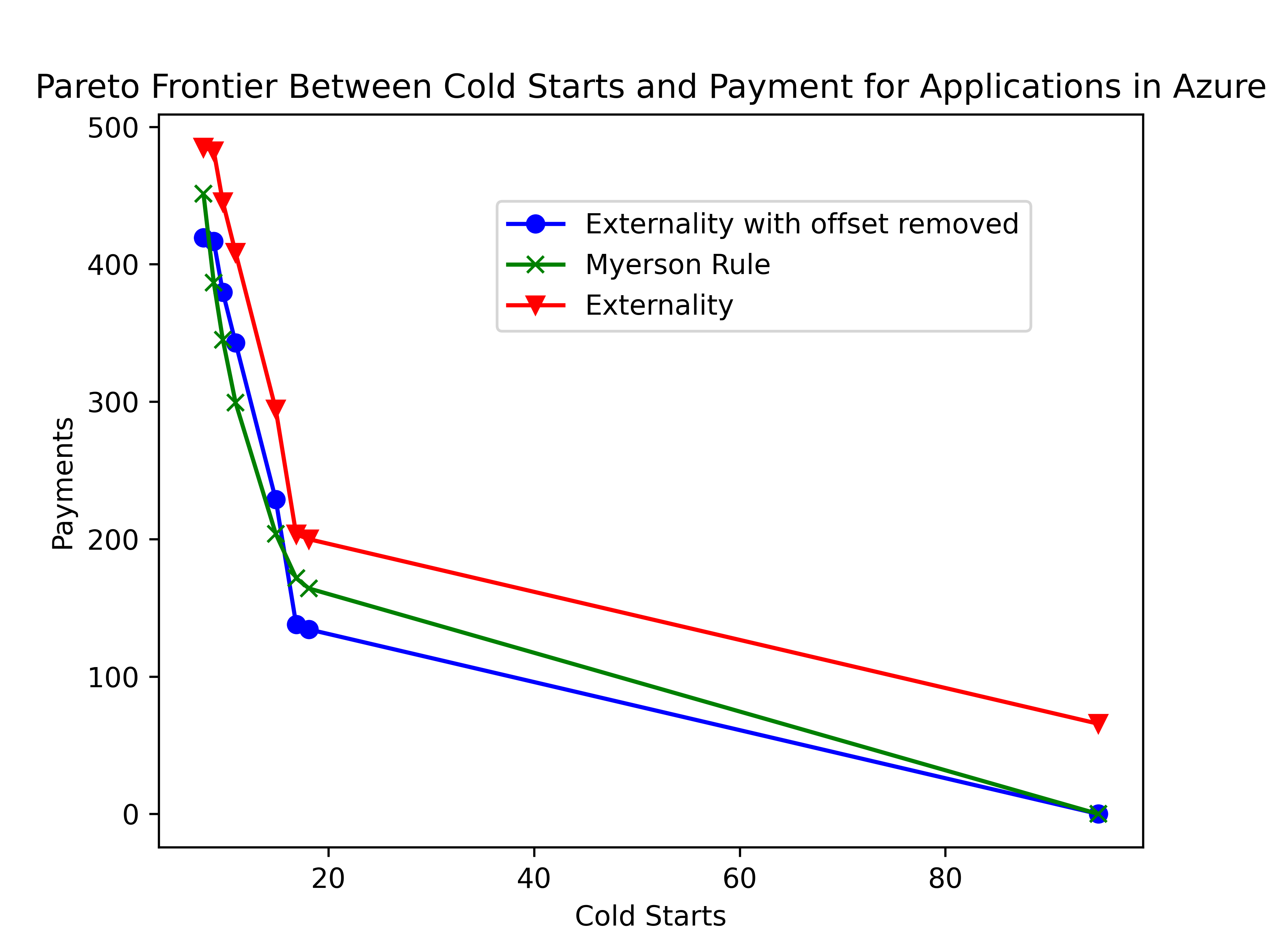}
        \caption{Azure application 3}
    \end{subfigure}%
    ~ 
    \begin{subfigure}[t]{0.44\textwidth}
        \centering
        \includegraphics[height=1.73in]{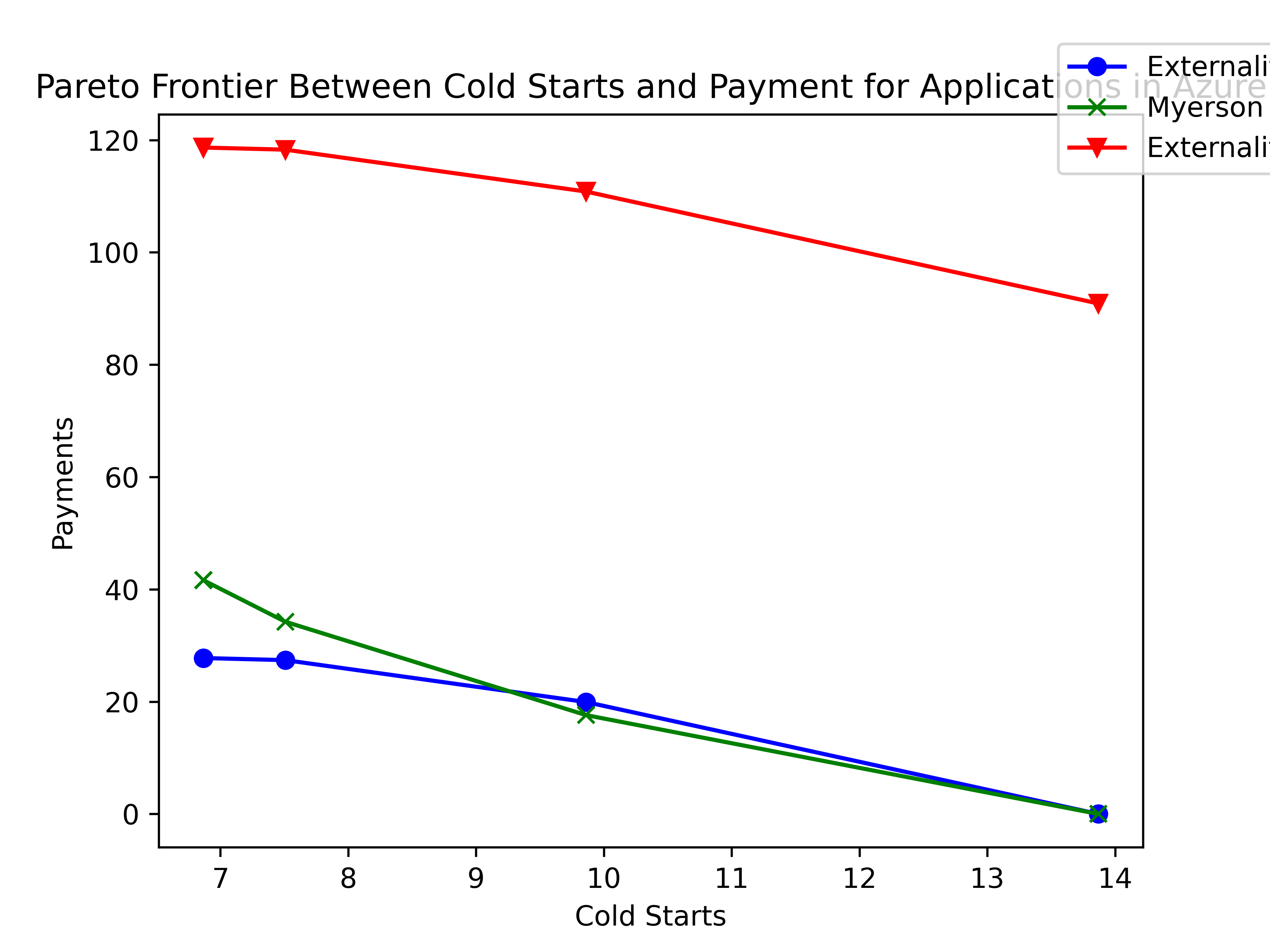}
        \caption{Azure application 4}
    \end{subfigure}    
     
    \begin{subfigure}[t]{0.44\textwidth}
        \centering
        \includegraphics[height=1.73in]{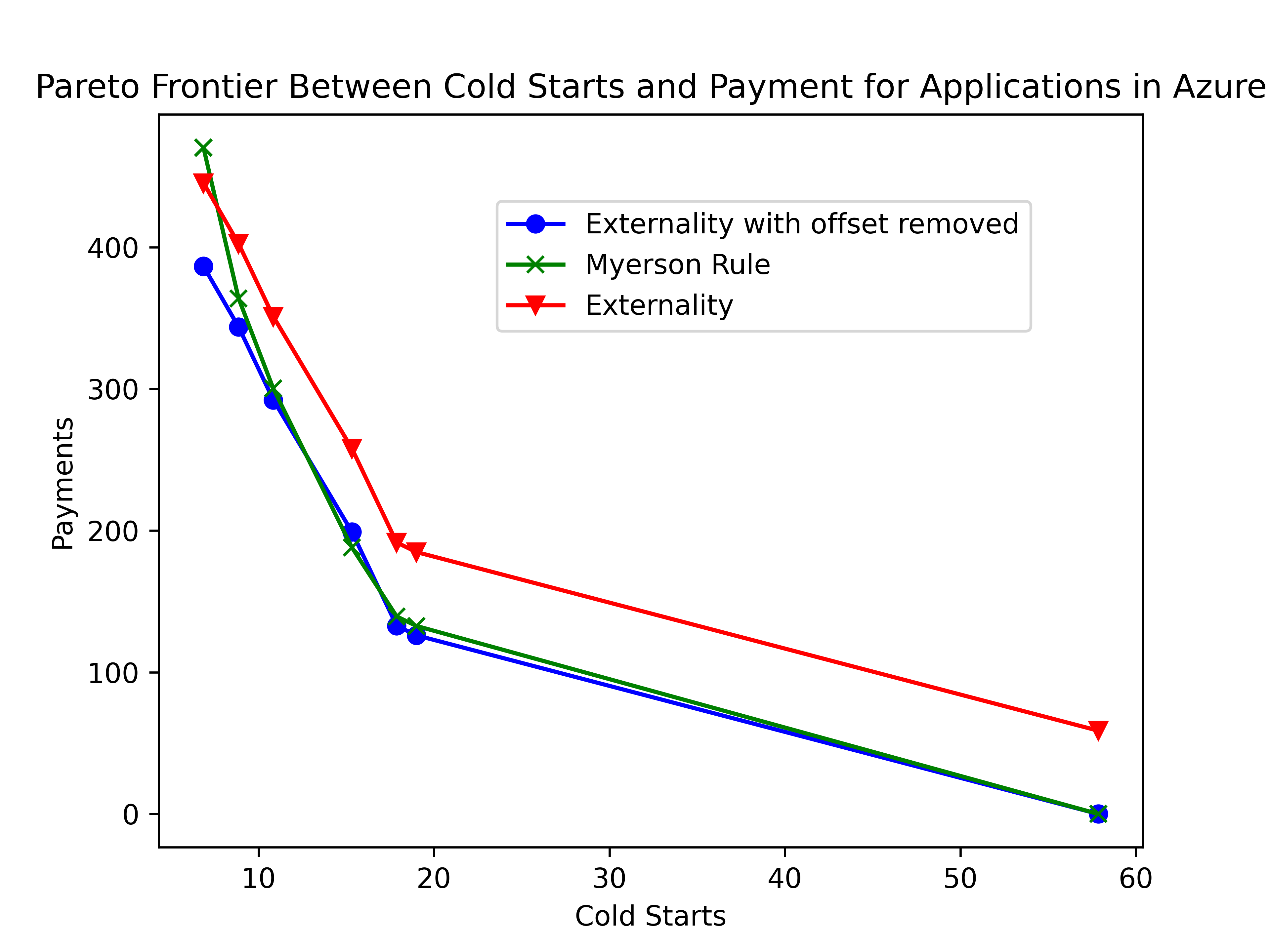}
        \caption{Azure application 5}
    \end{subfigure}        
    ~
    \begin{subfigure}[t]{0.44\textwidth}
        \centering
        \includegraphics[height=1.73in]{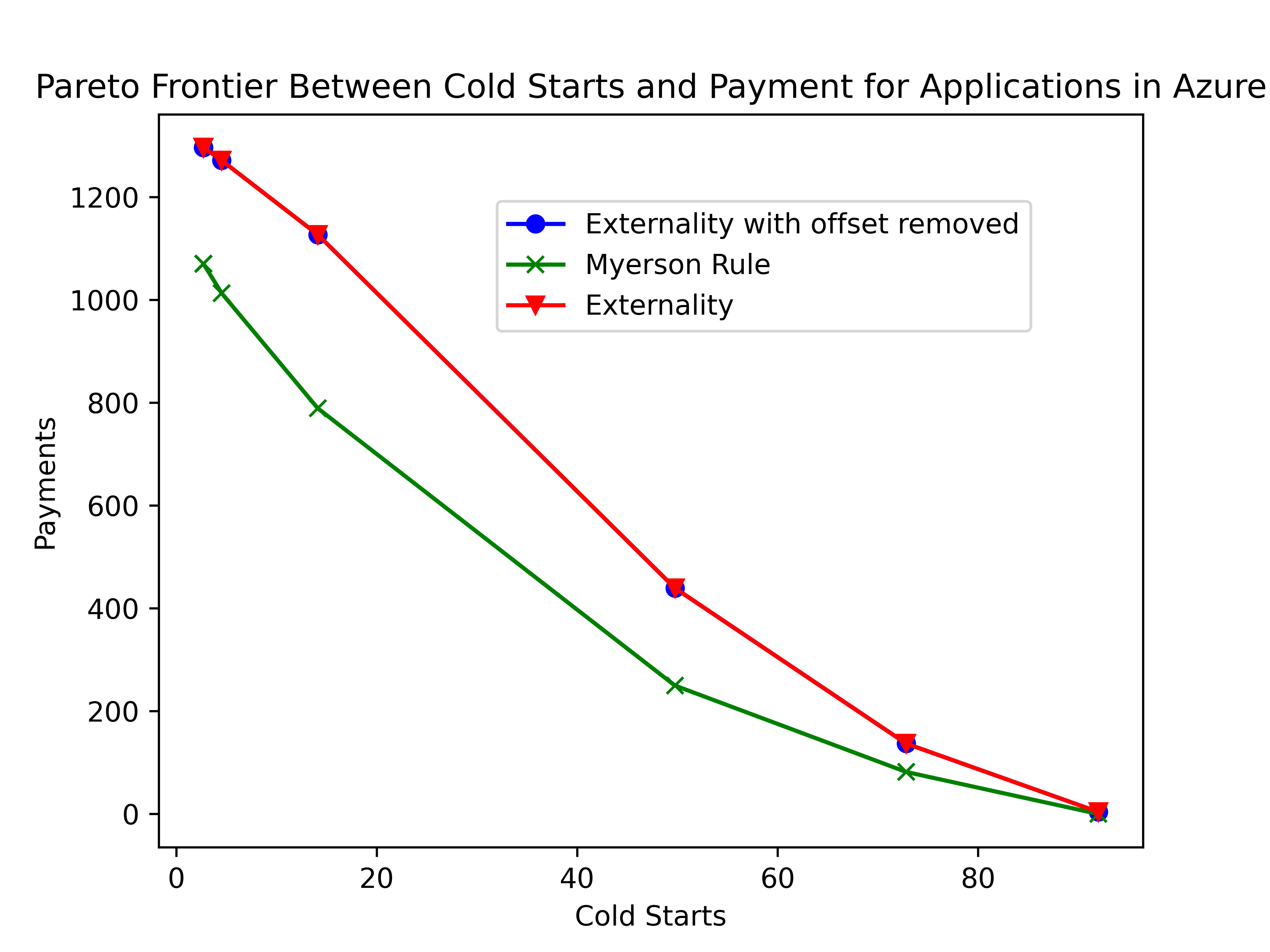}
        \caption{Azure application 6}
    \end{subfigure}    
     
   \caption{Trade-off curve of cumulative payments vs cold starts for Azure applications (3-6)}
    \label{figure_azure_appendix1}
\end{figure*}

\begin{figure*}[t!]
    \centering     
    \begin{subfigure}[t]{0.44\textwidth}
        \centering
        \includegraphics[height=1.73in]{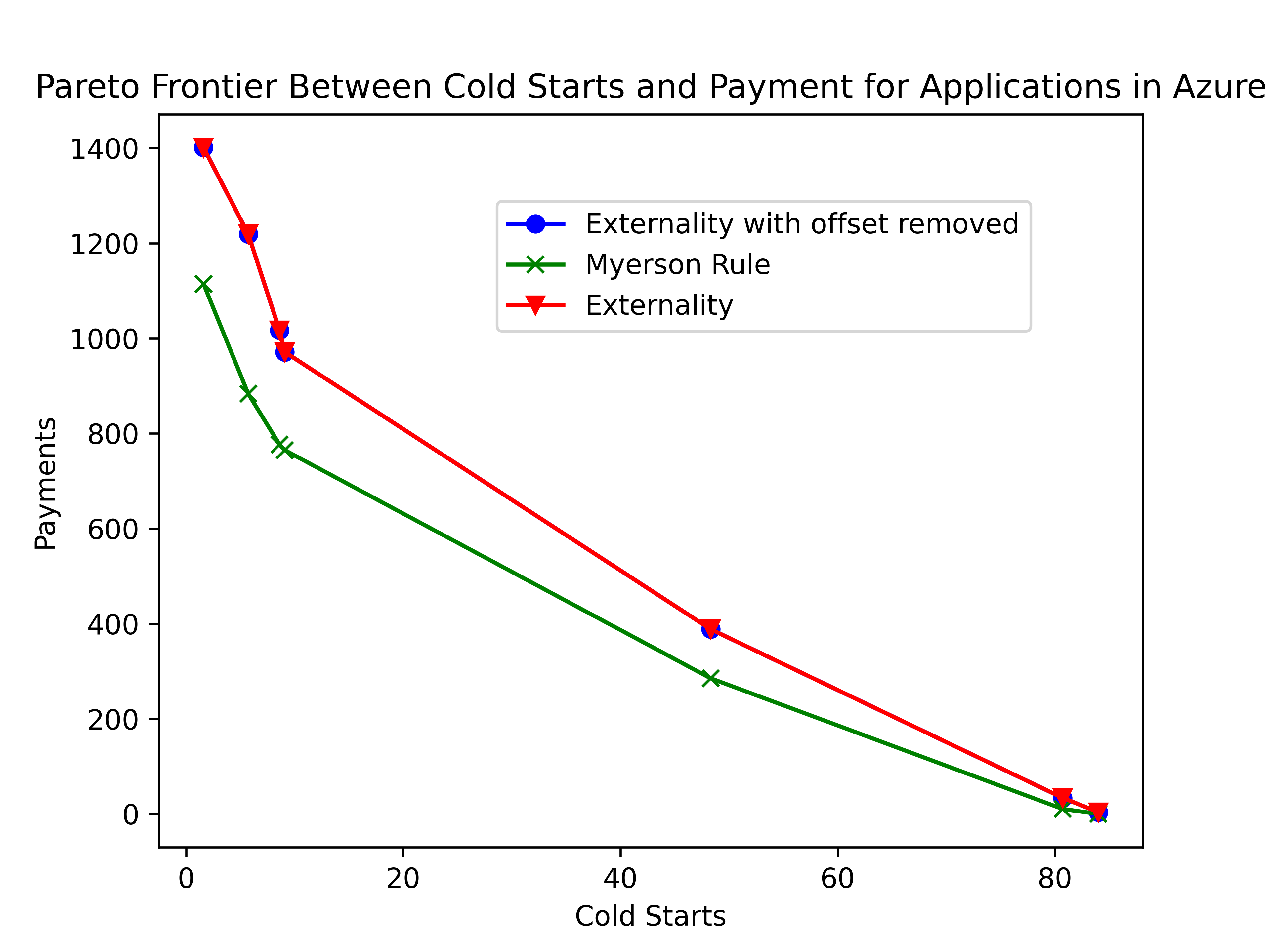}
        \caption{Azure application 7}
    \end{subfigure}    
    ~ 
    \begin{subfigure}[t]{0.44\textwidth}
        \centering
        \includegraphics[height=1.73in]{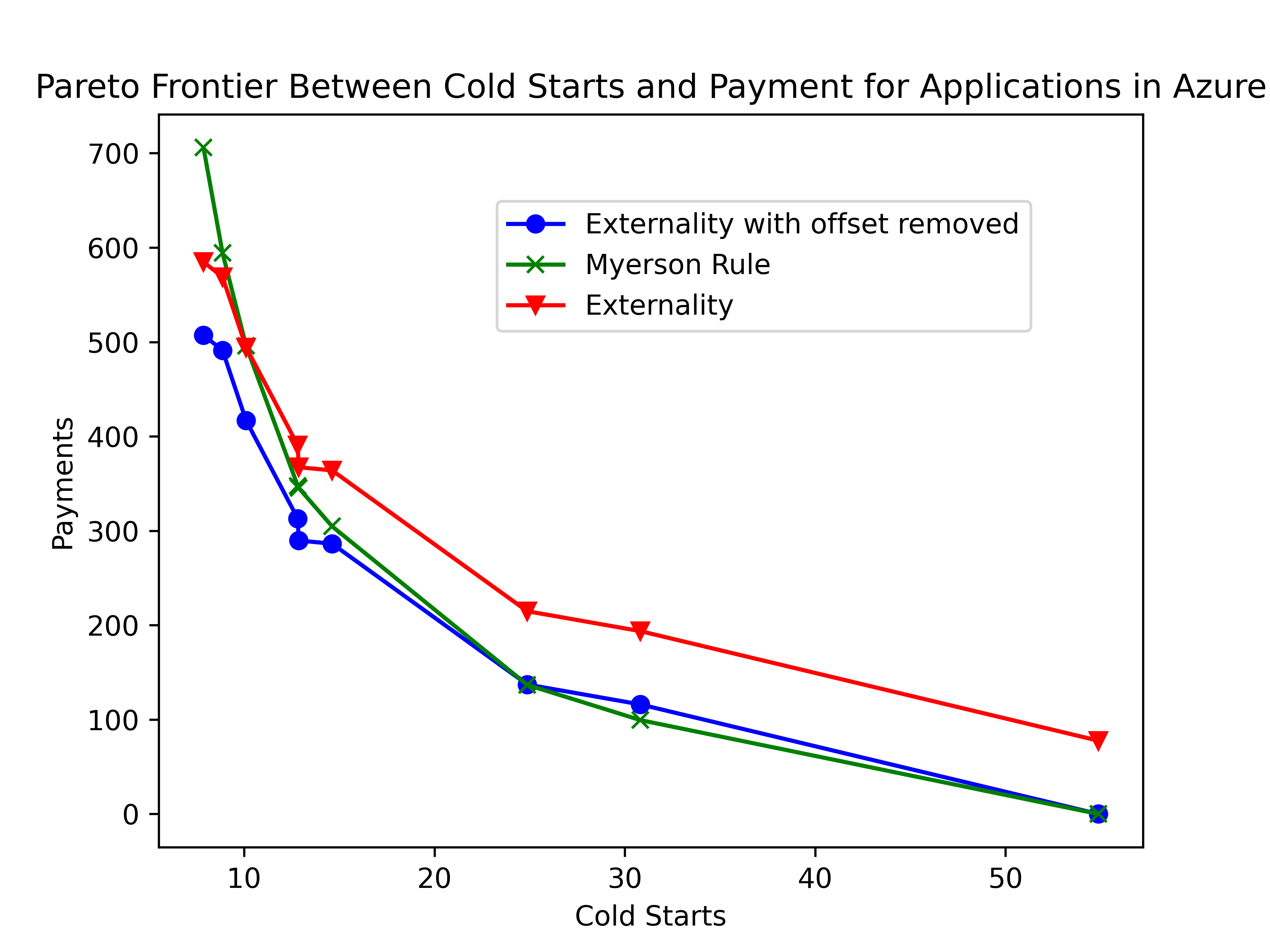}
        \caption{ Azure application 8}
    \end{subfigure}            

    \begin{subfigure}[t]{0.44\textwidth}
        \centering
        \includegraphics[height=1.73in]{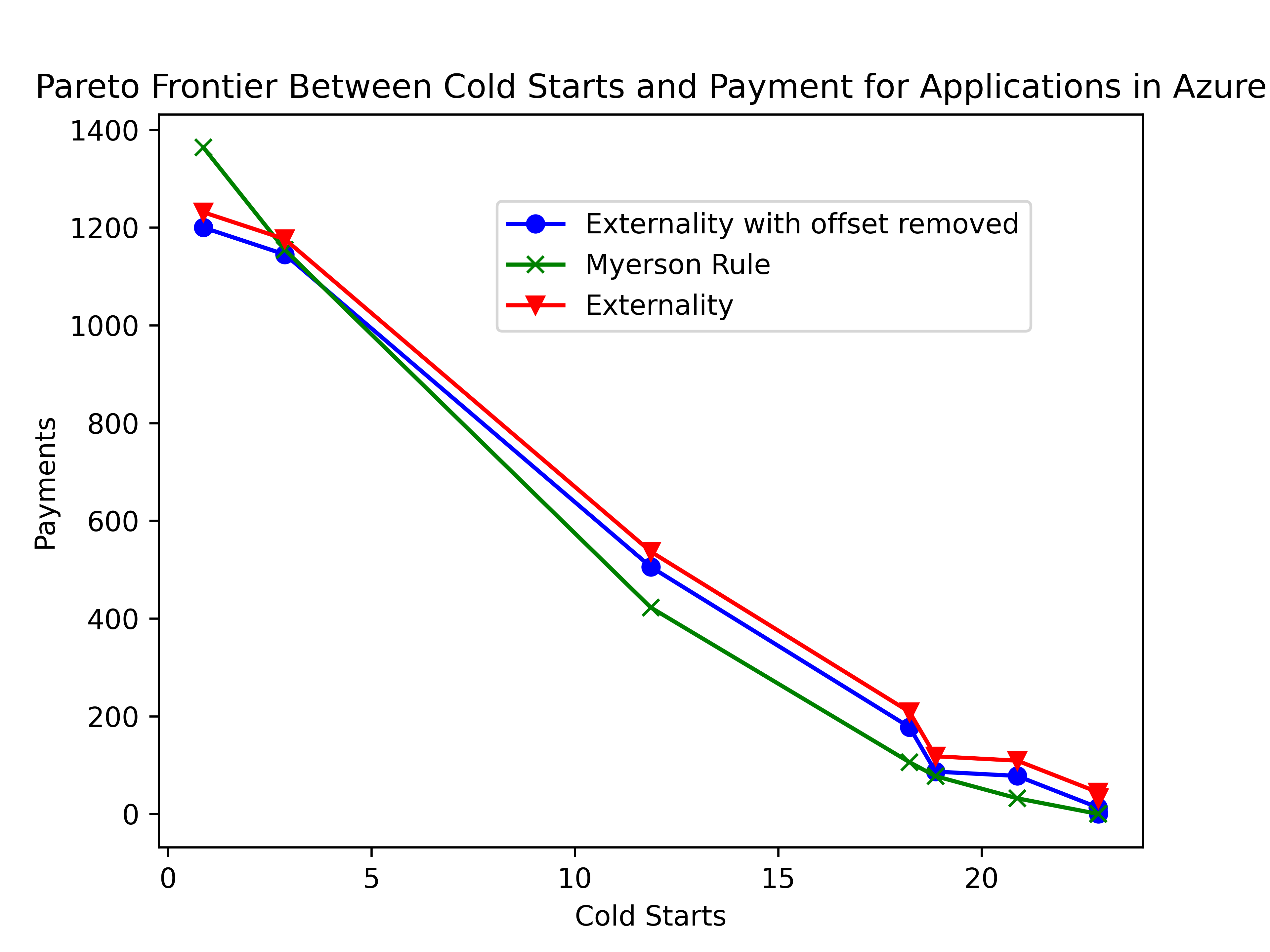}
        \caption{Azure application 9}
    \end{subfigure}    
    ~ 
    \begin{subfigure}[t]{0.44\textwidth}
        \centering
        \includegraphics[height=1.73in]{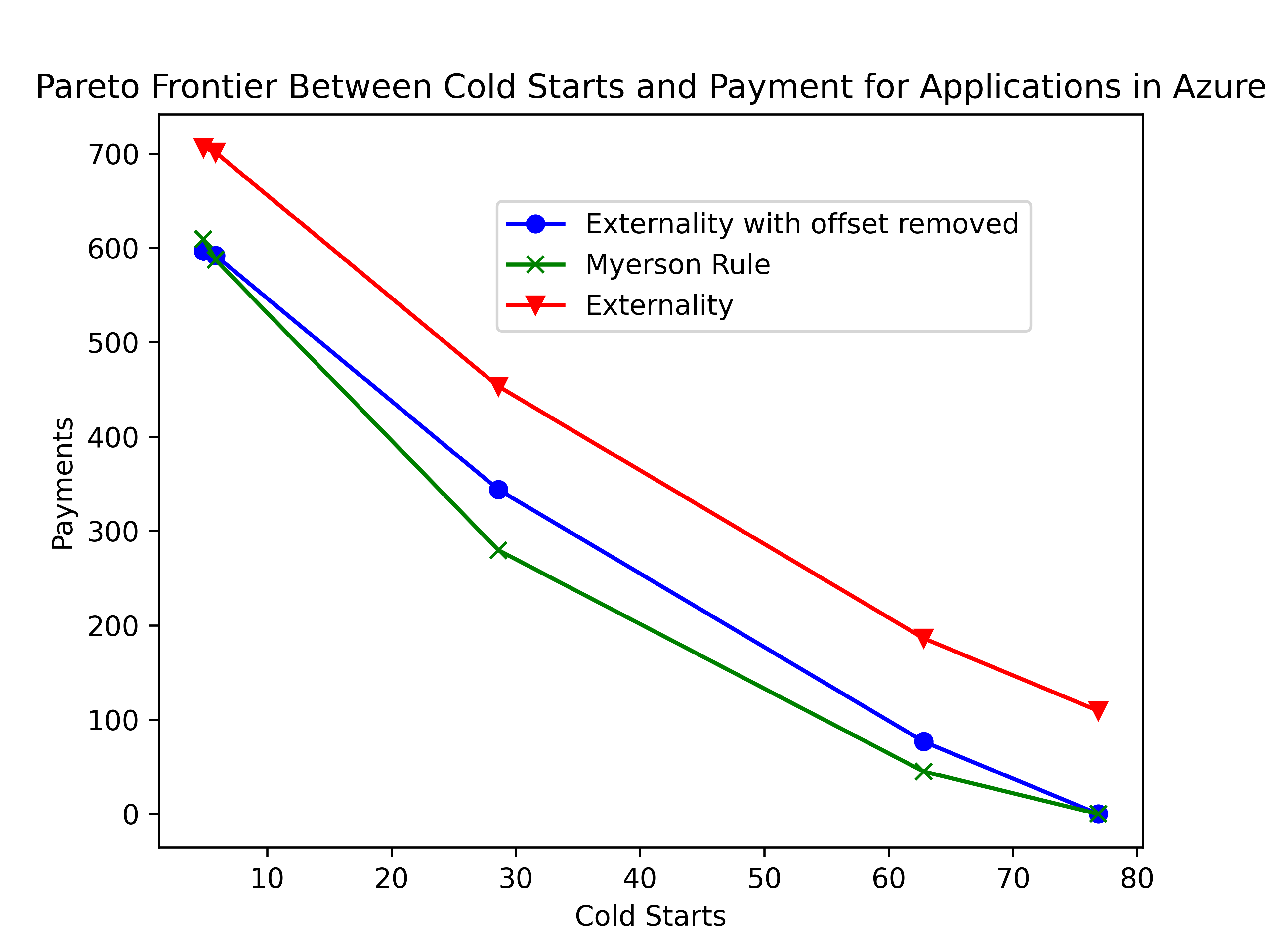}
        \caption{Azure application 10}
    \end{subfigure}            
   \caption{Trade-off curve of cumulative payments vs cold starts for Azure applications (7-10)}
    \label{figure_azure_appendix2}
\end{figure*}

\section{Code}

The experimental simulations were conducted on an Intel® CoreTM i5-8500U CPU with four cores operating at a base frequency of 1.60 GHz on a 12 GB RAM system. All the code was programmed in Python version 3.12, and executed on a Linux-based operating system.  

\end{document}